\documentclass[12pt,english,sort&compress]{DePaulThesis}
\usepackage{ae,aecompl}
\usepackage[T1]{fontenc}
\usepackage[latin9]{inputenc}
\usepackage{color}
\usepackage{array}
\usepackage{verbatim}
\usepackage{varioref}
\usepackage{refstyle}
\usepackage{float}
\usepackage{booktabs}
\usepackage{units}
\usepackage{textcomp}
\usepackage{multirow}
\usepackage{amsthm}
\usepackage{amsmath}
\usepackage{amssymb}
\usepackage{graphicx}
\usepackage{setspace}
\usepackage[numbers]{natbib}
\PassOptionsToPackage{normalem}{ulem}
\usepackage{ulem}
\usepackage{nameref}

\makeatletter

\AtBeginDocument{\providecommand\secref[1]{\ref{sec:#1}}}
\AtBeginDocument{\providecommand\tabref[1]{\ref{tab:#1}}}
\AtBeginDocument{\providecommand\propref[1]{\ref{prop:#1}}}
\AtBeginDocument{\providecommand\algref[1]{\ref{alg:#1}}}
\AtBeginDocument{\providecommand\conjref[1]{\ref{conj:#1}}}
\AtBeginDocument{\providecommand\figref[1]{\ref{fig:#1}}}
\AtBeginDocument{\providecommand\apnref[1]{\ref{apn:#1}}}
\AtBeginDocument{\providecommand\subbref[1]{\ref{subb:#1}}}
\AtBeginDocument{\providecommand\proref[1]{\ref{pro:#1}}}
\AtBeginDocument{\providecommand\lemref[1]{\ref{lem:#1}}}
\AtBeginDocument{\providecommand\thmref[1]{\ref{thm:#1}}}
\providecommand{\tabularnewline}{\\}
\RS@ifundefined{subref}
  {\def\RSsubtxt{section~}\newref{sub}{name = \RSsubtxt}}
  {}
\RS@ifundefined{thmref}
  {\def\RSthmtxt{theorem~}\newref{thm}{name = \RSthmtxt}}
  {}
\RS@ifundefined{lemref}
  {\def\RSlemtxt{lemma~}\newref{lem}{name = \RSlemtxt}}
  {}

  \theoremstyle{definition}
  \newtheorem{defn}{\protect\definitionname}[section]
  \theoremstyle{plain}
  \newtheorem{prop}{\protect\propositionname}[section]
 \ifx\proof\undefined\
   \newenvironment{proof}[1][\proofname]{\par
     \normalfont\topsep6\p@\@plus6\p@\relax
     \trivlist
     \itemindent\parindent
     \item[\hskip\labelsep
           \scshape
       #1]\ignorespaces
   }{%
     \endtrivlist\@endpefalse
   }
   \providecommand{\proofname}{Proof}
 \fi
  \theoremstyle{plain}
  \newtheorem{conjecture}{\protect\conjecturename}[section]
  \theoremstyle{remark}
  \newtheorem{rem}{\protect\remarkname}[section]
\newenvironment{lyxlist}[1]
{\begin{list}{}
{\settowidth{\labelwidth}{#1}
 \setlength{\leftmargin}{\labelwidth}
 \addtolength{\leftmargin}{\labelsep}
 }}
{\end{list}}
  \theoremstyle{plain}
  \newtheorem{thm}{\protect\theoremname}[section]
  \theoremstyle{plain}
  \newtheorem{lem}{\protect\lemmaname}[section]
  \theoremstyle{plain}
  \newtheorem{cor}{\protect\corollaryname}[section]
 \newenvironment{codesmall}
     {\par\begin{scriptsize}\begin{list}{}{
         \setlength{\rightmargin}{\leftmargin}
         \setlength{\listparindent}{0pt}
         \raggedright
         \setlength{\itemsep}{0pt}
         \setlength{\parsep}{0pt}
         \normalfont\ttfamily}%
         \item[]}
     {\end{list}\end{scriptsize}}

\college{College of Computing and Digital Media}
\university{DePaul University}
\degree{Master of Science in Computer Science}

\authors{Jason Spencer}
\subject{Enhanced cellular constructions for building cryptographic pseudorandom generators}
\keywords{pseudorandom generator, cellular automata, cryptography, provable security}

\usepackage{microtype}

\algoption{noend}
\keycomment{\{\{}{\}\}}

\usepackage{marvosym}

\newcommand \Nat{\mathbb{N}}

\@ifundefined{showcaptionsetup}{}{%
 \PassOptionsToPackage{caption=false}{subfig}}
\usepackage{subfig}
\makeatother

\usepackage{babel}
\providecommand{\conjecturename}{Conjecture}
\providecommand{\corollaryname}{Corollary}
\providecommand{\definitionname}{Definition}
\providecommand{\lemmaname}{Lemma}
\providecommand{\propositionname}{Proposition}
\providecommand{\remarkname}{Remark}
\providecommand{\theoremname}{Theorem}

\begin{document}

\title{Cellular Automata in Cryptographic Random Generators}

\author{Jason Spencer}

\date{May 1, 2013}

\maketitle
\begin{romanpages}

\begin{abstract}
Cryptographic schemes using one-dimensional, three-neighbor cellular
automata as a primitive have been put forth since at least 1985. Early
results showed good statistical pseudorandomness, and the simplicity
of their construction made them a natural candidate for use in cryptographic
applications. Since those early days of cellular automata, research
in the field of cryptography has developed a set of tools which allow
designers to prove a particular scheme to be as hard as solving an
instance of a well-studied problem, suggesting a level of security
for the scheme. However, little or no literature is available on whether
these cellular automata can be proved secure under even generous assumptions.
In fact, much of the literature falls short of providing complete,
testable schemes to allow such an analysis. 

In this thesis, we first examine the suitability of cellular automata
as a primitive for building cryptographic primitives. In this effort,
we focus on pseudorandom bit generation and noninvertibility, the
behavioral heart of cryptography. In particular, we focus on cyclic
linear and non-linear automata in some of the common configurations
to be found in the literature. We examine known attacks against these
constructions and, in some cases, improve the results.

Finding little evidence of provable security, we then examine whether
the desirable properties of cellular automata (i.e. highly parallel,
simple construction) can be maintained as the automata are enhanced
to provide a foundation for such proofs. This investigation leads
us to a new construction of a finite state cellular automaton (FSCA)
which is NP-Hard to invert. Finally, we introduce the Chasm pseudorandom
generator family built on this construction and provide some initial
experimental results using the NIST test suite.\end{abstract}
\begin{acknowledgements}
Thanks to all my friends and family for excusing the many absences
over the years.\medskip{}

Thanks to my colleagues at Volcano for listening to me drone on about
this project for nigh on two years.\medskip{}

Thanks Dr. Phillip Rogaway at UC Davis for getting me started in provable
security and inspiring this line of inquiry.\medskip{}

Special thanks to Dr. Chris Peikert at Georgia Tech for engaging in
many useful exchanges with a determined stranger.\medskip{}

Most of all, thanks to Dr. Marcus Schaefer for his useful suggestions,
detailed feedback, willingness to go where no one else would, and
gentle hand in dealing with a stubborn narcissist.
\end{acknowledgements}

\tableofcontents{}

\begin{center}
\vfill{}
Copyright © 2013 Jason Spencer. All rights reserved.\listoffigures

\par\end{center}

\listoftables

\listofalgorithms

\begin{dedication}
This work is dedicated to my girls, without whose sacrifice it would
not have been possible.\end{dedication}
\end{romanpages}

\section{Introduction}
\begin{quotation}
\begin{flushright}
\medskip{}
\textsl{\small \textquotedbl{}Any one who considers arithmetical methods
of producing random digits is, of course, in a state of sin.\textquotedbl{}--John
von Neumann.}
\par\end{flushright}{\small \par}
\end{quotation}
\smallskip{}
The modern world increasingly hinges on communication. Business is
ruled more and more by e-Commerce. Our computers run largely on downloaded
open-source software. Students turn in their distance-learning homework
through e-mail and websites. Nations attack each other's domestic
infrastructure over the internet. Notice that at least one party in
each of these scenarios has an interest in ensuring the secrecy and/or
the authenticity of the communication. In an age where information
itself becomes a prime mover, protecting that information becomes
more important. So it is seems fair to say that \emph{secure} communication
is increasingly critical in our daily life. And though security requires
an array of solutions to many challenging problems, certainly good
cryptography is one cornerstone.

The foundation of much of modern security and cryptography is good
random number generation. Pseudorandom generators (PRGs) are used
for exchanging session keys, creating public/private key pairs, creating
symmetric keys from user input, and generating nonces and initialization
vectors for various modes of encryption. Good PRGs are also at the
heart of many cryptographic primitives such as stream ciphers, block
ciphers, and hash functions. PRGs used in this setting must be carefully
designed lest they compromise the entire cryptosystem. Even the best
crypto primitives become useless when operated with poorly generated
pseudorandom data. We mention the cases of Netscape \cite{Goldberg96},
Kerberos \cite{Minar96}, the GSM wireless network \cite{Biryukov00},
and the Sony PlayStation 3 \cite{Debusschere12} as evidence. These
were all failures in seeding PRGs or generating random values. While
block ciphers and the like get much of the attention, these deterministic
components are almost boring when not fed sufficiently random data.

Randomness is crucial to many activities besides cryptography. Monte
Carlo simulations allow mathematical modeling of systems and functions
that are too complicated or expensive to solve directly. Such a simulation
gives us a statistically qualified numerical value by evaluating the
system or function at a number of randomly chosen inputs. This random
selection is often done with a pseudorandom generator for reasons
of cost and speed. Monte Carlo methods are a mainstay in fields as
varied as computational physics, financial modeling, numerical optimization,
and a number of engineering disciplines. 

In fact, it was the application of Monte Carlo methods to integrated
circuit (IC) testing in the 1980's that drove a body of research to
find better and cheaper pseudorandom bits \cite{Pries86,Gloster88,Hortensius89,Bardell90,Serra90,Cattell96,GuanT04}.
Because of the combinatoric nature of the possible failures in an
IC, most problems in this kind of testing are NP-Complete. Monte Carlo
methods offered a path to high-confidence test results without the
need to wrestle with these combinatoric problems. The common solution
was to add some circuitry to accept static test vectors from a test
fixture and produce some sort of pass/fail check inside the IC. In
addition, a separate mode was supported where a short seed sent to
the IC would generate a much longer yet deterministic sequence of
pseudorandom test vectors. These would be run and the same pass/fail
check performed. This cut down transmitted data to the IC and greatly
increased the speed of the test.

Since this circuitry was on the IC itself but was not directly valuable
to the customer, there was pressure to reduce the resources it consumed--transistors,
connection lengths, CMOS area, etc. The PRG used had traditionally
been a linear feedback shift register (LFSR), but in the late 1980's
researchers began to compare LFSRs with one-dimensional cellular automata
(CA). After understanding some behavioral basics such as how to achieve
a maximal period, experiments showed that CA gave better test coverage
with lower cost in ICs than LFSRs \cite{Gloster88,Hortensius89}.
Their abilities as PRGs were established.

It was natural then to investigate CA as primitives for cryptographic
schemes, and through the 1990's and 2000's, many such schemes were
put forward. As is often the case with cryptographic schemes, many
have been broken and none have really garnered widespread attention.
This may be due in part to the lack of specifics in many of these
schemes about key scheduling and initialization that would facilitate
implementation and focused cryptanalysis. 

If so, it would be somewhat ironic since this same period has seen
the opening and continued growth of standards-based algorithm selection
(for AES, SHA-1, 2, \& 3, and eStream), governing body standardization,
and the increasing popularity of provable security in cryptographic
theory. This last area began in the early 1980's as the application
of the techniques from complexity theory to problems in cryptography,
even in concrete situations of constant size. Proofs concerning cryptographic
primitives and protocols are given in relation to mathematically precise
definitions and (hopefully) minimal assumptions about hardness. Proofs
of protocol security are often provided as reductions from breaking
the primitive they employ. This allows assumptions to be minimal and
explicit so that effort spent on cryptanalysis can be focused and
re-used for these primitives.

In this thesis, we will attempt to bring a provable security approach
to cryptography based on CA. Specifically, we will consider the ability
of CA constructions presented in the literature to act as cryptographic
PRGs since PRGs are fundamental to all of cryptography. As we will
see, this ability depends almost entirely on the noninvertibility
of CA. The remainder of this section gives the necessary definitions
and criteria for this evaluation. In \secref{CA} we present an overview
of CA and related terminology then review some important schemes and
seminal breaks of CA constructions from the literature. \secref{AnalysisNonLinearCA}
presents new analysis of CA using non-linear rules, including a new
algorithm to invert certain non-linear rules in two-state, three-neighbor
cyclic CA as well as a proposed algorithm to derandomize cryptanalysis
of non-linear CA when an output sequence over time from a single cell
is known. Unconvinced of the suitability of simple CA for use as primitives
in secure cryptographic systems, we then examine how they may be enhanced
to allow for proofs of hardness and, eventually, security. A new construction
called finite state cellular automaton which adds a minimal amount
of complexity to cells is shown to be provably hard to invert in \secref{NewConstruction}.
We then use this theoretical construction to build the Chasm family
of concrete PRGs and give some experimental results in \secref{FscaBasedPrg}.

\subsection{Definitions of Security}

Let us first make clear the context in which we consider PRGs and
exactly what we mean by that term. Randomness in modern cryptography
begins with Shannon's information theory \cite{Shannon49a}, where
randomness is something of a measure of the lack of information. Kolmogorov
(and Chaitin) added the notion of descriptive complexity \cite{Kolmogorov65},
which classifies a string to be no more random than the program required
to generate it. Both of these notions allow us to quantify randomness,
but when we do so we find ``perfect'' randomness only at extreme,
theoretical limits. These concepts do not help us create or assess
practical randomness.

A third notion defines randomness relative to an observer. A string
is ``random'' to an observer if it cannot be distinguished from
a truly random string within the bounds of that observer's computational
resources. This in a sense defines randomness as the extent to which
the observer is unable to compute any meaningful information from
a string. Since this is a different concept from truly random, we
instead use the term \emph{pseudorandom}. The observers we are concerned
with are algorithms whose running time are bounded in some way.

Pseudorandom generators were the first primitives to be defined using
this notion of pseudorandomness. The original contemporaneous definitions
are due to Yao \cite{Yao82} and also Blum and Micali \cite{Blum84}.
We give a more common, modern definition.
\begin{defn}
\label{def:PRG}A deterministic polynomial-time algorithm $G:\{0,1\}^{n}\rightarrow\{0,1\}^{\ell(n)}$
is a \emph{pseudorandom generator} with stretching function $\ell:\Nat\rightarrow\Nat$
such that $\ell(n)>n$ if for any probabilistic polynomial-time algorithm
$A$, for any positive polynomial $p$:
\[
\left|\Pr\left[A(G(U_{n}))=1\right]-\Pr\left[A(U_{\ell(n)})=1\right]\right|<\frac{1}{p(n)}
\]
for all sufficiently large $n$ where $U_{k}$ is a $k$-bit string
drawn uniformly at random from $\{0,1\}^{k}$ and the probabilities
are taken over the respective $U_{k}$ and over the coins of $A$.
\end{defn}
Use of the term pseudorandom in this thesis should be assumed to imply
this notion of computational indistinguishability. Pseudorandomness
subject to a fixed set of statistical tests will be referred to as
\emph{statistical pseudorandomness}.\emph{ }

Using this definition, we can see that a statistical test against
the output of a generator $G$ is simply a special kind of distinguishing
algorithm $A$. While we can run a battery of statistical tests, we
still cannot truly satisfy this definition without some assumption
on the hardness of inverting $G$. If $G$ were easy for $A$ to invert,
$A$ could:
\begin{itemize}
\item Assume its input $I_{A}$ is the output of $G$, and invert $G$ for
that output, to arrive at an assumed input $I_{G}$.
\item Run $G(I_{G})$ and compare the result to $I_{A}$. If they are the
same, $A$ outputs a 1.
\end{itemize}
This algorithm would let $A$ distinguish the outputs of $G$ quite
easily and so $G$ would not be a pseudorandom generator as defined.
Thus, some notion of noninvertibility or one-wayness of $G$ is essential.
We define this notion as follows:
\begin{defn}
\label{def:OneWay}A function $f:\{0,1\}^{*}\rightarrow\{0,1\}^{*}$
is said to be \emph{one-way} if $f$ is polynomial-time computable
and for every probabilistic polynomial-time algorithm $A$, for any
positive polynomial $p$:
\[
\Pr_{x\in\{0,1\}^{n}}\left[f(A(f(x)))=f(x)\right]<\frac{1}{p(n)}
\]
for all sufficiently large $n$ where the probability is taken uniformly
over the choices of $x$ and the coins of $A$.
\end{defn}
It turns out that this notion is the key to PRGs. Impagliazzo, Levin,
and Luby prove in \cite{Impagliazzo89} that the existence of one-way
functions is a necessary and sufficient condition for the existence
of PRGs. Note that $f$ must be generally hard to invert as the probabilities
are taken over all $x$. This is juxtaposed with NP-Complete problems,
where having occasional strings which cannot be decided in polynomial-time
is sufficient to consider the whole language ``hard.'' Note also
that by this definition, $A$ need not compute $x$ exactly, just
any pre-image of $f(x)$. If this pre-image is not unique the function
is noninvertible by some definitions of that term. To be clear, we
will use noninvertible to mean one-way and inversion to mean finding
any pre-image of $f(x)$.

Other properties of PRGs are very useful in practice, especially in
scenarios where the internal state of the generator (especially software-based
generators) may become known to an attacker. In this scenario, the
attacker may be able to predict future pseudorandom outputs and/or
recreate past outputs. If the generator is shared among users (as
is /dev/random on Linux systems), and if those outputs are used in
another user's cryptosystem (as a key or initialization vector), that
system can easily be compromised. Security against this scenario is
captured by the following properties:
\begin{description}
\item [{Backward-Secure:}] Future outputs are secure against a compromise
of the internal state of the generator which occurred in the past.
Equivalently, given the current internal state of the generator, an
attacker is unable to predict future outputs with non-negligible success.
As PRGs are deterministic algorithms, this property is difficult for
a generator to display intrinsically, and is usually achieved only
be external re-seeding of the generator.
\item [{Forward-Secure:}] Previous outputs are secure against a compromise
of the internal state of the generator which may occur in the future.
Given the current internal state, an attacker is unable to guess past
outputs with non-negligible success. This implies the PRG's function
is effectively one-way.
\end{description}
The notion of forward security was first formalized by Bellare and
Yee \cite{Bellare03}, and is also applied to symmetric encryption
schemes regarding key compromise. As seen in our example of shared
/dev/random, this property is very powerful, and is to be expected
in modern PRGs.

We can see that we need noninvertibility to guarantee pseudorandom
behavior going forward, and we also need it to protect previous outputs.
Thus noninvertibility is a must-have property of any primitive used
as a cryptographic PRG.

\section{\label{sec:CA}Cellular Automata}

A Cellular Automaton is a discrete time and space dynamical system
consisting of an array of cells, each of which implements a (usually
simple) automaton. These cells can be arranged over one or two dimensions
and connected to their neighbors in various configurations. The cells
use the state of their neighbors and their own state to decide a next
state to transition to. In some cases, very simple configurations
of cells using very simple rules display surprisingly complex behavior.

Cellular Automata were first proposed by John von Neumann while working
at Los Alamos on the problem of building self-replicating systems%
\footnote{See \cite{Wolfram02} and \cite{Sarkar00} for a more complete history.%
}. His 1953 construction had 200,000 cells over two-dimensions where
each cell used 29 states to model various operations of the robot.
Stanislaw Ulam picked up this concept again in the 1960's with work
on recursively-defined geometrical objects. He noted that in two-state
two-dimensional CA with simple rules, a single non-conforming cell
generated complex patterns which may model biologic interactions.
John Conway experimented further with connected cells using simple
rules and developed ``The Game of Life,'' which popularized two-state
two-dimensional CA. Many other constructions of one and two dimensions
are useful for a host of applications too numerous to mention.

Pseudorandom generators and cryptography using CA flow from work by
Stephen Wolfram in the 1980s and the response to it. This is the trail
we will follow.

\subsection{Definitions}

Formally, a CA is a vector $F=\langle f_{1},f_{2},\ldots,f_{n}\rangle$
of functions $f_{i}:\{0,1\}^{N}\mapsto\{0,1\}$ for some $N\in\mathbb{N}$.
Each $f_{i}$ is evaluated at some discrete time $t$ to produce a
vector $S^{(t)}=\langle s_{1}^{(t)},s_{2}^{(t)},\ldots,s_{n}^{(t)}\rangle$
of values $s_{i}^{(t)}$. The values $s_{i}$ are also known as the
\emph{state} (in the sense of stored value) of cell $i$, and so $S$
is sometimes called the \emph{state vector}.\emph{ }When each value
$s_{i}^{(t)}\in\{0,1\}$, the CA is referred to as a \emph{two-state}
CA.\emph{ }Each $f_{i}$ is traditionally called the \emph{transition
function} of the cell (as in Finite State Automata).\emph{ }The $N$
inputs to cell $i$ normally come from its own output at the previous
time step as well as $(N-1)/2$ of its immediate neighbors on either
side. Thus, in an $N$\emph{-neighbor }CA, the value of cell $i$
at time $t+1$ is defined by 
\[
s_{i}^{(t+1)}=f_{i}(s_{i-(N-1)/2}^{(t)},\ldots,s_{i}^{(t)},\ldots,s_{i+(N-1)/2}^{(t)})
\]
For convenience, we will denote a neighborhood of values at time $t$
from cell $i$ to cell $j$ inclusive with $S_{i:j}^{(t)}$, and likewise
for neighborhoods of functions.

To determine what happens at the outer-most cells, a CA must specify
an input for the missing neighbor. A \emph{null-boundary} CA provides
a constant 0 for these neighbors, while a \emph{cyclic-boundary} CA
provides the value from the outer-most cell from the opposite end
of the array. We will be concerned only with cyclic-boundary CA.

Two-state, 3-neighbor, cyclic-boundary CA are of particular interest
in the literature. This is the simplest configuration shown to have
complex behavior, depending on $F$. Most authors refer to the transition
function of the CA as the ``rule,'' following the numbering convention
of Wolfram\cite{Wolfram86b}. Rules are numbered by considering the
output bits for each of the $2^{N}$ possible inputs as a binary number,
then interpreting that number in decimal. For example, a function
that, on input strings of $111,110,\ldots,000$, produces output bits
$0,0,0,1,1,1,1,0$ respectively is called rule 30. Using this scheme,
the vector $F$ will be referred to as the \emph{rule set} or the
\emph{rule vector}.

We will often work in standard Boolean algebra to describe these rules,
using $\cdot$ or concatenation to indicate AND, $+$ for OR, and
$\oplus$ for XOR. Rules using only XOR giving $f$ the form $f(x_{1},x_{2},x_{3})=a_{0}\oplus a_{1}x_{1}\oplus a_{2}x_{2}\oplus a_{3}x_{3}$
we define as \emph{affine} in GF(2). Affine rules having $a_{0}=0$
are called linear. For a linear rule, complementing one of the inputs
is equivalent to setting $a_{0}=1$, and so affine rules are linear
rules with an additive offset (as in other domains). This leads to
the occasional use of the term \emph{additive} as a synonym for affine.
Rules that are not linear are called \emph{nonlinear}. Though all
rules have three arguments by definition in a 3-neighbor CA, not all
rules make use of all inputs. We will describe those rules that make
use of only two inputs as \emph{binary} and those that use all three
as \emph{ternary}. The rules that have attracted the most attention
are: 
\begin{eqnarray*}
\text{rule 30:\quad}s_{i}^{(t+1)} & = & s_{i-1}^{(t)}\,\oplus\,(s_{i}^{(t)}\,+\, s_{i+1}^{(t)})\\
\text{rule 150:\quad}s_{i}^{(t+1)} & = & s_{i-1}^{(t)}\,\oplus\, s_{i}^{(t)}\,\oplus\, s_{i+1}^{(t)}\\
\text{rule 90:\quad}s_{i}^{(t+1)} & = & s_{i-1}^{(t)}\,\oplus\, s_{i+1}^{(t)}\\
\text{rule 105:\quad}s_{i}^{(t+1)} & = & 1\oplus s_{i-1}^{(t)}\,\oplus\, s_{i}^{(t)}\,\oplus\, s_{i+1}^{(t)}\\
\text{rule 165:\quad}s_{i}^{(t+1)} & = & 1\oplus s_{i-1}^{(t)}\,\oplus\, s_{i+1}^{(t)}
\end{eqnarray*}
These and some slight variations will be the focus of our analysis. 

Even such a small set of rules still allows a variety of options in
constructing CAs. A \emph{uniform} CA is one where each cell applies
the same rule at each time step. In such cases, the rule vector may
be denoted $F_{r}$ to indicate that Wolfram rule $r$ is used for
each cell. Alternatively, a \emph{hybrid} CA may assign different
rules to different cells. A hybrid CA where the rules are symmetric
about a center cell will be referred to as \emph{symmetric}; those
hybrid CA without this property are called \emph{asymmetric}. Uniform
and hybrid CAs keep the assigned rules constant across time steps.
CAs which vary a given cell's rule over time according to some scheme
are called \emph{programmable}. A CA (uniform or hybrid) which uses
only linear rules is known as a \emph{linear CA}, and likewise for
nonlinear rules.

The use of XOR on just a single input can be of critical importance.
A rule $f$ is said to be \emph{left-toggle} if the following property
holds for all $S\in\{0,1\}^{3}$:
\[
1\oplus s_{i}^{(t+1)}=f(1\oplus s_{i-1}^{(t)},s_{i}^{(t)},s_{i+1}^{(t)})
\]
 A rule is \emph{right-toggle} if instead
\[
1\oplus s_{i}^{(t+1)}=f(s_{i-1}^{(t)},s_{i}^{(t)},1\oplus s_{i+1}^{(t)})
\]
holds. Rules combining the left or right bits with XOR will be left-
or right-toggle, respectively. It will be important in the analysis
of toggle rules to have the following lemma.
\begin{prop}
\label{prop:ToggleProp}Let $S\in\{0,1\}^{n}$ be the state vector
of an $n$-cell CA, $f:\{0,1\}^{3}\rightarrow\{0,1\}$ be a left-toggle
rule, and $g:\{0,1\}^{3}\rightarrow\{0,1\}$ be a right-toggle rule.
After applying $f$ to cell $i$ at time $t$, $s_{i-1}^{(t)}=f(s_{i}^{(t+1)},s_{i}^{(t)},s_{i+1}^{(t)})$.
Similarly, after applying $g$ to cell $i$ at time $t$, $s_{i+1}^{(t)}=f(s_{i-1}^{(t)},s_{i}^{(t)},s_{i}^{(t+1)})$.\end{prop}
\begin{proof}
We first address the case of a left-toggle rule, and consider

\begin{equation}
(s_{i}^{(t+1)}\oplus s_{i-1}^{(t)})\oplus s_{i}^{(t+1)}=f((s_{i}^{(t+1)}\oplus s_{i-1}^{(t)})\oplus s_{i-1}^{(t)},s_{i}^{(t)},s_{i+1}^{(t)}).\label{eq:ToggleProof}
\end{equation}
 Either $s_{i}^{(t+1)}=s_{i-1}^{(t)}$ or $s_{i}^{(t+1)}\neq s_{i-1}^{(t)}$.
In the first case, we have
\[
0\oplus s_{i}^{(t+1)}=f(0\oplus s_{i-1}^{(t)},s_{i}^{(t)},s_{i+1}^{(t)})
\]
which is just the definition of $f$. But since $s_{i}^{(t+1)}=s_{i-1}^{(t)}$,
we can exchange these values and write (\ref{eq:ToggleProof}) as
\[
s_{i-1}^{(t)}=f(s_{i}^{(t+1)},s_{i}^{(t)},s_{i+1}^{(t)})
\]
which is the identity we seek. In the second case, we have 
\[
1\oplus s_{i}^{(t+1)}=f(1\oplus s_{i-1}^{(t)},s_{i}^{(t)},s_{i+1}^{(t)}).
\]
Since $s_{i}^{(t+1)}\neq s_{i-1}^{(t)}$, $1\oplus s_{i}^{(t+1)}=s_{i-1}^{(t)}$
and $1\oplus s_{i-1}^{(t)}=s_{i}^{(t+1)}$. Substituting these identities
gives us 
\[
s_{i-1}^{(t)}=f(s_{i}^{(t+1)},s_{i}^{(t)},s_{i+1}^{(t)})
\]
This proves the Lemma for all $f$. The proof for right-toggle rules
is similar. 
\end{proof}
Thus, a left-toggle rule allows us to substitute the rule output for
the left neighbor in order to solve for that left neighbor. The reflected
observation holds for right-toggle rules as well.

A \emph{temporal sequence} is a sequence of the output values of a
single cell taken over multiple time steps. We say a CA \emph{produces
}a temporal sequence (or just sequence) $\sigma_{i}$ at cell $i$
if the values $s_{i}$ over time match the elements of $\sigma_{i}$.
Of particular interest in the literature is the central temporal sequence,
used as a pseudorandom stream in many CA-based constructs. Unless
otherwise noted, we assume $n$ is odd to make clear which is the
central temporal sequence. The \emph{right-adjacent} sequence of a
temporal sequence is the temporal sequence one position to its right,
and similarly for the \emph{left-adjacent} sequence. 

Finally, it bears mentioning that cellular automata can do some funny
things. We refer to the state vector at time step $t=0$ as the \emph{initial
state} and also the \emph{seed}, a term frequently used for the initial
input to random generators. Some CA are not capable of generating
their initial state at a later time step. Such a state vector is known
as a \emph{garden of Eden} state as it can only exist in the beginning.
Some CA evolve into a single, fixed state (usually all 0s), which
is known as a \emph{dead-end state} for that CA.

\subsection{A Brief Overview of Research on Pseudorandom Generation with CA }

CA entered the cryptography and random number generation domains with
Stephen Wolfram's claims about rule 30. In \cite{Wolfram86b}, he
proposed a uniform $n$-cell cyclic arrangement using rule 30 with
the temporal sequence of the center cell used as a random stream.
The periods of sequences produced depend on the number of cells, but
also on the initial seed. The maximal period was estimated to be $2^{0.61(n+1)}$.
This generator was shown to pass a suite of 7 statistical tests, performing
best when the sequence was much shorter than the period of the CA.
The CA performed better than an LFSR of the same size, but not as
well as a linear congruential generator or the bytes of $\sqrt{2}$,
$e$, or $\pi$. 

Hortensius, et. al. \cite{Hortensius89} compared this rule 30-based
generator with a hybrid configuration using rules 90 and 150, based
on work in \cite{Pries86}, and also with traditional Linear Feedback
Shift Registers (LFSRs). The focus of these experiments was generation
of test vectors in VLSI manufacturing where layout concerns are important.
They showed that both kinds of CA performed better than LFSRs in statistical
tests. The null boundary hybrid CA had on average longer periods than
cyclic boundary uniform rule 30 CAs. They also catalog the configuration
of rule sets which produce maximal periods in the hybrid configurations,
along with periods for other configurations tested. These results
were arrived at through exhaustive search.

Serra, et. al. took up the question of how to synthesize a null-boundary
linear CA given a primitive polynomial in \cite{Serra90}. Given a
linear CA, each new state vector can be described as a linear system,
and so can be represented as a matrix operation over GF(2). This matrix
is called the transition matrix of the automaton. The authors first
establish an isomorphism between the transition matrix of an LFSR
with that of a CA by showing they are similar (i.e. have the same
characteristic polynomial) and so describe the same linear transform
under different bases. They then give an algorithm to produce the
transition matrix for a null-boundary linear CA using only rules 90
and 150 (from which the rule set is clear) given a characteristic
polynomial. While this algorithm is based on searching a space of
$2^{\left\lfloor n/2\right\rfloor }$ vectors, later algorithms in
\cite{Cattell96} and \cite{Cho07} improve these results. The first
solves a quadratic congruence on subpolynomials using Euclid's algorithm
in a finite field, the second uses a revised Lanczos tridiagonalization
method in GF(2) to find one of two possible CA for the given polynomial.

These results are only for null-boundary linear CA. Bardell in \cite{Bardell90}
shows that the outputs of linear CA and LFSRs are identical when a
phase shift between output bit sequences is accounted for. Bardell
also conjectures that no \emph{cyclic}-boundary linear CA has maximal
period. This conjecture is proved by Nandi in \cite{Nandi96}, who
reported that the characteristic polynomial of any transition matrix
using cyclic boundary conditions is factorizable, and so cannot be
primitive. Nandi also claims that periodic boundary linear CA provide
better statistical randomness than null boundary, due to the fixed
0s at the ends. This claim is supported by \cite{Shin09}.

Nandi, et. al. examine the group behavior of hybrid CAs over rules
51, 153, and 195 with null boundary conditions in \cite{Nandi94}
and show these rules lead to CA which are even permutations. By combining
several rule sets in a programmable CA, they create transformations
which generate an alternating group of even permutations. These transformations
then become primitives on which they base block and stream ciphers.These
cryptosystems were broken in \cite{Blackburn97}, which showed that
the groups formed are actually a subset of the affine group, not the
alternating group, of degree $n$ and therefore are easily recreated
with sufficient plaintext/ciphertext pairs. This was improved in \cite{Mihaljev97}
using ciphertext only.

In the same year, Sipper and Tomassini gave a genetic algorithm approach
to evolving a single ``good'' CA by using an entropy metric and
introducing mutations of rule changes randomly. They looked at 3-neighbor,
two-state, cyclic hybrid CA with $n=50$ over 300 random initial states
run for 4096 steps. The entropy of each cell is computed over time,
and with probability 0.001, cells would mix rules with a neighbor
by swapping the neighbor rule's output value assigned to certain input
combinations. Two resulting CA (a mixture of rules 165, 90, and 150
in one case and rules 165 and 225 in the other) were compared to a
uniform rule 30 CA and a hybrid 150/90 CA over 4 statistical tests,
showing favorable results.

The approach of evolving CA gave rise to a series of papers. See \cite{Guan03,Tomassini99,Tomassini01,Seredynski04}.
Most of these end up focusing on the main 4 linear rules, 150, 105,
90, and 165.

In \cite{Wolfram86a}, Wolfram proposed a cryptosystem based on his
rule 30 temporal sequence random generator from \cite{Wolfram86b}
as a key stream with which the plaintext could be XORed. This scheme
was broken by Meier and Staffelbach in \cite{Meier91} %
. The authors first showed that the temporal sequence is not hard
to recover in the case of known plaintext attacks. Then, given this
temporal sequence, they showed that the left half of the CA's computational
history was uniquely determined if the right-adjacent sequence could
be guessed due to the left-toggle property of rule 30. Further, the
right-adjacent sequence could be determined by guessing the right
half of the seed and running the CA forward for $n/2$ time steps.
They also showed that not all right halves of the seed are equiprobable,
so far fewer than $n/2$ guesses are required. For $n=300$, for instance,
a probabilistic algorithm requires 18.1 bits of entropy to recover
the seed with a probability of 0.5. This algorithm remains a seminal
one in cryptanalysis of CA cryptosystems, and one which we will seek
to improve.

In \cite{Koc97}, Koç and Apohan investigate a claim made by Wolfram
that recovering a seed value given a sequence of states over an $n$-cell
automaton using rule 30 was NP-Complete. Koç and Apohan present an
inversion algorithm which finds the best affine approximation of the
transition function and then solves an $n$-variable Boolean linear
equation to get a good approximation $S^{*}$of $S^{(t-1)}$. It then
checks the affine approximation by using $S^{*}$ to re-compute $S^{(t)}$.
If there are errors, the algorithm resorts to a search for combinations
that correct them. Using this algorithm, Koç and Apohan shows that
rule 30 can be inverted in time $O(n)$ for some seeds and $O(2^{n/2})$
worst case. We will improve this result.

Sen et. al. \cite{Sen02} present a (rare) fairly complete description
of an entire cryptosystem based on CAs. The algorithm is multistage,
with some key management and both linear and non-linear rules applied
to the plaintext. The system was broken by Bao in \cite{Bao03} using
only hundreds of chosen plaintexts and very little computation. Bao
presents an equivalent transform and searches a small space of one
of the parameters of the cryptosystem to decrypt any ciphertext with
probability 0.5.

A series of papers propose using programmable CAs, wherein the rule
vector $F$ applied at each time step can be controlled by external
circuitry. These schemes usually focus on 2 or 4 of the linear rules
(150, 105, 90, 165). This seems to do very well in statistical testing,
but makes little or no claims about use in cryptography. See \cite{GuanT04}.

Shin et. al. \cite{Shin09} analyzed the conditional probability distributions
of different combinations of binary operations and showed that only
XOR produces ``cryptographic'' PRGs, since all other combinations
are skewed. They show that 64-cell hybrid CA based on rules having
uniformly distributed outputs pass almost all tests from the Diehard
\cite{diehard} statistical test battery.

\section{\label{sec:AnalysisNonLinearCA}Analysis of Non-Linear CA}

When viewed from a computational complexity perspective, it's not
clear that a 2 state, 3 neighbor CA is capable of hardness at all.
Any resolution to this question would certainly depend on the rule
set of the CA. Many rules simply do not generate any complex output
patterns, and some degenerate to very small cycles or constant patterns
very quickly. The authors of \cite{Len03} conclude that no CA using
only uniform rules is suitable for use in cryptographic applications
by demonstrating the following: Of all uniform CA subjected to frequency,
serial, poker, gap, and auto-correlation statistical randomness tests
using the evolution of the center cell as a random stream, only 22
rules passed. These rules were then subjected to a linear complexity
test. The linear complexity of a sequence is defined as the length
of the shortest LFSR that produces the sequence. We know via the Berlekamp/Massey
algorithm \cite{Massey69} that an LFSR of length no more than $\ell/2$
can be synthesized for any sequence of length $\ell$. The linear
complexity test showed that only the following rules generate $\ell$-bit
sequences whose linear complexity approaches $\ell/2$ (the ideal):
30, 45, 75, 86, 89, 101, 106, 120, 135, 149, 169, 225. The authors
then observe that all of these rules are left- or right-toggle rules,
and adapt the Meier-Staffelbach algorithm to work on either side and
recover the initial state of the CA with no more than $2^{\left\lfloor n/2\right\rfloor }$
trials.

These results and the algorithms of Koç/Apohan \cite{Koc97} and Meier/Staffelbach
\cite{Meier91} for attacking rule 30 all seem to imply that $2^{n/2}$
bits are the most one would have to guess to know everything about
the CA's history one wanted. While this is still an exponential bound,
we'd like to know for certain that the security of the PRG is related
to the full seed length. Further, both algorithms often do much better
than worst case. This apparent weakness invites further investigation
into whether $2^{n/2}$ is the tightest upper bound that can be achieved.
We will focus our efforts on rule 30 due to the large body of literature
for this rule. The other rules listed above are in most cases simple
variations using negation of a term or reflection of the inputs, and
we would expect results against rule 30 to also apply there as well.
Unless otherwise noted, we are concerned with cyclic CA as they are
most commonly used in the literature.

We therefore examine the common structures using rule 30 in search
of techniques to improve the worst case bounds. Specifically, we examine
the case of a known state vector over $n$ cells for a given time
step (the problem addressed by Koç/Apohan) and the case of a known
temporal sequence from one of $n$ cell over for $n/2$ time steps
(addressed by Meier/Staffelbach). In both cases we investigate techniques
to recover the initial state.

\subsection{Improvements to Koç and Apohan}

The algorithm presented in \cite{Koc97} selects the best affine approximation
of the rule used in a CA and applies the inverse affine transform
to estimate the previous state. The success of this technique depends
primarily on the rule(s) used in the CA. Those CAs using affine rules
can be represented by a linear system in GF(2). In this case, the
system can be solved for the state vector at time $t-1$ if the entire
state vector is known at time $t$. This solution requires only $O(N^{2}n)$
operations for an $N$-neighbor CA with $n$ cells using Gaussian
elimination. To recover the state $t$ time steps ago, this process
must be repeated $t$ times. However, for affine rules, the inverse
transform can also be represented as a linear recurrence $S^{(t-1)}=MS^{(t)}+b$
in GF(2), where $M$ is an $n\times n$ Boolean matrix and the $1\times n$
offset vector $b$ models any negation operations in the rule of each
cell. Given some $S^{(t)}$, this recurrence can be composed $t$
times and solved to recover $s^{(0)}$ directly. The run time of inverting
uniform CA using known affine rules then is clearly polynomially bounded.

The affine approximation does, however, have difficulty with non-linear
rules. Applying the inverse of an affine approximation to $S^{(t)}$
gives an estimate $S^{*}$ which can differ from the true $S^{(t-1)}$.
If so, $S^{*}$ may not be a valid predecessor of $S^{(t)}$ under
the non-linear rule. The algorithm of Koç/Apohan resolves this by
searching through templates of these mismatches in the context of
the neighborhood in which they occur to find possible modifications
to the estimate vector $S^{*}$ which resolves the differences. It
is this search that pushes the bound on the running time up to $O(2^{n/2})$.

We observe that not all prior state vectors are equiprobable under
non-linear rules. Prior probabilities for 3-neighbor cells are given
in \tabref{rule30Priors}. Note that any sequence of length 3 has
only 4 possible prior states of length 5 under rule 30. Experimentation
shows this holds for other sequence lengths as well. Since rule 30
is left-toggle, choosing $s_{i}^{(t)}$ and $s_{i+1}^{(t)}$ uniquely
determine $s_{i-1}^{(t)}$ when $s_{i}^{(t+1)}$ is fixed by \propref{ToggleProp}.
But then $s_{i-1}^{(t)}$ and $s_{i}^{(t)}$ are known, so $s_{i-2}^{(t)}$
is determined if $s_{i-1}^{(t+1)}$ is known, and so on. Therefore,
the two right bits are sufficient to determine the $k+2$ predecessor
bits of any sequence of length $k$. This information can be used
to optimize the search for erroneous predecessor bits in the templates.
For instance, if the center value $s_{i}^{(t)}$ of a 3-cell neighborhood
is 1, there is a $\nicefrac{3}{4}$ chance that $s_{i-1}^{(t-1)}$
is 0. If the affine approximation does not immediately yield a unique
prior state, using such probabilities may inform the template search
and reduce the average-case search space of the algorithm.

\begin{table}
\begin{centering}
\begin{tabular}{|c|c|>{\centering}p{3.5cm}|}
\hline 
3-neighbor state & Possible 5-neighbor predecessors & \# of prior states with 1 in each position\tabularnewline
\hline 
\hline 
000 & (00000) (11101) (11110) (11111) & {[}3, 3, 3, 2, 2{]}\tabularnewline
\hline 
001 & (00001) (11010) (11011) (11100) & {[}3, 3, 1, 2, 2{]}\tabularnewline
\hline 
010 & (10101) (10110) (10111) (11000) & {[}4, 1, 3, 2, 2{]}\tabularnewline
\hline 
011 & (00010) (00011) (10100) (11001) & {[}2, 1, 1, 2, 2{]}\tabularnewline
\hline 
100 & (01101) (01110) (01111) (10000) & {[}1, 3, 3, 2, 2{]}\tabularnewline
\hline 
101 & (01010) (01011) (01100) (10001) & {[}1, 3, 1, 2, 2{]}\tabularnewline
\hline 
110 & (00101) (00110) (00111) (01000) & {[}0, 1, 3, 2, 2{]}\tabularnewline
\hline 
111 & (00100) (01001) (10010) (10011) & {[}2, 1, 1, 2, 2{]}\tabularnewline
\hline 
\end{tabular}
\par\end{centering}

\textsf{\footnotesize \caption{\label{tab:rule30Priors}\textsf{Prior states of 3 neighbors in uniform
rule 30 CA}}
}
\end{table}

More noteworthy is that rule 30 has certain patterns that always have
fixed bits in the prior state. Notice the patterns 010 and 110 centered
on cell $s_{i}^{(t)}$ have a fixed value in cell $s_{i-2}^{(t-1)}$
for all possible prior 5-neighbor states. The following proposition
shows why this must always be true.
\begin{prop}
\label{prop:1}Let $(abcde)\in\{0,1\}^{5}$ be the values in a 5-cell
neighborhood at time $t$ in a uniform rule 30 CA of $n\geq5$ cells.
Let $(xyz)\in\{0,1\}^{3}$ be the vales of the 3-cell neighborhood
resulting from evaluating the rule 30 function $f$ as $f(a,b,c),\: f(b,c,d)\: f(c,d,e)$
respectively. Then ($xyz=010)$ implies $a=1$ and $(xyz=110)$ implies
$a=0$.\end{prop}
\begin{proof}
First we show ($xyz=010)\implies(a=1)$. In rule 30, note the following
identities:

\begin{eqnarray}
x=0 & = & a\oplus(b+c)\implies a=b+c\label{eq:a}\\
y=1 & = & b\oplus(c+d)\implies b=\overline{c+d}\label{eq:b}\\
z=0 & = & c\oplus(d+e)\implies c=d+e\label{eq:c}
\end{eqnarray}
Then, by substituting (\ref{eq:c}) into (\ref{eq:a}), we have 
\begin{equation}
a=b+d+e
\end{equation}
Thus $a$ can only take the value 0 when $b=d=e=0$. But substituting
(\ref{eq:c}) into (\ref{eq:b}), we get

\begin{eqnarray*}
b & = & \overline{d+e+d}\\
 & = & \overline{d+e}
\end{eqnarray*}
So when $d=e=0,$ $b$ cannot be 0. Therefore, $a$ can never be 0
and must always be 1.

To see that $(xyz=110)\implies(a=0)$, 

\begin{eqnarray}
x=1 & = & a\oplus(b+c)\implies a=\overline{b+c}\label{eq:a2}\\
y=1 & = & b\oplus(c+d)\implies b=\overline{c+d}\label{eq:b2}\\
z=0 & = & c\oplus(d+e)\implies c=d+e\label{eq:c2}
\end{eqnarray}
 In particular, we note that (\ref{eq:a2}) is just the complement
of (\ref{eq:a}). For completeness, we can see

\[
a=\overline{b+d+e}
\]
 by substituting for c in (\ref{eq:a2}). But by (\ref{eq:b2}) and
(\ref{eq:c2}), $b\neq d+e$, so $a$ can never be 1.
\end{proof}
This information can be used to invert the entire state vector in
rule 30 anytime the pattern 010 occurs by using the left-toggle property
of the rule. We capture this in the following proposition.
\begin{prop}
\label{prop:2}Let $S^{(t)}\in\{0,1\}^{n}$ be the known state vector
in a cyclic boundary $n$-cell uniform rule 30 CA at time $t$. Suppose
$S_{i-1:i+1}^{(t)}=010$ for some $0\leq i<n$. Then $S^{(t-1)}$
is uniquely determined. \end{prop}
\begin{proof}
First, note that 
\[
s_{i-3}^{(t-1)}=s_{i-2}^{(t)}\oplus(s_{i-2}^{(t-1)}+s_{i-1}^{(t-1)})
\]
and that $s_{i-2}^{(t-1)}=1$ by \propref{1}. Then $s_{i-2}^{(t-1)}+s_{i-1}^{(t-1)}=1$
and so 
\[
s_{i-3}^{(t-1)}=s_{i-2}^{(t)}\oplus1=\overline{s_{i-2}^{(t)}}
\]
Since $s_{i-2}^{(t-1)}$ and $s_{i-3}^{(t-1)}$ are known, we can
then evaluate 
\[
s_{i-4}^{(t-1)}=s_{i-3}^{(t)}\oplus(s_{i-3}^{(t-1)}+s_{i-2}^{(t-1)})
\]
and likewise for $i-k$ where $k=1,\dots,n-i$ where $i-i-1$ evaluates
to $n$ because of the cyclic boundary. This allows the calculation
of all $s_{i}^{(t-1)}$ for $0\leq i<n$.
\end{proof}
Note that this technique does not work for the 110 pattern, since
knowing the middle bit $s_{i}^{(t-1)}$ in a 3-neighborhood when its
value is 0 does not determine $s_{i}^{(t-1)}+s_{i+1}^{(t-1)}$. However,
if we examine longer predecessor patterns we may be able to guess
two adjacent bits in the predecessor with high probability. Then we
can evaluate 
\[
s_{i-1}^{(t-1)}=s_{i}^{(t)}\oplus(s_{i}^{(t-1)}+s_{i+1}^{(t-1)})
\]
 and so on until we complete the previous state vector. Looking at
the 32 possible 5-neighbor blocks under rule 30, 11 have a fixed 1
position in their 7-neighbor predecessors, and another 16 reveal two
adjacent positions with probability $3/4$. Only 5 of the possible
values leave the probability of guessing two adjacent predecessor
bits at $1/2$. 

We observe that the analog of \propref{1} also holds for all left-
or right-toggle rules, and so each such rule can be inverted using
the techniques discussed above. \tabref{FixedPositions} shows the
patterns of interest for each of the rules identified. These results
suggest some general weakness in toggle rules that can be used to
find predecessors for any given state vector.

\begin{table}
\begin{centering}
\begin{tabular}{l>{\raggedright}p{1.7cm}>{\raggedright}p{4cm}>{\raggedright}p{3.5cm}}
\toprule 
{\footnotesize Rule} & {\footnotesize 3-neighbor state} & {\footnotesize Possible 5-neighbor predecessors} & {\footnotesize \# of prior states with 1 in each position}\tabularnewline
\midrule
\multirow{2}{*}{{\footnotesize 30}} & {\footnotesize 010} & {\footnotesize 10101 10110 10111 11000} & {\footnotesize {[}4, 1, 3, 2, 2{]}}\tabularnewline
 & {\footnotesize 110} & {\footnotesize 00101 00110 00111 01000} & {\footnotesize {[}0, 1, 3, 2, 2{]}}\tabularnewline
\cmidrule{2-4} 
\multirow{4}{*}{{\footnotesize 45}} & {\footnotesize 000} & {\footnotesize 11001 11100 11110 11111} & {\footnotesize {[}4, 4, 3, 2, 2{]}}\tabularnewline
 & {\footnotesize 001} & {\footnotesize 10010 10011 11000 11101} & {\footnotesize {[}4, 2, 1, 2, 2{]}}\tabularnewline
 & {\footnotesize 100} & {\footnotesize 01001 01100 01110 01111} & {\footnotesize {[}0, 4, 3, 2, 2{]}}\tabularnewline
 & {\footnotesize 101} & {\footnotesize 00010 00011 01000 01101} & {\footnotesize {[}0, 2, 1, 2, 2{]}}\tabularnewline
\cmidrule{2-4} 
\multirow{6}{*}{{\footnotesize 75}} & {\footnotesize 001} & {\footnotesize 01000 01001 01011 11110} & {\footnotesize {[}1, 4, 1, 2, 2{]}}\tabularnewline
 & {\footnotesize 010} & {\footnotesize 10010 10111 11100 11101} & {\footnotesize {[}4, 2, 3, 2, 2{]}}\tabularnewline
 & {\footnotesize 011} & {\footnotesize 10000 10001 10011 10110} & {\footnotesize {[}4, 0, 1, 2, 2{]}}\tabularnewline
 & {\footnotesize 101} & {\footnotesize 01110 11000 11001 11011} & {\footnotesize {[}3, 4, 1, 2, 2{]}}\tabularnewline
 & {\footnotesize 110} & {\footnotesize 00010 00111 01100 01101} & {\footnotesize {[}0, 2, 3, 2, 2{]}}\tabularnewline
 & {\footnotesize 111} & {\footnotesize 00000 00001 00011 00110} & {\footnotesize {[}0, 0, 1, 2, 2{]}}\tabularnewline
\cmidrule{2-4} 
\multirow{2}{*}{{\footnotesize 86}} & {\footnotesize 010} & {\footnotesize 00011 01101 10101 11101} & {\footnotesize {[}2, 2, 3, 1, 4{]}}\tabularnewline
 & {\footnotesize 011} & {\footnotesize 00010 01100 10100 11100} & {\footnotesize {[}2, 2, 3, 1, 0{]}}\tabularnewline
\cmidrule{2-4} 
\multirow{4}{*}{{\footnotesize 89}} & {\footnotesize 010} & {\footnotesize 00111 01001 10111 11101} & {\footnotesize {[}2, 2, 3, 2, 4{]}}\tabularnewline
 & {\footnotesize 011} & {\footnotesize 00110 01000 10110 11100} & {\footnotesize {[}2, 2, 3, 2, 0{]}}\tabularnewline
 & {\footnotesize 110} & {\footnotesize 00001 01101 10001 11001} & {\footnotesize {[}2, 2, 1, 0, 4{]}}\tabularnewline
 & {\footnotesize 111} & {\footnotesize 00000 01100 10000 11000} & {\footnotesize {[}2, 2, 1, 0, 0{]}}\tabularnewline
\cmidrule{2-4} 
\multirow{6}{*}{{\footnotesize 101}} & {\footnotesize 000} & {\footnotesize 00111 01111 10011 11111} & {\footnotesize {[}2, 2, 3, 4, 4{]}}\tabularnewline
 & {\footnotesize 001} & {\footnotesize 00110 01110 10010 11110} & {\footnotesize {[}2, 2, 3, 4, 0{]}}\tabularnewline
 & {\footnotesize 010} & {\footnotesize 00100 01100 10001 11100} & {\footnotesize {[}2, 2, 3, 0, 1{]}}\tabularnewline
 & {\footnotesize 011} & {\footnotesize 00101 01101 10000 11101} & {\footnotesize {[}2, 2, 3, 0, 3{]}}\tabularnewline
 & {\footnotesize 100} & {\footnotesize 00011 01001 10111 11001} & {\footnotesize {[}2, 2, 1, 2, 4{]}}\tabularnewline
 & {\footnotesize 101} & {\footnotesize 00010 01000 10110 11000} & {\footnotesize {[}2, 2, 1, 2, 0{]}}\tabularnewline
\cmidrule{2-4} 
\multirow{2}{*}{{\footnotesize 106}} & {\footnotesize 010} & {\footnotesize 00010 01010 10010 11100} & {\footnotesize {[}2, 2, 1, 3, 0{]}}\tabularnewline
 & {\footnotesize 011} & {\footnotesize 00011 01011 10011 11101} & {\footnotesize {[}2, 2, 1, 3, 4{]}}\tabularnewline
\cmidrule{2-4} 
\multirow{2}{*}{{\footnotesize 120}} & {\footnotesize 010} & {\footnotesize 00111 01000 01001 01010} & {\footnotesize {[}0, 3, 1, 2, 2{]}}\tabularnewline
 & {\footnotesize 110} & {\footnotesize 10111 11000 11001 11010} & {\footnotesize {[}4, 3, 1, 2, 2{]}}\tabularnewline
\cmidrule{2-4} 
\multirow{2}{*}{{\footnotesize 135}} & {\footnotesize 001} & {\footnotesize 10111 11000 11001 11010} & {\footnotesize {[}4, 3, 1, 2, 2{]}}\tabularnewline
 & {\footnotesize 101} & {\footnotesize 00111 01000 01001 01010} & {\footnotesize {[}0, 3, 1, 2, 2{]}}\tabularnewline
\cmidrule{2-4} 
\multirow{2}{*}{{\footnotesize 149}} & {\footnotesize 100} & {\footnotesize 00011 01011 10011 11101} & {\footnotesize {[}2, 2, 1, 3, 4{]}}\tabularnewline
 & {\footnotesize 101} & {\footnotesize 00010 01010 10010 11100} & {\footnotesize {[}2, 2, 1, 3, 0{]}}\tabularnewline
\cmidrule{2-4} 
\multirow{2}{*}{{\footnotesize 169}} & {\footnotesize 100} & {\footnotesize 00010 01100 10100 11100} & {\footnotesize {[}2, 2, 3, 1, 0{]}}\tabularnewline
 & {\footnotesize 101} & {\footnotesize 00011 01101 10101 11101} & {\footnotesize {[}2, 2, 3, 1, 4{]}}\tabularnewline
\cmidrule{2-4} 
\multirow{2}{*}{{\footnotesize 225}} & {\footnotesize 001} & {\footnotesize 00101 00110 00111 01000} & {\footnotesize {[}0, 1, 3, 2, 2{]}}\tabularnewline
 & {\footnotesize 101} & {\footnotesize 10101 10110 10111 11000} & {\footnotesize {[}4, 1, 3, 2, 2{]}}\tabularnewline
\cmidrule{2-4} 
\end{tabular}
\par\end{centering}

\textsf{\caption{\textsf{Fixed position patterns in toggle rules\label{tab:FixedPositions}}}
}
\end{table}

Rule 30 is known to have some garden of Eden state vectors and is
not strictly injective \cite{amoroso72,McIntosh91,Tyler}. In practice,
those $S\in\{0,1\}^{n}$ for which $F_{30}^{-1}(S)$ is not defined
appear only as initial conditions, so a partial inversion is sufficient
to recover a seed. Various authors \cite{Brown04,paek03} have addressed
a means to realize a partial inversion once a rule 30 CA is run past
time step 1, but none give a specific algorithm to find any possible
predecessor in a uniform cyclic rule 30 CA. The observation that toggle
rules require only two bits to make the full preceding state vector
known leads to such a partial inversion algorithm that runs in time
$\Theta(n)$. While the general mechanism may be known, there seemed
to still be some question as recently as 2011 (see chapter 10 by Wolfram
in \cite{Zenil11}). We give a formal statement and proof of the algorithm
here for completeness.
\begin{prop}
Let $S\in\{\mathtt{0},\mathtt{1}\}^{n}$ and let $f_{i}$ be a function
either of the form $f_{i}(S_{i-1:i+1})=S_{i-1}\oplus g(S_{i},S_{i+1})$
or $f_{i}(S_{i-1:i+1})=g(S_{i-1},S_{i})\oplus S_{i+1}$ for some $g:\{0,1\}^{2}\rightarrow\{0,1\}$.
Let $r$ be the rule number of $f_{i}$ and let $F_{r}$ be the rule
vector of a cyclic boundary $n$-cell CA having state vector $S^{(t)}=S$
at time $t$. Then all valid $S^{(t-1)}$ if any exist are computable
in time $\Theta(n)$. \end{prop}
\begin{proof}
We first consider the case of the left-toggle function $f_{i}(S_{i-1:i+1})=S_{i-1}\oplus g(S_{i},S_{i+1})$.
For $S^{(t-1)}$ to exist, it must satisfy $S_{i}^{(t)}=S_{i-1}^{(t-1)}\oplus g(S_{i}^{(t-1)},S_{i+1}^{(t-1)})$
for $2\leq i<n$ and $S_{1}^{(t)}=S_{n}^{(t-1)}\oplus g(S_{1}^{(t-1)},S_{2}^{(t-1)})$
as well as $S_{n}^{(t)}=S_{n-1}^{(t-1)}\oplus g(S_{n}^{(t-1)},S_{1}^{(t-1)})$
due to boundary conditions. Consider \algref{InvertToggleRule}, \nameref{alg:InvertToggleRule}. 

\begin{algorithm}
\textbf{Input:} $S$

\textbf{Output:} $\mathcal{P}=\{P\in\{\mathtt{0},\mathtt{1}\}^{n}\:|\: F_{r}(P)=S\}$
\begin{algor}[1]
\item [{.}] $\mathcal{P}\gets\emptyset$
\item [{.}] $R\gets\mathtt{0}^{n+2}$
\item [{for}] $\eta\in\{(\mathtt{0},\mathtt{0}),(\mathtt{0},\mathtt{1}),(\mathtt{1},\mathtt{0}),(\mathtt{1},\mathtt{1})\}$

\begin{algor}[1]
\item [{.}] $R_{n+1,n+2}\gets\eta$
\item [{for}] $i=n$ downto $1$

\begin{algor}[1]
\item [{.}] $R_{i}\gets S_{i}\oplus g(R_{i+1},R_{i+2})$
\end{algor}
\item [{endfor}]~
\item [{if}] $R_{1}=R_{n+1}\wedge g(R_{1},R_{2})=g(R_{n+1},R_{n+2})$

\begin{algor}[1]
\item [{.}] $\mathcal{P}\gets\mathcal{P}\cup R_{2:n+1}$
\end{algor}
\item [{endif}]~
\end{algor}
\item [{endfor}]~
\item [{.}] \textbf{return} $\mathcal{P}$
\end{algor}
\caption{\label{alg:InvertToggleRule}\textsc{InvertToggleRule}}
\end{algorithm}

\nameref{alg:InvertToggleRule} runs in time $\Theta(n)$, as controlled
by line 6 which executes $n$ times for each of $4$ possible values
of $\eta$. We must show that $\mathcal{P}$ is exactly the set of
possible prior state vectors leading to $S$ under $F_{r}$. 

Suppose $\mathcal{P}\neq\emptyset$ and let $P\in\mathcal{P}$. Notice
that $P$ is the inner $n$ bits of the full $R$ computed by line
6, and so $P_{1}=R_{2}$ and $P_{n}=R_{n+1}$. We show that for all
$P$, $F_{r}(P)=S$. It is clear that for any $P\in\mathcal{P}$,
$S_{i}=P_{i-1}\oplus g(P_{i},P_{i+1})$ holds for $2\leq i<n$ since
$R_{i}=P_{i-1}$ is assigned to be $S_{i}\oplus g(P_{i},P_{i+1})$
based on the left-toggle property of $f$. That leaves us to verify
just the boundaries. Since $P\in\mathcal{P}$, the condition on line
7 must hold, and so we know that $R_{n+1}=R_{1}$. From line 6, we
have $R_{1}=S_{1}\oplus g(R_{2},R_{3})$, so $R_{n+1}=S_{1}\oplus g(R_{2},R_{3})$,
and therefore $P_{n}=S_{1}\oplus g(R_{2},R_{3})$. XORing $g(R_{2},R_{3})$
to both sides gives $S_{1}=P_{n}\oplus g(P_{1},P_{2})$, which is
the definition of $f_{1}$ with cyclic boundaries. Also from line
6, we have $R_{n}=S_{n}\oplus g(R_{n+1},R_{n+2})$. By line 6 we know
$R_{n}=S_{n}\oplus g(R_{1},R_{2})$ and by line 7, this becomes $P_{n-1}=S_{n}\oplus g(P_{n},P_{1})$
(i.e. the left toggle of $f_{n}$) since $R_{1}=R_{n+1}=P_{n}$. Therefore,
$S_{n}=P_{n-1}\oplus g(P_{n},P_{1})$, the definition of $f_{n}$.
This proves $S_{i}=P_{i-1}\oplus g(P_{i},P_{i+1})$ for $1\leq i\leq n$,
so $S=F_{r}(P)$ for any $P\in\mathcal{P}$.

Suppose that there exists $Q\in\{\mathtt{0},\mathtt{1}\}^{n}$ such
that $F_{r}(Q)=S$. We can construct $Q^{\prime}=Q_{n}Q_{1}Q_{2}\ldots Q_{n}Q_{1}$
in which $Q_{1}^{\prime}=Q_{n+1}^{\prime}$ and $Q_{2}^{\prime}=Q_{n+2}^{\prime}$
so $g(Q_{1}^{\prime},Q_{2}^{\prime})=g(Q_{n+1}^{\prime},Q_{n+2}^{\prime})$.
Notice that, by the definition of $f_{i}$, $Q_{i-1}=S_{i}\oplus g(Q_{i},Q_{i+1})$
for $2\leq i<n$ and also that $Q_{n}=S_{1}\oplus g(Q_{1},Q_{2})$
and $Q_{n-1}=S_{n}\oplus g(Q_{n},Q_{1})$. Then by our construction,
$Q_{i}^{\prime}=S_{i}\oplus(Q_{i+1}^{\prime}+Q_{i+2}^{\prime})$ must
hold for all $1\leq i\leq n$. Thus $Q^{\prime}$ has exactly the
form of some $R$ generated on line 6 since any such $Q^{\prime}$
must end in one of the possible $\eta$ considered on line 3. Further,
$Q^{\prime}$ meets the criteria on line 7 of the algorithm and therefore
$Q$ must be in $\mathcal{P}$. 

It is easy to see that reversing the direction of the algorithm proves
the case of a right-toggle $f_{i}(S_{i-1:i+1})=g(S_{i-1},S_{i})\oplus S_{i+1}$. 
\end{proof}
\algref{InvertToggleRule} inverts all rules in \tabref{FixedPositions}
except rule 101 with only rules 135 and 149 requiring a slight modification
to handle negation of the toggle input. The algorithm is successful
whenever a given state vector has at least one predecessor. This would
seem to be the majority of state vectors: under rule 30 for $n=6$,
there are 12 state vectors with no predecessor, 41 states with exactly
1, 10 with 2, and 1 with 3. For $n=9$, these numbers are 57 with
no predecessor, 393 with 1, 61 with 2, and 1 with 3.

Another implication of this algorithm is that a state vector with
no predecessor discovered mid-stream in a pseudorandom generator must
be the result of re-seeding the generator or updating its entropy.
Leaking this information may be more damaging than knowledge of the
internal state itself: if entropy timing and values can be discerned,
they may become controllable.

Finally, this algorithm improves the running time to invert a rule
30 CA given by Koç/Apohan from $O(2^{n/2})$ to $\Theta(n)$. We are
not aware of any previous bound on the running time to invert other
toggle rules as presented above.

\subsection{\label{subb:ImproveMeierStaffel}Improvements to Meier and Staffelbach}

The Meier/Staffelbach algorithm can be used to recover the CA state
$S^{(0)}$ given a temporal sequence of at least $\left\lceil n/2\right\rceil $
time steps starting at time $t=0$ for the central cell $s_{c}$ where
$c=\left\lceil n/2\right\rceil $. The algorithm can be summarized
as:
\begin{itemize}
\item Guess $\left\lfloor n/2\right\rfloor $ bits for cells $S_{c+1:n}^{(0)}$
for time $t=0.$
\item Evaluate the cells $S_{c+1:n-t}^{(t)}$ for $1\leq t<c$ time steps.
Each time step, one fewer cell on the right end can be computed since
the boundary neighbor value is not known. This leaves only the central
cell known at time $\left\lceil n/2\right\rceil $. Plotted as a two-dimensional
chart over time, the computed cells now form a triangle between the
central cell's values and the right half of the initial state.
\item Using the computed right-adjacent sequence $S_{c+1}^{(0)},\ldots,S_{c+1}^{(\left\lfloor n/2\right\rfloor )}$,
solve the left triangle from $t=\left\lfloor n/2\right\rfloor $ backward
up to $t=0$ to complete the seed.
\end{itemize}
\begin{table}
\begin{centering}
\begin{tabular}{|c|c|c|c|c|c|c|c|}
\hline 
\emph{t} & \emph{c}-3 & \emph{c}-2 & \emph{c}-1 & \emph{\ \,c\,\ } & \emph{c}+1 & \emph{c}+2 & \emph{c}+3\tabularnewline
\hline 
\hline 
0 & \texttt{\textcolor{red}{1}} & \texttt{\textcolor{red}{1}} & \texttt{\textcolor{red}{0}} & \texttt{0} & \texttt{\textcolor{green}{1}} & \texttt{\textcolor{green}{1}} & \texttt{\textcolor{green}{0}}\tabularnewline
\hline 
1 &  & \texttt{\textcolor{magenta}{0}} & \texttt{\textcolor{magenta}{1}} & \texttt{1} & \texttt{\textcolor{blue}{1}} & \texttt{\textcolor{blue}{0}} & \tabularnewline
\hline 
2 &  &  & \texttt{\textcolor{magenta}{1}} & \texttt{0} & \texttt{\textcolor{blue}{0}} &  & \tabularnewline
\hline 
3 &  &  &  & \texttt{1} &  &  & \tabularnewline
\end{tabular}
\par\end{centering}

\textsf{\caption{\label{tab:ExSolving}\textsf{Example of applying the Meier-Staffelbach
algorithm in rule 30. Step 1: The seed values on the right half (in
green) are guessed. Step 2: The right triangle (in blue) is computed.
Step 3: The left triangle (in magenta) is solved to complete the left
half of the seed (in red).}}
}
\end{table}

This algorithm is illustrated in Table \ref{tab:ExSolving}. Meier
and Staffelbach show for $n=300$, the center temporal sequence of
a uniform rule 30 CA requires about 18 bits of entropy to guess a
compatible seed. We would like to see if taking advantage of the observations
made above allows us to improve this result.

Meier and Staffelbach note in \cite{Meier91} that where the temporal
sequence is a sequence of \texttt{0}s, the right-adjacent sequence
must match $\mathtt{0}^{*}\mathtt{1}^{*}$. We make a related observation
on the left-adjacent sequence.
\begin{prop}
\label{prop:3}Any occurrence of \texttt{10} in the temporal sequence
starting at time $t$ in a uniform rule 30 CA must have a left-adjacent
value of \texttt{1} at $t$; any occurrence of \texttt{11} must have
a left-adjacent value of \texttt{0} at $t.$ \end{prop}
\begin{proof}
Let $i$ be the index of the cell for which the temporal sequence
is known. By the definition of rule 30, 
\[
s_{i-1}^{(t)}=s_{i}^{(t+1)}\oplus(s_{i}^{(t)}+s_{i+1}^{(t)})
\]
and we know that $s_{i}^{(t)}=1$. Then $s_{i}^{(t)}+s_{i+1}^{(t)}=1$
and so 
\[
s_{i-1}^{(t)}=s_{i}^{(t+1)}\oplus1=\overline{s_{i}^{(t+1)}}.
\]
 When \texttt{10} occurs in the temporal sequence at cell $i$, $s_{i}^{(t+1)}=\mathtt{0}$
and so the left-adjacent $s_{i-1}^{(t)}$ must be \texttt{1}. When
\texttt{11} occurs, $s_{i-1}^{(t)}$ must be \texttt{0}.
\end{proof}
Further, if two such occurrences are temporally adjacent (as must
always be the case when \texttt{11} appears since the third bit is
either \texttt{1} or \texttt{0}) starting at time $t$ for a sequence
in cell $i$, then two left-adjacent cells are known and can be solved
backwards to produce the value of $s_{i-2}^{(t)}$. Table \ref{tab:SolvedNeighbors}
gives an example temporal sequence at cell $i$ and the solved values
in red. In addition, the known predecessor patterns can also be used
to solve even more of the CA history. By using the prior states from
\tabref{rule30Priors}, cells in green have been filled in as well.

\begin{center}
\begin{table}
\begin{centering}
\begin{tabular}{|c|c|c|c|c|c|}
\hline 
\emph{t} & \emph{i}-4 & \emph{i}-3 & \emph{i}-2 & \emph{i}-1 & \emph{i}\tabularnewline
\hline 
\hline 
0 &  &  &  &  & \texttt{0}\tabularnewline
\hline 
1 &  &  &  & \texttt{\textcolor{red}{1}} & \texttt{1}\tabularnewline
\hline 
2 & \texttt{\textcolor{green}{0}} &  &  &  & \texttt{0}\tabularnewline
\hline 
3 &  & \texttt{\textcolor{red}{1}} & \texttt{\textcolor{red}{1}} & \texttt{\textcolor{red}{0}} & \texttt{1}\tabularnewline
\hline 
4 &  &  & \texttt{\textcolor{red}{0}} & \texttt{\textcolor{red}{0}} & \texttt{1}\tabularnewline
\hline 
5 &  &  &  & \texttt{\textcolor{red}{1}} & \texttt{1}\tabularnewline
\hline 
6 &  &  &  &  & \texttt{0}\tabularnewline
\hline 
7 &  &  &  &  & \texttt{0}\tabularnewline
\hline 
8 &  &  &  &  & \texttt{1}\tabularnewline
\hline 
\end{tabular}
\par\end{centering}

\textsf{\caption{\label{tab:SolvedNeighbors}\textsf{Example of solved neighbors in
rule 30}}
}
\end{table}

\par\end{center}

There are a couple of applications of these observations to improving
the Meier/Staffelbach algorithm. First, let $\sigma$ be the temporal
sequence of length $n$ from cell $i$ in a rule 30 CA beginning at
$t=0$ and let $j$ be the number of 1s in $\sigma$. Then no more
than $n-j-1$ random bits, or coins, would be necessary to chose those
unknown left neighbors that would allow us to solve the entire left
triangle back to $t=0$ and thus recover the entire seed. If $\sigma$
is evenly distributed, then we would expect to need $n/2$ coins.
This is comparable to the Meier-Staffelbach algorithm. If, however,
we spend those coins to fill in missing values of $s_{i-1}$ from
the bottom up, we may reach a point were examining predecessors of
varying widths in different positions fills in prior time steps deterministically.
This leads us to the following algorithm:
\begin{enumerate}
\item Rotate the CA cells to the right, placing the known temporal sequence
on the right edge, leaving only a left triangle of size $n$ to solve. 
\item Repeat until a full state vector is known or no changes are possible:

\begin{enumerate}
\item %
Recursively apply \propref{3} to fill in left-adjacent cells where
$\sigma_{i}=1$.
\item Recursively apply \propref{2} to fill in fixed predecessor values.
\item Run the CA forward to fill in any cell for which all three inputs
are known.
\item For any $s_{i-1}^{(t)}=1$ and $s_{i}^{(t+1)}=1$, set $s_{i}^{(t)}=s_{i+1}^{(t)}=0$,
accounting for boundaries.
\end{enumerate}
\item If the full state vector is known:

\begin{enumerate}
\item Run the CA forward to fill in unknown cells. Check all known cells
for correctness. 
\item If discrepancies arise: 

\begin{enumerate}
\item Backtrack to revisit random bits chosen in step 4 and flip them.
\end{enumerate}
\item Otherwise, use \algref{InvertToggleRule} to invert the CA back to
$t=0$.
\end{enumerate}
\item Otherwise, if no changes are possible, choose a random bit for the
bottom-most empty cell adjacent to the left edge of known cell values.
\item Return to step 2.
\end{enumerate}
Pencil and paper experimentation show this algorithm to hold promise,
but a more complete effort is required to determine limits on the
number of random bit selections needed.

Table \ref{tab:SolvedWithPreds} illustrates the use of coins (in
blue) and predecessor information to complete the previous example.
Once the state contains a 010 string, all predecessors to the left
of that string are known for as long and the successor is known to
the left. This means that coins are only required where the temporal
sequence cannot fill in those positions to the right until the boundary
conditions are known. Past that point, these holes can also be filled
in by solving forward with knowledge of the right-adjacent values.
Cells colored in magenta show cells which can be determined once the
boundary conditions are known.

\begin{center}
\begin{table}
\begin{centering}
\begin{tabular}{|c|c|c|c|c|c|c|c|c|c|c|}
\hline 
\emph{t} & \emph{i}-8 & \emph{i}-7 & \emph{i}-6 & \emph{i}-5 & \emph{i}-4 & \emph{i}-3 & \emph{i}-2 & \emph{i}-1 & \emph{i} & \emph{i}-8\tabularnewline
\hline 
\hline 
0 & \texttt{\textcolor{green}{0}} & \texttt{\textcolor{magenta}{0}} & \texttt{\textcolor{magenta}{1}} & \texttt{\textcolor{magenta}{1}} & \texttt{\textcolor{magenta}{1}} & \texttt{\textcolor{magenta}{1}} & \texttt{\textcolor{magenta}{0}} & \texttt{\textcolor{magenta}{1}} & \texttt{0} & \texttt{\textcolor{green}{0}}\tabularnewline
\hline 
1 &  & \texttt{\textcolor{green}{1}} & \texttt{\textcolor{red}{1}} & \texttt{\textcolor{red}{0}} & \texttt{\textcolor{red}{0}} & \texttt{\textcolor{green}{0}} & \texttt{\textcolor{red}{0}} & \texttt{\textcolor{red}{1}} & \texttt{1} & \tabularnewline
\hline 
2 &  &  & \texttt{\textcolor{green}{0}} & \texttt{\textcolor{green}{1}} & \texttt{\textcolor{green}{0}} & \texttt{\textcolor{green}{0}} & \texttt{\textcolor{green}{1}} & \texttt{\textcolor{blue}{1}} & \texttt{0} & \tabularnewline
\hline 
3 &  &  &  & \texttt{\textcolor{red}{1}} & \texttt{\textcolor{red}{1}} & \texttt{\textcolor{red}{1}} & \texttt{\textcolor{red}{1}} & \texttt{\textcolor{red}{0}} & \texttt{1} & \tabularnewline
\hline 
4 &  &  &  &  & \texttt{\textcolor{red}{0}} & \texttt{\textcolor{red}{0}} & \texttt{\textcolor{red}{0}} & \texttt{\textcolor{red}{0}} & \texttt{1} & \tabularnewline
\hline 
5 &  &  &  &  &  & \texttt{\textcolor{red}{0}} & \texttt{\textcolor{red}{1}} & \texttt{\textcolor{red}{1}} & \texttt{1} & \tabularnewline
\hline 
6 &  &  &  &  &  &  & \texttt{\textcolor{red}{1}} & \texttt{\textcolor{blue}{1}} & \texttt{0} & \tabularnewline
\hline 
7 &  &  &  &  &  &  &  & \texttt{\textcolor{blue}{0}} & \texttt{0} & \tabularnewline
\hline 
8 &  &  &  &  &  &  &  &  & \texttt{1} & \tabularnewline
\hline 
\end{tabular}
\par\end{centering}

\textsf{\caption{\textsf{\footnotesize \label{tab:SolvedWithPreds}}\textsf{Example
of solved neighbors in rule 30. In this example, the sequence 111101
in $s^{(3)}$ with 0 at $s_{i}^{(2)}$ doesn't quite determine that
$s^{(2)}$ so spending a coin on $s_{i-1}^{(2)}$ is still required.
In other cases, a 010 may appear in $s^{(2)}$, saving the coin. All
of $s^{(1)}$ is determined at the right edge and by $s^{(2)}$.)
Since $s^{(1)}$ begins with 110, $s_{i-8}^{(0)}=0$. Suppose $n=9$;
then $s_{i}^{(0)}$ is determined and so on to the left.}}
}
\end{table}

\par\end{center}

\section{Analysis of Linear and Affine CAs\label{sec:Analysis-of-Linear-CAs}}

Over and again, the literature shows interesting results around rules
90, 105, 150, and 165 \cite{Hortensius89,Guan03,GuanT04,Nandi94,Sipper96,Seredynski04,Tomassini01,Tomassini99}.
A cryptanalyst with a firm command of constructions based on these
rules stands to gain good advantage over the majority of random number
generators and cryptosystems based on these CAs. Therefore, we would
like to understand the extent to which we can apply the same techniques
to the case of linear and affine CAs based on these four rules. 

There are some obvious challenges to using linear CA for cryptography.
As previously mentioned, knowledge of the complete state vector in
a linear CA makes its entire history solvable---simply solve the linear
equations backwards for the number of time steps desired. If the rule
vector is known, using the full state vector seems unwise. Using only
a temporal sequence would keep knowledge of the full state vector
secret, but in linear CA, Meier/Staffelbach approach seems only to
get easier as we now have full linearity in both directions.

We will focus on the approach taken by Tomassini and Perrenoud \cite{Tomassini01}
to deal with these problems. Their scheme is as follows: Select via
the key one of these four rules (details of how to do this are not
given) as well as the initial state for each of $n$ cells in a cyclic
boundary CA. Use the central temporal sequence as a random stream
and XOR it with the plaintext. The key space would be $4^{n}\times2^{n}=2^{3n}$
making the key $3n$ bits with suggested values of $n$ around 100
(for 2001 compute power). It is claimed that $2^{(5n-9)/2}$ guesses
over rules and values would be necessary to solve the sequence backwards
to find a single, unique rule set and seed that produces that sequence.
This is one bound we seek to improve.

\subsection{Analysis of CAs with Symmetric Rule Sets}

Much of the key space proposed in \cite{Tomassini01} will result
in CA whose rule sets are symmetric about a center cell. Recall that
a CA is symmetric if rule $f_{i}=f_{n-i+1}$ for all $1\leq i\leq n$.
Experiments performed on these configurations show very low periods,
making them very inefficient. We capture this observation in the following
conjecture.
\begin{conjecture}
\label{conj:n/2-Seeds}Let $n$ be an odd number and $\Sigma$ be
an $n$-cell symmetric hybrid linear CA over rules $R={90,105,150,165}$
with cyclic boundaries. There are at most $2^{\left\lfloor n/2\right\rfloor }$
initial states of $\Sigma$ that give the same temporal sequence of
length $\ell\geq\left\lceil n/2\right\rceil $ at cell $\left\lceil n/2\right\rceil $.
\end{conjecture}
Evidence for this conjecture begins with the following proposition. 
\begin{prop}
There are $2^{\left\lfloor n/2\right\rfloor }$ initial states of
$\Sigma$ that give the same temporal sequence of length $\left\lceil n/2\right\rceil $
at cell $\left\lceil n/2\right\rceil $.\end{prop}
\begin{proof}
Let $k=\left\lceil n/2\right\rceil $ and let $\sigma=s_{k}^{(0)},\ldots,s_{k}^{(k)}$
be the first $k$ bits of the center temporal sequence at cell $k$.
Choose any $\rho\in\{0,1\}^{k-1}$ as the right-adjacent sequence
of cell $k$ beginning at $t=0$. Since each transition function in
$\Sigma$ is affine, both the right and left triangles of $\Sigma$
are easily solved for, resulting in a full initial state. This state
necessarily produces $\sigma$, regardless of the choice of $\rho$.
\end{proof}
This shows that the bound on the number of initial states holds at
$\ell=\left\lceil n/2\right\rceil $, giving a sort of base case for
the argument.

Further evidence for \conjref{n/2-Seeds} can be seen in the experimental
results in \tabref{MaxPeriods}, and suggest that the period of any
linear CA (symmetric or asymmetric) is no more than $2^{\left\lceil n/2\right\rceil }$.
If this is true and some period $p$ can be derived from the rule
set, then a proof need only deal with values of $k<\ell<p$. For $t\geq p$,
the sequence must be fixed.

\begin{table}
\begin{centering}
\begin{tabular}{|c|c|c|}
\hline 
$n$ & Max Period & $2^{\left\lceil n/2\right\rceil }$\tabularnewline
\hline 
\hline 
5 & 8 & 8\tabularnewline
\hline 
7 & 14 & 16\tabularnewline
\hline 
9 & 30 & 32\tabularnewline
\hline 
11 & 62 & 64\tabularnewline
\hline 
51 & 67108860 & 67108864\tabularnewline
\hline 
\end{tabular}
\par\end{centering}

\textsf{\caption{\label{tab:MaxPeriods}\textsf{Maximum periods of linear cyclic CA.
The value for $n=51$ is based on a single observation of $F=\langle\{150\}^{12},90,\{150\}^{25},90,\{150\}^{12}\rangle$
with $s^{(0)}=562964991182857$. A few other rule sets have been tried
with only a few other seeds, all having far lower periods. This is
a symmetric CA, and larger periods may be possible for asymmetric
CA. Lower values for $n$ are the results from exhaustive computation.}}
}
\end{table}

A likely approach to proving \conjref{n/2-Seeds} uses a linear system
in GF(2) to model the evolution of the CA. Since all rules in $R$
are affine, we can represent $\Sigma$'s transition function across
all cells as 
\begin{equation}
\begin{bmatrix}w_{1,2} & w_{1,3} & 0 & 0 &  & w_{1,1}\\
w_{2,1} & w_{2,2} & w_{2,3} & 0 &  & 0\\
0 & w_{3,1} & w_{3,2} & w_{3,3} &  & 0\\
0 & 0 & w_{4,1} & w_{4,2} & \cdots & 0\\
 &  &  & \vdots & \ddots\\
0 & 0 & 0 & w_{n-1,1} & w_{n-1,2} & w_{n-1,3}\\
w_{n,3} & 0 & 0 & 0 & w_{n,1} & w_{n,2}
\end{bmatrix}\begin{bmatrix}s_{1}^{(t)}\\
s_{2}^{(t)}\\
\\
\vdots\\
\\
\\
s_{n}^{(t)}
\end{bmatrix}+\begin{bmatrix}b_{1}\\
b_{2}\\
\\
\vdots\\
\\
\\
b_{n}
\end{bmatrix}=\begin{bmatrix}s_{1}^{(t+1)}\\
s_{2}^{(t+1)}\\
\\
\vdots\\
\\
\\
s_{n}^{(t+1)}
\end{bmatrix}\label{eq:LinearSystem}
\end{equation}

where $w_{i,j}$ is the weight for cell $i$ on input $j$ in the
cells transition function $f_{i}(x_{1},x_{2},x_{3})$ and $b_{i}$
is an offset for each cell to affect a complement rule. If $M$ is
the matrix of $w_{i,j}$, then the recurrence $S^{(t+1)}=MS^{(t)}+b$
models the evolution of $\Sigma$. If this recurrence is periodic,
so is $\Sigma$.

The matrix $M$ has a few interesting properties in the cases of interest.
First, for all rules in $R$, $w_{i,1}=w_{i,3}=1$. Therefore, $M$
is symmetric. The values on the diagonal are determined by whether
$f_{i}\in\{90,165\}$(where $w_{i,2}=0$) or $f_{i}\in\{150,105\}$
($w_{i,2}=1)$. Second, since $\Sigma$ is symmetric, $M$ is symmetric
relative to both diagonals, i.e. $w_{i,2}=w_{n-i+1,2}$ for $i<n/2$.
Lastly, $M$ is a band matrix%
\footnote{\textsf{\footnotesize The strict definition of band matrix precludes
the non-zero values at $(1,n)$ and $(n,1)$, but some authors refer
to matrices of this form as band matrices in the literature. To avoid
this technicality, we can simply duplicate $s_{1}$ and $s_{n}$ on
either side, append 0's to $b$, and increase the matrix dimension
to match.}%
}. 

If the recurrence is periodic, then applying the recurrence some number
of times must yield the starting value. That is 
\[
M(M(\ldots(M(MS^{(t)}+b)+b)\ldots)+b)+b=MS^{(t)}+b.
\]
Showing that repeating the recurrence leads to an earlier result of
the recurrence rather than just the original $S^{(t)}$ allows for
proper handling of garden of Eden state vectors. Collecting terms
over $p$ applications of the recurrence, we see that the period is
$p$ when 
\[
M^{p+1}S^{(t)}+M^{p}b+\ldots+Mb+b=MS^{(t)}+b
\]
This may be further broken down to showing that raising a matrix with
these properties to the power $p$ is idempotent, i.e. $M^{p+1}=M$,
followed by showing that $M^{p}b+\ldots+Mb=\mathbf{0}$. This last
part can be re-written as 
\begin{equation}
\left(\sum_{i=1}^{p}M^{i}\right)b=\mathbf{0}.\label{eq:SumMi=00003D0}
\end{equation}

In fact, the condition in (\ref{eq:SumMi=00003D0}) also implies $M^{p+1}=M$.
Suppose $Q(x)$ is the characteristic polynomial of $M$ so that $Q(M)=0.$
So also $M\cdot Q(M)=0$ and $M^{2}\cdot Q(M)+M\cdot Q(M)=\mathbf{0}$.
Since $Q$ is over GF(2), the coefficients of $Q$ can be written
as a binary vector $c$ of length $n$. Multiplying $M\cdot Q(M)$,
as each iteration of the recurrence does, gives a polynomial with
coefficients $c\ll1$ which also evaluates to $\mathbf{0}$ at $M$.
Here, $\ll$ denotes the left-shift operator. We can likewise sum
any number of shifts of $c$ and arrive at a polynomial which evaluates
to $\mathbf{0}$. 

Now suppose that through this process, we create a polynomial with
a coefficient vector $c'=1^{p}$. We conjecture it is always possible
to create such a polynomial if $Q(x)$ has terms $1x^{1}+0x^{0}$.
If $Q(x)$ contains the term $1x^{0}$, we simply start from $c\ll1$.
We know the polynomial represented by $c'$ evaluates to $\mathbf{0}$
at $M$. Therefore, this polynomial satisfies \ref{eq:SumMi=00003D0}.
To check whether $M^{p+1}=M$, we can see 

\begin{eqnarray*}
\sum_{i=1}^{p}M^{i} & = & \mathbf{0}\\
(M+I)\sum_{i=1}^{p}M^{i} & = & \mathbf{0}\\
\sum_{i=2}^{p+1}M^{i}+\sum_{i=1}^{p}M^{i} & = & \mathbf{0}\\
M^{p+1} & = & M
\end{eqnarray*}

The reverse implication also holds as long as $M+I$ is invertible.
An algorithm to build such a $c'$ exists (by repeatedly shifting
the lowest order 1-bit up to the lowest order 0-bit in the running
sum over GF(2)), but there is no proof of a bound on the resulting
$p$. Recall that Nandi, et. al. prove the characteristic polynomial
for a cyclic linear CA must have $x$ or $(x+1)$ as a factor, but
it is not clear how this leads to a bound of $2^{n/2}$.

As an example, we illustrate the case for a 5-cell uniform rule 90
CA, which has the following transition matrix:
\[
M_{90}=\begin{bmatrix}0 & 1 & 0 & 0 & 1\\
1 & 0 & 1 & 0 & 0\\
0 & 1 & 0 & 1 & 0\\
0 & 0 & 1 & 0 & 1\\
1 & 0 & 0 & 1 & 0
\end{bmatrix}
\]
 with successive powers
\[
M_{90}^{2}=\begin{bmatrix}0 & 0 & 1 & 1 & 0\\
0 & 0 & 0 & 1 & 1\\
1 & 0 & 0 & 0 & 1\\
1 & 1 & 0 & 0 & 0\\
0 & 1 & 1 & 0 & 0
\end{bmatrix},\; M_{90}^{3}=\begin{bmatrix}0 & 1 & 1 & 1 & 1\\
1 & 0 & 1 & 1 & 1\\
1 & 1 & 0 & 1 & 1\\
1 & 1 & 1 & 0 & 1\\
1 & 1 & 1 & 1 & 0
\end{bmatrix},\; M_{90}^{4}=\begin{bmatrix}0 & 1 & 0 & 0 & 1\\
1 & 0 & 1 & 0 & 0\\
0 & 1 & 0 & 1 & 0\\
0 & 0 & 1 & 0 & 1\\
1 & 0 & 0 & 1 & 0
\end{bmatrix}=M_{90}.
\]
Notice in this case that 
\[
\sum_{i=1}^{3}M_{90}^{i}=\mathbf{0}
\]
 so even if $b$ were not $\mathbf{0}$, the period would still be
3. Thus any 5-cell hybrid CA over rules in $\{90,165\}$ has period
3 for any starting seed.
\begin{rem}
There are $2^{n}$ symmetric CAs. To see this, we choose $n/2$ bits
for the diagonal of $M$ to select between binary ternary rules at
each cell, then reflect them to the other half. We then choose $n/2$
bits for the vector $b$ to select between normal and complementary
rules and reflect those bits as well. These $n$ bits cover the range
of all possible values of $M$ and $b$ for symmetric CAs.
\end{rem}
Another possible approach to proving \conjref{n/2-Seeds} is to think
of a symmetric CA as a left block of $k-1$ cells, a center cell,
and a right block of $k-1$ cells. We notice the left and right blocks
are equivalent in construction: they have inner and outer neighbors
that apply the same function, and they deliver their outputs to inner
and outer neighbors after applying the same function as inputs move
from one side to the other. (See \figref{SymmetricCA}.) Their only
difference is their initial state and the order in which the initial
state vector $s^{(0)}$ is acted upon. This model may explain the
reflection seen in cell values over several steps. Further, notice
that at $t=k$, each cell is a function of all cells' initial state.
In particular, each block is a function of its initial state, the
first $k-1$ bits of $\sigma$, and the output of the opposite block.
With the functions so similar, the intuition is that the period must
be small (certainly no more than $(k-1)^{2}$), but a proof is not
known. 

Other considerations may further reduce the period. Spatially symmetric
seeds can generate periods of no greater than $2^{\left\lfloor n/2\right\rfloor }$
since both blocks will be exactly the same. It may also be of interest
to know what the spacial period is in the case of symmetric seeds;
some combinations dead-end at 0, some have very low periods. McIntosh
covers the effects of spatially periodic seeds in depth in \cite{Mcintosh87}.

\begin{figure}
\begin{centering}
\includegraphics[scale=0.9]{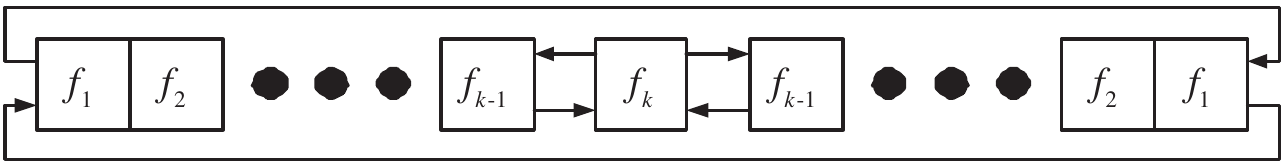}
\par\end{centering}

\textsf{\caption{\label{fig:SymmetricCA}\textsf{Symmetric CA viewed as having right-
and left-side blocks}}
}
\end{figure}

\begin{rem}
Assuming \conjref{n/2-Seeds}, some sequences of length $\ell>\left\lceil n/2\right\rceil $
are not possible in a given symmetric CA. This is supported by a pigeon-hole
argument: If each sequence is produced by $2^{\left\lfloor n/2\right\rfloor }$seeds
then all seeds lead to one of $2^{\left\lceil n/2\right\rceil }$
sequences, leaving $2^{\left\lfloor n/2\right\rfloor }$ sequences
with no possible seeds. This is born out by experimentation. See \apnref{SymEffects}
and \apnref{CorrSeqSeeds}.
\end{rem}

\begin{rem}
Assuming \conjref{n/2-Seeds}, sequences which are not eventually
periodic with eventual period $\leq2^{\left\lceil n/2\right\rceil }$
cannot be produced by a symmetric CA. If such a CA were used as a
key stream, a single known plaintext of more than $2^{\left\lceil n/2\right\rceil }$
bits would quickly be distinguishable from random.
\end{rem}
In summary, using a symmetric CA does not seem advisable. Suppose
$n$ is large enough to provide for an adequate period. There are
still many seeds that would allow an attacker to compute past and
future bits, defeating the whole system. Excluding symmetric CAs (both
in arity and complementarity) reduces the rule vector space by $2^{n}$.
To guess the rule vector and initial state for only asymmetric CAs
given a temporal sequence of length $n/2$, we would first need to
choose from $2^{2n}-2^{n}=2^{n}(2^{n}-1)$ rule vectors. Then we would
need to choose one bit per time step for $(n-1)/2$ time steps to
allow solving the state vectors backwards once the rule set is fixed.
Multiplying these gives us a bound of $2^{(3n-1)/2}(2^{n}-1)$ guesses
for the whole key space, reducing the bound in \cite{Tomassini01}
by $2^{(3n-1)/2}$ trials.

\subsection{Analysis of CAs with Asymmetric Rule Sets}

Asymmetric CAs prove more difficult to recover a seed from. Many seem
to have unique seeds for each temporal sequence and periods generally
seem longer in asymmetric CA, though these statements are not true
in all cases. Deriving any identities concerning periods, number of
distinct sequences, relation to other rule sets, etc. has proved very
difficult.

To explore these constructs further, we instead examine data on particular
qualities. As documented in \tabref{MaxPeriods}, exhaustive computation
gives us the maximum eventual period for all $n$-cell CAs for $n\leq11$.
Taking a 128-ruleset sample from the data from $n=9,$ we can chart
the number of CA having a certain maximum period which generate a
certain number of distinct sequences. The rule sets were chosen by
looking at all possible rule sets over rules 90 and 150. This seems
justified since the complementing of terms in rules 105 and 165 has
been shown to add little if any value at all in most experiments%
\footnote{\textsf{\footnotesize But not none. There are cases where simply complementing
certain rules can increase a rule set\textquoteright{}s period}%
}. Fixing the two left rules at 150 gives 128 possible rule sets to
examine. For each of these, its maximum period and the total number
of distinct sequences it generates over all possible seeds were computed.
The results are shown in the bubble chart in \figref{NumberOfPerSeq}
and raw data for the table is in \apnref{NumPerSeq}.

\begin{figure}
\begin{centering}
\includegraphics[bb=45bp 65bp 745bp 550bp,clip,height=3.5in]{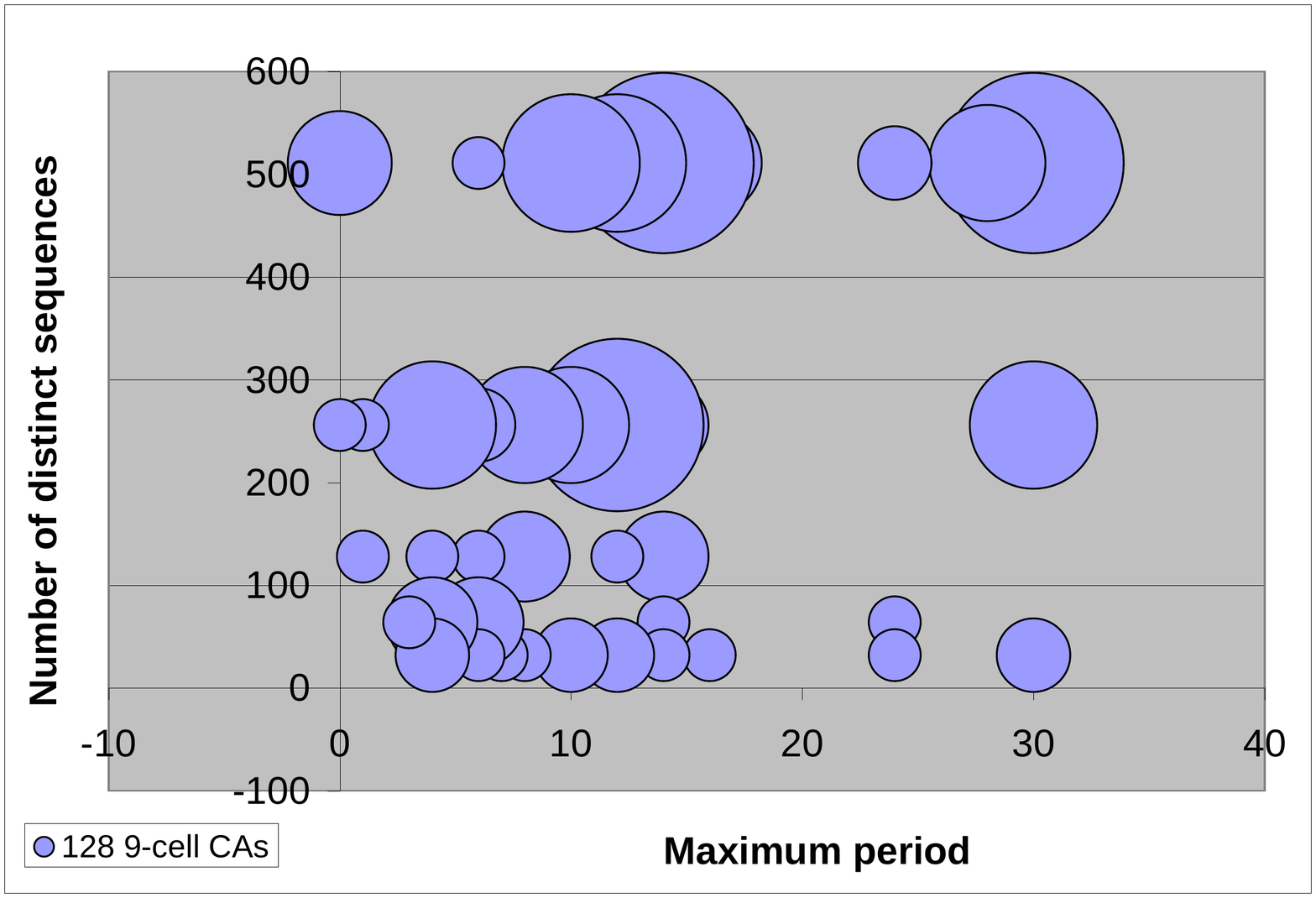}
\par\end{centering}

\textsf{\small \caption{\textsf{Number (indicated by bubble area, maximum of 12) of Asymmetric
CA rule sets having a given period and generating a number of distinct
sequences\label{fig:NumberOfPerSeq}}}
}
\end{figure}

Some interesting observations are possible looking at this data. First,
41\% of the rule sets sampled generate a distinct central temporal
sequence for every possible seed. Next, $29/128$ have a period greater
than or equal to 24. And finally, $19/128$ have both of these properties.
These numbers represent the maximums for the rule sets; many seeds
under a given rule set may do worse than the maximum.

This suggests that at least roughly 85\% of the key space over $R\times\{0,1\}^{n}$
has some considerable weakness.

\subsection{\label{subb:CaOpenProblems}Open Problems}

The following sections record open investigations that may hold some
promise for improving cryptanalysis in various CA. The ideas are rather
unpolished, and a time-constrained reader may certainly skip this
section.

\subsubsection{Solving an Arbitrary Temporal Sequence in an Arbitrary Symmetric
CA}

\conjref{n/2-Seeds} suggests the workings for an algorithm analogous
to that of Meier-Staffelbach to recover a seed given a central temporal
sequence generated in a symmetric CA. What's missing is a method to
select rules under which to solve the triangles of the CA backward
(or run forward, for that matter.) Such a method is not obvious since
different symmetric CA have different periods. It may be possible
to solve a part of a sequence under a chosen rule set, then map to
another rule set with a longer period. This new rule set can then
be run forward to check for a match. Mapping between rule sets seems
to be possible, as suggested by experiments in \apnref{SolveTempSeq}.

\subsubsection{\label{subb:MapSeeds}Mapping Seeds from Symmetric to Asymmetric
CA}

Let $A$ be an asymmetric CA that generates a temporal sequence $\alpha$
from seed $s_{A}^{(0)}$. Suppose $\alpha$ is also generated by a
symmetric CA $B$ (which can be tested using the algorithm stated
above) using seed $s_{B}^{(0)}$. We would like to know if there a
mapping from the seed $s_{B}^{(0)}$ to $s_{A}^{(0)}$ or, short of
a mapping, if there is an efficient algorithm for finding $s_{A}^{(0)}$
given $s_{B}^{(0)}$. Such an algorithm would reduce the problem of
recovering a seed in an asymmetric CA to recovering a seed in a symmetric
CA, which we've already solved. 

Experiments in \apnref{MapSeedsExpr} show that each of the $2^{\left\lfloor n/2\right\rfloor }$possible
$s_{B}^{(0)}$ has a unique difference from $s_{A}^{(0)}$. Further,
their difference is the same for $\overline{\alpha}$, which may reduce
the search space by half. Of course, $\alpha$ may not be a possible
temporal sequence of $B$, which may also be useful information.

\subsubsection{Mapping Sequences from Symmetric to Asymmetric CA}

We approach the problem orthogonally to \subbref{MapSeeds}: for all
seeds $s_{B}^{(0)}$ that generate each sequence $\sigma$ in symmetric
CA $B$, examine the sequence $\alpha$ generated by $A$ using the
same seed $s_{B}^{(0)}$. The difference between $\sigma$ and $\alpha$
may provide information about the rule set $F_{A}$. Experiments in
\apnref{MapSeqExpr} show example difference patterns. In the specific
case shown, these differences only occur after the first $\left\lceil n/2\right\rceil $
time steps, which indicates that the asymmetry only affects the temporal
sequence in positions where contributing seed bits have been used
as input to cells symmetric about the sequence, i.e. after wrapping
around the boundary. This suggests that the asymmetry is one sided.
If we think of the information from the initial state as ``flowing''
through the cells on any of their connection paths towards the central
temporal sequence, we notice conditions that might change the information
that arrives in the central cell compared to a uniform or symmetric
CA. For all temporal sequence positions before $t=\left\lceil n/2\right\rceil $,
an asymmetric seed could correct for rule differences on just one
side. But after passing through both symmetric cells, the effect of
asymmetry changes the information contribution to the central cell
when it arrives there. Other pairs of rule sets in this same experiment
show differences earlier than $t=\left\lceil n/2\right\rceil $, and
indeed, have rule changes on both sides and closer to the central
cell. This supports the general idea, though the actual mechanics
are not fully understood.

If an efficient algorithm could be found to identify the actual values
of the asymmetry in the rule set, then recovering the original seed
is reduced to solving a known sequence under known hybrid affine rules,
which can be done with the Meier/Staffelbach algorithm. Recovering
the rule values will likely require comparing the sequences from an
unknown, asymmetric CA to those of uniform rule 90, 105, 150, and
165 CAs.

\section{\label{sec:NewConstruction}A New Cellular Construction}

Given the observations made previously and the body of literature
showing weakness in 3-neighbor, fixed-rule cellular automata (CA),
we may begin to wonder if \emph{any} CA construction is capable of
exhibiting provably NP-hard behavior. The 3-neighbor uniform CA using
rule 110 has been proven to be universal (i.e. capable of simulating
an arbitrary Turing machine) \cite{Cook04}. Unfortunately, this rule
does not seem capable of passing basic statistical tests by itself
since it is highly non-linear and therefore biased in its output.
It seems all 3-neighbor fixed-rule set CAs, uniform or hybrid, are
unsuitable for cryptographic applications. 

Yet the highly parallel nature and simple operation of CA are still
appealing. Whereas block ciphers and other cryptographic primitives
must be specifically designed for a predetermined block width, a strong
CA construction can allow scalability simply by adding more cells.
Since cells are identical, they can be packaged in ASICs or programmable
logic blocks and configured at a width just adequate for the job at
hand. For software implementations, the ubiquity of vector register
operations on most CPUs and the advent of GPUs give CA constructions
an easy path to performance improvements. Most block ciphers, in contrast,
do not naturally decompose into parallel tasks in an obvious way.
Lastly, many cryptographic primitives have no proofs of security properties
or only derive provable security properties through reduction to other
primitives which have no such proofs (see e.g. \cite{Barak05}). The
simple, regular operations of CA, on the other hand, seem more likely
to lend themselves to proving certain properties than ad-hoc combinations
of shifts and XORs. If so, and if the implementation is efficient
and scalable, such a construct would have natural advantages for cryptographic
application designers.

Therefore it seems worth understanding what might be required to produce
hard-to-invert CA. We might first ascribe the observed failings to
the 3-neighbor construct after noting that any two neighbors share
2 of their 3 predecessors in common. This construct leaks information
about a cell's value to its three descendants. With knowledge of the
function applied at each cell, an attacker can build a system of equations,
even if its non-linear in some cells. One way to plug this leak is
to protect the knowledge of the function applied at each cell. Suppose
we fix the neighborhood of cells at 3 but allow the rule applied at
each step to be selected uniformly at random from all possible 3-neighbor
rules and we seed the CA with $n$ bits, also sampled uniformly at
random. It is clear that knowing the full state of the CA at time
$t+1$ would provide no information about the state at time $t$.
Each iteration of such a CA would realize Shannon's notion of \emph{perfect
secrecy} \cite{Shannon49a}. 

Unfortunately, requiring truly random rule selection leaves us no
better off than where we started--in need of good random bits. We
might consider using pseudorandom functions to preserve some notion
of semantic security, that is, perfect secrecy under computational
bounds. Even so, that would leave us requiring one pseudorandom primitive
to produce another. We may then wonder: Can a 1-D CA with enhanced
cells which modulate their rule by a simple process (e.g. a weak random
number generator or a fixed permutation) be provably NP-hard to invert
and still produce cryptographic-strength pseudorandom bits?

\subsection{Finite State Transducers}

To understand the capabilities of cellular constructs, we must first
formalize a computational model to evaluate. There are a few factors
that guide the selection of our model. First, we'd like to find the
simplest model possible which is still capable of the required computation.
This is a general principle but also a practical concern since simple
models are easiest to reason about. Second, we'd like physical implementations
to be able to match the computational model closely. This allows any
provable properties which exist in theory to also be claimed by the
implementations (up to differences required of the mapping to the
physical world.)

The simplest computational model in theory is the finite state machine
(FSM), which is ostensibly the model for each cell. We consider FSMs
with greater than two states and the capabilities of the overall automaton
when we place various restrictions on the construction of the FSMs
in the cells. 

Historically, the output of a two-state cell has been referred to
as its \emph{state} since the output directly reflects the current
state of the automaton. This terminology becomes confusing when we
consider many-state machines which still produce only two outputs.
We will therefore adopt the notion of a cell producing an \emph{output}
as distinct from its current state. This notion is captured nicely
in the model known as finite state transducers (FSTs), which originated
with Mealy \cite{Mealy55} and Moore \cite{Moore56} after whom the
popular variants are named. Using these models provides the advantage
of having well known ways to map such transducers into combinatorial
logic, making practical applications more straightforward. 

Conceptually, an FST comprises two tapes, an input tape and an output
tape, and computes a function that maps strings on the input tape
to strings on the output tape. For a cellular FST operating repeatedly
in discrete time steps with instantaneous communication of outputs
to neighbors, the concept of tapes does not seem a natural fit. Some
awkward constructions would be needed to copy outputs from neighbors'
tapes to each cell's input tape or for them to be shared some how.

To address this difficulty, we will provide our cellular FSTs with
$N$ direct, discrete inputs and a single output. These inputs and
outputs can be routed and connected together to allow various constructions
just as if they were wires, similar to traditional CA. More formally,
we define a finite state cell (FSC) as a quadruple $(Q,\Sigma,q_{s},\delta)$,
where:
\begin{itemize}
\item $Q$ is a finite set of states, 
\item $\Sigma$ is the alphabet (input and output) of the cell,
\item $q_{s}\in Q$ is the start state, 
\item $\delta:Q\times\Sigma^{N}\rightarrow Q\times\Sigma$ is the transition
function.
\end{itemize}
Again, we limit our discussion to $N=3$. For convenience, we name
the 3 input values read each time step $\lambda$, $\omega$, and
$\rho$, where $\omega$ is the output of the cell routed back as
an input. An FSC requires an initial value $\omega_{s}\in\Sigma$
to begin operation. Once received, the cell outputs $\omega_{s}$
and enters $q_{s}$. At each time step, the cell computes a function
from $\Sigma^{3}$ to $\Sigma$ depending on the current state and
then changes to the next state, all according to $\delta$. The ordered
string $\lambda\omega\rho$ defines the input to this function, and
its result becomes the next value of $\omega$. FSCs have no final
states, and simply operate continuously after initialization.

It is easy to see the analogy to traditional CA. FSCs, however, are
not fixed in the function they compute at each time step. Instead,
these cells may have an arbitrary but finite number of states, the
transitions between each of which compute different functions. The
path through the states may be dependent on the inputs received or
may be fixed.

We will be concerned only with the case of $\Sigma=\{0,1\}$, though
it is easy to imagine FSCs with a larger $\Sigma$. We could also
consider larger values of $N$. These variations will not be necessary
for our present purposes and so will not be considered.

\subsection{Cellular Automata based on FSCs}

We can now consider a cellular automaton which aggregates $n$ such
cells. We define a finite state cellular automaton, or FSCA, $A=\langle a_{1},a_{2},\ldots,a_{n}\rangle$
as an array of $n$ finite state cells as defined above. The left
neighbor input $\lambda_{i}=\omega_{i-1}$ and the right neighbor
input $\rho_{i}=\omega_{i+1}$ for $i=2,\ldots,n-1$. A cyclic-boundary
FSCA connects the end inputs to the opposite end's output, so that
$\lambda_{1}=\omega_{n}$ and $\rho_{n}=\omega_{1}$. We say that
$A$ has a value $\mathbf{s}\in\{0,1\}^{n}$ when $\mathbf{s}=\omega_{1}|\omega_{2}|\ldots|\omega_{n}$
and more specifically we denote the value of $A$ at time step $t$
by $\mathbf{s}^{(t)}$.

It is useful to have a notion of the current configuration of the
entire FSCA which describes the current output values and state for
each cell. Let $Q_{i}$ be the set of states for $a_{i}$, and $\mathcal{Q}=Q_{1}\times Q_{2}\times\ldots\times Q_{n}$.
We say an FSCA $A$ has \emph{configuration} $C^{(t)}=(\mathbf{q}^{(t)},\mathbf{s}^{(t)})$
at time step $t$ for $\mathbf{q}^{(t)}\in\mathcal{Q}$ if $\mathbf{q}_{i}^{(t)}$
is the current state of $a_{i}$ at time $t$ for $i\leq n$, and
$\mathbf{s}^{(t)}$ is the value of $A$ at time $t$. We say $C^{(t)}$
\emph{yields }$C^{(t+1)}$, written $C^{(t)}\vdash_{\! A}C^{(t+1)}$
if operating $A$ with current configuration $C^{(t)}$ for one time
step produces configuration $C^{(t+1)}$. For short hand, we may also
write $C^{(t)}\vdash_{\! A}^{k}C^{(t+k)}$ to show the operation of
$A$ for $k$ time steps.

We can represent elementary CA as a special kind of FSCA where each
cell has only a single state. We simply define 8 self-transitions
that map all possible $\lambda\omega\rho$ inputs to a new output
$\omega'$. This collection of transitions then defines a function
from 3 bits to 1 bit, which is the rule of the cell over all time
steps. An example FSC for an elementary rule 30 cell is shown in \figref{rule30FSC}. 

\begin{figure}
\begin{centering}
\includegraphics[bb=0bp 5bp 135bp 300bp,clip,scale=0.45]{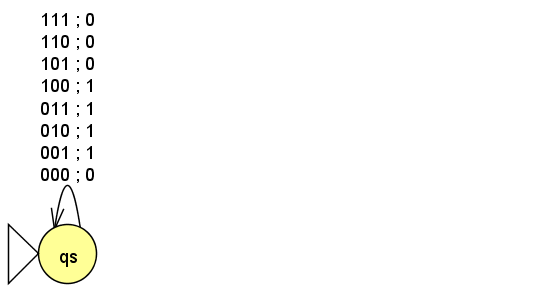}
\par\end{centering}

\caption{\label{fig:rule30FSC}\textsf{FSC diagram for an elementary rule 30
cell.}}
\end{figure}

\subsection{Variations on FSCA}

We have seen that single-state cells create FSCA which are equivalent
to traditional elementary CA and it seems clear that two states can
mimic alternating between two such rules (as proposed in \cite{GuanT04}).
We now explore the computational complexity of an $n$-cell FSCA with
cells having for example $256$ or $n$ or $n^{2}$ states. Specifically,
we would like to know at what number of cells does an FSCA become
computationally non-invertible. We may also wonder about the effects
of other limiting properties, such as the effect of a limited branching
factor on the transition paths through the states. 

To capture these limiting notions, we use a function $B:\Nat\mapsto\Nat$
to provide an upper bound on the number of states any cell in an FSCA
can have as a function of the number of cells in the FSCA, so that
$|Q_{i}|\leq B(n)$ for $i\leq n$.

We define an \emph{elementary} cell as one which obeys $B(n)=1$ having
only one state, $q_{s}$. With only one state, an elementary cell
can have only one transition type: $\delta(q_{s},\lambda\omega\rho)=(q_{s},f(\lambda\omega\rho))$
where $f:\Sigma^{3}\rightarrow\Sigma$ is the function defined by
the set of self transitions of $q_{s}$. Therefore, the output of
the cell at each time step is the result of a single, fixed function.
Thus an FSCA with elementary cells is equivalent to an elementary
CA.

We say $\delta$ is \emph{simple} if there is a $\delta':Q\rightarrow Q$
and an $f:\Sigma^{3}\rightarrow\Sigma$ such that for all $q\in Q$,
$\delta(q,\lambda\omega\rho)=(\delta'(q),f(\lambda\omega\rho))$.
Intuitively, this means that if there is any transition from $q$
to another state $r$, then all input combinations take an FSC $a$
from $q$ to $r$. The inputs have no effect on selecting the next
state and only the output bit on these transitions may vary. Under
a simple $\delta$, every state has exactly one successor state with
$f$ as the \emph{time step function} computed at that time step.
We say a cell is simple if its $\delta$ is simple, and an FSCA is
simple if all cells are simple.

This notion of simple cells turns out to be quite an interesting one.
If we construct special purpose cells to affect a specific function
at each time step, we can have two neighboring cells swap their values
at a certain time step, or compute the sum (XOR) and a carry (AND)
of their two values for example. If an FSCA has only simple cells,
it will perform the same computation without regard to the value of
the cells at any time step. This begins to have the feel of a machine
capable of universal computation.

It will be convenient to have a shorthand notation defining transitions
which compute a given function for all combinations of $Q\times\Sigma^{3}$.
We define a \emph{transition set} $T_{f}(q,r)$ from a state $q$
to a state $r$ with respect to a time step function $f$ as the set
$\{((q,\lambda\omega\rho),(r,f(\lambda\omega\rho)))\}$. Where the
implicit definition of $f$ is simple, we will use its expression
in the notation, e.g. $T_{\overline{\omega}}(q,r)$ denotes the use
of the complement as the time step function. Such transition sets
can be combined to define $\delta$ as a function. Defining some commonly
used transition sets illustrates the concept and will also be useful
in our discussion of FSC capabilities to follow.

\begin{figure}
\begin{centering}
\begin{minipage}[c]{0.45\columnwidth}%
\begin{center}
\subfloat[\label{fig:LeftTransFunc}\textsf{The }\textsf{\textsc{$T_{\lambda}$}}\textsf{
transition set}]{\begin{centering}
\includegraphics[bb=0bp 585bp 375bp 670bp,clip,scale=0.4]{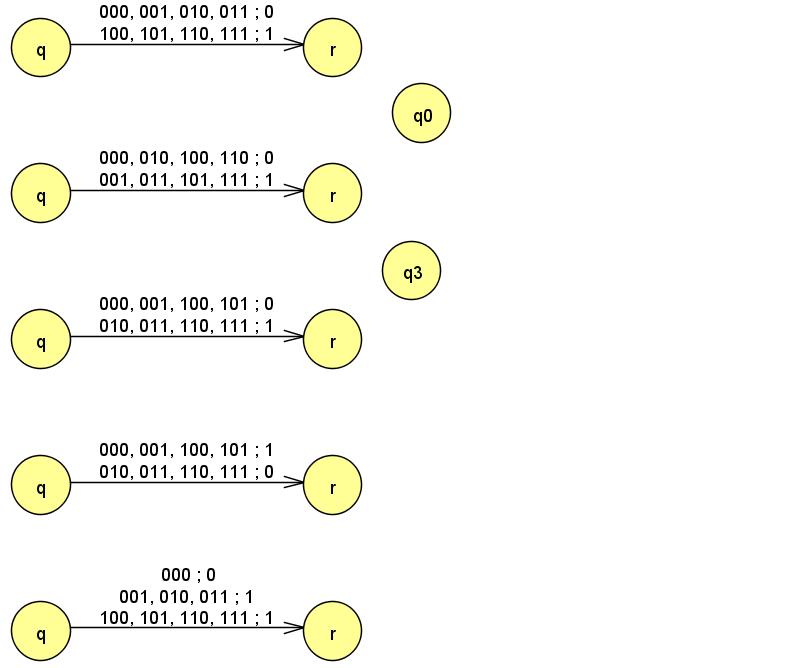}
\par\end{centering}

}
\par\end{center}%
\end{minipage}\hfill{}%
\begin{minipage}[c]{0.45\columnwidth}%
\begin{center}
\subfloat[\label{fig:RightTransFunc}\textsf{The }\textsf{\textsc{$T_{\rho}$}}\textsf{
transition set}]{\centering{}\includegraphics[bb=0bp 440bp 375bp 525bp,clip,scale=0.4]{fscaFuncs}}
\par\end{center}%
\end{minipage}
\par\end{centering}

\begin{centering}
\begin{minipage}[c]{0.45\columnwidth}%
\begin{center}
\subfloat[\label{fig:CenterTransFunc}\textsf{The }\textsf{\textsc{$T_{\omega}$}}\textsf{
transition set}]{\begin{centering}
\includegraphics[bb=0bp 290bp 375bp 375bp,clip,scale=0.4]{fscaFuncs}
\par\end{centering}

}
\par\end{center}%
\end{minipage}\hfill{}%
\begin{minipage}[c]{0.45\columnwidth}%
\begin{center}
\subfloat[\label{fig:ComplementTransFunc}\textsf{The }\textsf{\textsc{$T_{\overline{\omega}}$}}\textsf{
transition set}]{\centering{}\includegraphics[bb=0bp 145bp 375bp 230bp,clip,scale=0.4]{fscaFuncs}}
\par\end{center}%
\end{minipage}
\par\end{centering}

\centering{}%
\begin{minipage}[c]{0.45\columnwidth}%
\begin{center}
\subfloat[\label{fig:OrTransFunc}The \textsc{$T_{\lambda+\omega+\rho}$} transition
set]{\centering{}\includegraphics[bb=0bp 0bp 375bp 110bp,clip,scale=0.4]{fscaFuncs}}
\par\end{center}%
\end{minipage}\caption{\label{fig:TransitionSets}\textsf{Transition sets for useful functions
in FSCs. Input combinations separated by a comma in the diagrams is
short-hand for distinct inputs which share the same resulting output
which appears after a semi-colon.}}
\end{figure}

\begin{lyxlist}{00.00.0000}
\item [{\textsc{$T_{\lambda}(q,r)$:}}] $\{((q,\lambda\omega\rho),(r,\lambda))\}$. 
\item [{\textsc{$T_{\rho}(q,r)$:}}] $\{((q,\lambda\omega\rho),(r,\rho))\}$. 
\item [{\textsc{$T_{\omega}(q,r)$:}}] $\{((q,\lambda\omega\rho),(r,\omega))\}$. 
\item [{\textsc{$T_{\overline{\omega}}(q,r)$:}}] $\{((q,\lambda\omega\rho),(r,\overline{\omega}))\}$. 
\item [{\textsc{$T_{\lambda+\omega+\rho}(q,r)$:}}] $\{((q,\lambda\omega\rho),(r,\lambda+\omega+\rho))\}$
where + denotes Boolean OR. 
\item [{\textsc{$T_{\alpha}(q,r)$:}}] $\{((q,\lambda\omega\rho),(r,\alpha))\:|\:\alpha\in\Sigma\}$. 
\end{lyxlist}
These transition sets are shown in \figref{TransitionSets} in diagrammatic
form.

\subsection{Evaluating 3-CNF formulas with FSCA}

Given these functions, it's not hard to imagine building a simple
FSCA capable of performing basic computational tasks, such as evaluating
a Boolean formula for instance. Variable assignments are input as
initial values in certain cells. Cells route these assignments to
positions in each clause of the formula using the transition sets
\textsc{$T_{\lambda}$}, \textsc{$T_{\rho}$}, and\textsc{ $T_{\omega}$,
}all to arrive at some fixed time step. Literals that are complements
of their variables would require the transition set \textsc{$T_{\overline{\omega}}$}.
When all variables are positioned correctly and complemented according
to the formula, the evaluation of the clauses begins. If the clauses
have only three literals (i.e. the formula is in 3-CNF), this can
be accomplished in one time step with the transition set \textsc{$T_{\lambda+\omega+\rho}$}
defined above. Those cells which do not compute the disjunction of
a clause can simply use a transition set $T_{\alpha}$ with $\alpha=0$.
The conjunction of clauses then follows, with each clause value being
routed again down to a single cell which holds the final formula value. 

We illustrate this concept with an example. Suppose we have a 3-CNF
formula over 5 variables in 4 clauses given by: 
\[
\phi=(x_{1}\vee x_{2}\vee\overline{x_{3}})\wedge(x_{1}\vee x_{3}\vee x_{5})\wedge(\overline{x_{2}}\vee\overline{x_{4}}\vee x_{5})\wedge(x_{3}\vee x_{4}\vee\overline{x_{5}}).
\]
 An FSCA $A$ might use twelve cells to evaluate $\phi$, one for
each literal. At time step 0, five specially selected cells would
be provided the assignments to $x_{1},\ldots,x_{5}$ while other cells
get 0. We will label the actual assignment values as $\alpha_{1},\ldots,\alpha_{5}$.
A computational history of $A$, showing the cell values at each time
step, might look like \tabref{ExFscaHistory}. Here, evaluation is
carried out only up to clause evaluation, at which point satisfiability
is clear. Notice that at time step 1, all the dummy 0 values are gone,
and each assignment value appears as many times as its variable appears
in $\phi$. From this point, routing is just a matter of permuting
the contents of the cells. At time step 6, all cells contain the correct
assignments ignoring negation. This is accounted for in step 7. Finally,
step 8 performs the ORing of clauses. Cells not evaluating a clause
get a fixed 0 value.

\begin{center}
\begin{table}
\begin{centering}
\begin{tabular}{|c|c|c|c|c|c|c|c|c|c|c|c|c|}
\hline 
 & \multicolumn{12}{c|}{cell}\tabularnewline
\hline 
\emph{t} & 1 & 2 & 3 & 4 & 5 & 6 & 7 & 8 & 9 & 10 & 11 & 12\tabularnewline
\hline 
\hline 
0 & \textcolor{red}{$\alpha_{1}$} & \texttt{\textsc{0}} & \textcolor{green}{$\alpha_{2}$} & \texttt{\textsc{0}} & \texttt{\textsc{0}} & \textcolor{blue}{$\alpha_{3}$} & \texttt{\textsc{0}} & \textcolor{magenta}{$\alpha_{4}$} & \texttt{\textsc{0}} & \texttt{\textsc{0}} & $\alpha_{5}$ & \texttt{\textsc{0}}\tabularnewline
\hline 
1 & \textcolor{red}{$\alpha_{1}$} & \textcolor{red}{$\alpha_{1}$} & \textcolor{green}{$\alpha_{2}$} & \textcolor{green}{$\alpha_{2}$} & \textcolor{blue}{$\alpha_{3}$} & \textcolor{blue}{$\alpha_{3}$} & \textcolor{blue}{$\alpha_{3}$} & \textcolor{magenta}{$\alpha_{4}$} & \textcolor{magenta}{$\alpha_{4}$} & $\alpha_{5}$ & $\alpha_{5}$ & $\alpha_{5}$\tabularnewline
\hline 
2 & \textcolor{red}{$\alpha_{1}$} & \textcolor{green}{$\alpha_{2}$} & \textcolor{red}{$\alpha_{1}$} & \textcolor{blue}{$\alpha_{3}$} & \textcolor{green}{$\alpha_{2}$} & \textcolor{blue}{$\alpha_{3}$} & \textcolor{blue}{$\alpha_{3}$} & \textcolor{magenta}{$\alpha_{4}$} & \textcolor{magenta}{$\alpha_{4}$} & $\alpha_{5}$ & $\alpha_{5}$ & $\alpha_{5}$\tabularnewline
\hline 
3 & \textcolor{red}{$\alpha_{1}$} & \textcolor{green}{$\alpha_{2}$} & \textcolor{blue}{$\alpha_{3}$} & \textcolor{red}{$\alpha_{1}$} & \textcolor{blue}{$\alpha_{3}$} & \textcolor{green}{$\alpha_{2}$} & \textcolor{magenta}{$\alpha_{4}$} & \textcolor{blue}{$\alpha_{3}$} & $\alpha_{5}$ & \textcolor{magenta}{$\alpha_{4}$} & $\alpha_{5}$ & $\alpha_{5}$\tabularnewline
\hline 
4 & \textcolor{red}{$\alpha_{1}$} & \textcolor{green}{$\alpha_{2}$} & \textcolor{blue}{$\alpha_{3}$} & \textcolor{red}{$\alpha_{1}$} & \textcolor{blue}{$\alpha_{3}$} & \textcolor{green}{$\alpha_{2}$} & \textcolor{magenta}{$\alpha_{4}$} & $\alpha_{5}$ & \textcolor{blue}{$\alpha_{3}$} & $\alpha_{5}$ & \textcolor{magenta}{$\alpha_{4}$} & $\alpha_{5}$\tabularnewline
\hline 
5 & \textcolor{red}{$\alpha_{1}$} & \textcolor{green}{$\alpha_{2}$} & \textcolor{blue}{$\alpha_{3}$} & \textcolor{red}{$\alpha_{1}$} & \textcolor{blue}{$\alpha_{3}$} & \textcolor{green}{$\alpha_{2}$} & $\alpha_{5}$ & \textcolor{magenta}{$\alpha_{4}$} & $\alpha_{5}$ & \textcolor{blue}{$\alpha_{3}$} & \textcolor{magenta}{$\alpha_{4}$} & $\alpha_{5}$\tabularnewline
\hline 
6 & \textcolor{red}{$\alpha_{1}$} & \textcolor{green}{$\alpha_{2}$} & \textcolor{blue}{$\alpha_{3}$} & \textcolor{red}{$\alpha_{1}$} & \textcolor{blue}{$\alpha_{3}$} & $\alpha_{5}$ & \textcolor{green}{$\alpha_{2}$} & \textcolor{magenta}{$\alpha_{4}$} & $\alpha_{5}$ & \textcolor{blue}{$\alpha_{3}$} & \textcolor{magenta}{$\alpha_{4}$} & $\alpha_{5}$\tabularnewline
\hline 
7 & \textcolor{red}{$\alpha_{1}$} & \textcolor{green}{$\alpha_{2}$} & \textcolor{blue}{$\overline{\alpha_{3}}$} & \textcolor{red}{$\alpha_{1}$} & \textcolor{blue}{$\alpha_{3}$} & $\alpha_{5}$ & \textcolor{green}{$\overline{\alpha_{2}}$} & \textcolor{magenta}{$\overline{\alpha_{4}}$} & $\alpha_{5}$ & \textcolor{blue}{$\alpha_{3}$} & \textcolor{magenta}{$\alpha_{4}$} & $\overline{\alpha_{5}}$\tabularnewline
\hline 
8 & \texttt{\textsc{0}} & \textcolor{cyan}{$c_{1}$} & \texttt{\textsc{0}} & \texttt{\textsc{0}} & \textcolor{cyan}{$c_{2}$} & \texttt{\textsc{0}} & \texttt{\textsc{0}} & \textcolor{cyan}{$c_{3}$} & \texttt{\textsc{0}} & \texttt{\textsc{0}} & \textcolor{cyan}{$c_{4}$} & \texttt{\textsc{0}}\tabularnewline
\hline 
\end{tabular}
\par\end{centering}

\textsf{\caption{\textsf{\label{tab:ExFscaHistory}Computational history of $A$ for
formula $\phi$ with assignments $\alpha_{1},\ldots,\alpha_{5}$.
$c_{1},\ldots,c_{4}$ give the values of each clause in $\phi$ for
these assignments.}}
}
\end{table}

\par\end{center}

To formalize this example, we give the following theorem and constructive
proof in which we consider 3-CNF formulas where no clause uses the
same variable twice and no clause appears more than once (as these
are easily reducible.) The theorem will only cover evaluation of the
clauses of the formula. We adopt the following conventions in all
pseudocode:
\begin{itemize}
\item {[}\ {]} denotes array/table indexing
\item for a 3-CNF formula $\phi$, $\langle\phi\rangle$ denotes an encoding
of $\phi$ as an array of literals such that $\langle\phi\rangle[i]\vee\langle\phi\rangle[i+1]\vee\langle\phi\rangle[i+2]$
is a clause in $\phi$ if $i\equiv1\bmod3$. Each literal $\ell$
in $\phi$ is encoded as a pair $\langle j,c\rangle$ where $x_{j}$
is a variable in $\phi$ and $c=1$ if $\ell=\overline{x_{j}}$ or
$c=0$ if $\ell=x_{j}$. $\langle\phi\rangle[8][1]$ is then the encoded
variable of the second literal in the third clause of $\phi$.
\item Variables will appear as single symbols or as ``$Variable$.''
\item Sub-algorithms and procedures will appear as ``\textsc{SubAlgorithm}.''
\end{itemize}
\begin{thm}
\label{thm:AevalsPhi}Let $\phi$ be a 3-CNF formula with $c$ unique
clauses and $v$ variables $x_{1},x_{2},\ldots,x_{v}$ such that no
variable appears twice in the same clause. There exists a simple FSCA
$A$ that evaluates the clauses of $\phi$ for any encoded assignments
$\alpha_{1},\alpha_{2},\ldots,\alpha_{v}$.\end{thm}
\begin{proof}
The proof is by construction. We wish to show we can create a machine
whose computation will correctly evaluate $\phi$. To do so, we will
perform the computation in the abstract and construct the machine
to mirror each step. The construction will follow the same steps as
the example above: 
\begin{enumerate}
\item Create an array of cells over which to perform the computation. Every
group of 3 cells will correspond to a clause.
\item Determine the cells which will accept the assignments as initial values. 
\item Duplicate the assignment values so there is one copy for each literal
that needs it. 
\item Distribute those values to their positions in the clauses.
\item Account for negation in the literals where necessary
\item OR the literals in the same clause together.
\end{enumerate}
Consider \algref{3SatToFsca}, \nameref{alg:3SatToFsca}, \vpageref{alg:3SatToFsca}
and its sub-algorithms \nameref{alg:GrowSpans} (\algref{GrowSpans}
\vpageref{alg:GrowSpans}), \nameref{alg:FindDestinations} (\algref{FindDestinations}
\vpageref{alg:FindDestinations}), and \nameref{alg:OrderAssignments}
(\algref{OrderAssignments} \vpageref{alg:OrderAssignments}) which
we now sketch. \nameref{alg:3SatToFsca} takes as input a formula
$\phi$ and constructs a machine which takes as input $3c$ Boolean
values, $v$ of which are the assignments to variables. As we cannot
know what assignments will be provided, the algorithm simply identifies
the locations which will contain an assignment $\alpha_{j}$ by storing
the value $j$. The algorithm first constructs the FSCA cells and
computes the indices of cells whose initial values will be the assignments
to evaluate. These indices are determined by the number of times each
assignment is needed, and that number is stored in the array $Count$.
If, for example, some $\alpha_{j}$ is used five times in the formula,
it will be assigned a starting position $p[j]$ that is in the middle
of a span of five cells, and the span for $\alpha_{j+1}$ will immediately
follow. The array $V$ will store the contents of each cell---either
an index $j$ of the indeterminate assignment value, or 0. The values
in $V$ at each stage of the construction indicate the value of the
FSCA at the corresponding stage of its computational history.

\nameref{alg:GrowSpans} replicates the initial assignment values
for each cell in that assignment's span so that there are as many
cells with an assignment $\alpha_{j}$ as there are appearances of
its variable $x_{j}$ in $\phi$. As it does so, it adds states and
transitions to the cells of the constructed FSCA $A$ which perform
the same copy operation. When complete, all cells will contain the
assignment for the span to which they are assigned.

\nameref{alg:FindDestinations} examines where in $\phi$ each assignment
copy in $V$ is needed and assigns that position as a destination.
These destinations are stored in the array $D$ which is one-to-one
with $V$. If each assignment in $V$ were moved to its corresponding
destination in $D$, it would be in the same position as its variable
in $\phi$. 

Creating the machinery to actually move these elements is the responsibility
of \nameref{alg:OrderAssignments}. This algorithm uses an Odd-Even
sort on the destinations in $D$ to re-order the contents in $V$.
While doing so, it creates states and transitions in the cells of
$A$ to perform the same reordering. The result is the list of assignments
which exactly mirrors the variables in $\phi$.

Finally, \nameref{alg:3SatToFsca} adds states and transitions to
account for complementing literals that require it, and then for evaluating
clauses.

These algorithms are in turn aided by \nameref{pro:SplitLeft} (\proref{SplitLeft}
\vpageref{pro:SplitLeft}), \nameref{pro:SplitRight} (\proref{SplitRight}
\vpageref{pro:SplitRight}), \nameref{pro:SplitLeftRight} (\proref{SplitLeftRight}
\vpageref{pro:SplitLeftRight}), \nameref{pro:Swap} (\proref{Swap}
\vpageref{pro:Swap}), and \nameref{pro:Remain} (\proref{Remain}
\vpageref{pro:Remain}). These procedures perform the moving of elements
in $V$ paired with creating states and transitions to perform the
same moving of values in the corresponding cells of $A$. This helps
to ensure the computational history of $A$ will mirror what has happened
in $V$. 

\begin{algorithm}
\textbf{Input:} $\langle\phi\rangle$, an encoding of a 3-CNF Boolean
formula over the $v$ variables $x_{1},x_{2},\ldots,x_{v}$.

\textbf{Output:} FSCA $A$, and $k\in\Nat$ 
\begin{algor}[1]
\item [{.}] $n=3|\langle\phi\rangle|$
\item [{.}] \textbf{for} $\; j=1,\ldots,v\;$ \textbf{do} $\; Count[i]\gets0$
\item [{for}] $i=1,\ldots,n$ \{\{initialize data structures used to create
$A$\}\}

\begin{algor}[1]
\item [{.}] $Count[\langle\phi\rangle[i][1]]\gets Count[\langle\phi\rangle[i][1]]+1$
\item [{.}] $V[i]\gets0$
\item [{.}] Create an initial state $q_{i,0}$
\item [{.}] $Q_{i}\gets\{q_{i,0}\}$
\item [{.}] Create a finite state cell $a_{i}$ with output $\omega_{i},Q=Q_{i},\Sigma=\{0,1\},q_{s}=q_{i,0},\delta=\emptyset$
\item [{.}] $Last[i]\gets q_{i,0}$
\end{algor}
\item [{endfor}]~
\item [{.}] $i\gets1$
\item [{for}] $j=1,\ldots,|Count|$ \{\{determine initial positions for
variables\}\}

\begin{algor}[1]
\item [{.}] $p[j]\gets i+\left\lfloor (Count[j]-1)/2\right\rfloor $
\item [{.}] $V[p[j]]\gets j$
\item [{.}] $i\gets i+Count[j]$
\end{algor}
\item [{endfor}]~
\item [{.}] $\lambda_{1}\gets\omega_{n}$; $\rho_{1}\gets\omega_{2}$ \{\{connect
all cells using cyclic boundaries\}\}
\item [{.}] $\rho_{n}\gets\omega_{1}$; $\lambda_{n}\gets\omega_{n-1}$
\item [{for}] $i=2,\ldots,n-1$

\begin{algor}[1]
\item [{.}] $\lambda_{i}\gets\omega_{i-1}$
\item [{.}] $\rho_{i}\gets\omega_{i+1}$
\end{algor}
\item [{endfor}]~
\item [{.}] $A\gets\langle a_{1},a_{2},\ldots,a_{n}\rangle$
\item [{.}] $t\gets$\nameref{alg:GrowSpans}$(A,V,p,Count,Last)$
\item [{.}] $D\gets$\nameref{alg:FindDestinations}$(\langle\phi\rangle,v)$
\item [{.}] $t\gets$\nameref{alg:OrderAssignments}$(A,V,D,t,Last)$
\item [{for}] $i=1,\ldots,n$ \{\{account for complementation, evaluate
all clauses\}\}

\begin{algor}[1]
\item [{.}] create states $q_{i,t+1},q_{i,t+2}$
\item [{.}] $Q_{i}\gets Q_{i}\cup\{q_{i,t+1},q_{i,t+2}\}$
\item [{if}] $\langle\phi\rangle[i][2]=1$ 

\begin{algor}[1]
\item [{.}] $\delta_{i}\gets\delta_{i}\cup\: T_{\overline{\omega}}(Last[i],q_{i,t+1})$
\end{algor}
\item [{else}]~

\begin{algor}[1]
\item [{.}] $\delta_{i}\gets\delta_{i}\cup\: T_{\omega}(Last[i],q_{i,t+1})$
\end{algor}
\item [{endif}]~
\item [{if}] $i\equiv2\mod3$

\begin{algor}[1]
\item [{.}] $\delta_{i}\gets\delta_{i}\cup\: T_{\lambda+\omega+\rho}(q_{i,t+1},q_{i,t+2})$
\end{algor}
\item [{else}]~

\begin{algor}[1]
\item [{.}] $\delta_{i}\gets\delta_{i}\cup\: T_{0}(q_{i,t+1},q_{i,t+2})$
\end{algor}
\item [{endif}]~
\item [{.}] $\delta_{i}\gets\delta_{i}\cup\: T_{\omega}(q_{i,t+2},q_{i,t+2})$
\{\{self-transition forever with same value\}\}
\end{algor}
\item [{endfor}]~
\item [{.}] \textbf{return} $(A,t+2)$
\end{algor}
\caption{\textsc{\label{alg:3SatToFsca}3-SatToFsca}}
\end{algorithm}

\begin{algorithm}
\textbf{Input:} $A,V,p,Count,Last$ where:

\quad{}$A$ is the array of cells, 

\quad{}$V$ is the array of variables held in each cell,

\quad{}$p$ is the starting cell index of each variable,

\quad{}$Count$ is the number of occurrences of each variable,

\quad{}$Last$ is the last state created in each $a_{i}\in A$.

\textbf{Output:} $t$, the number of time steps used 
\begin{algor}[1]
\item [{for}] $j=1,\ldots,|Count|$

\begin{algor}[1]
\item [{.}] $Need[j]\gets Count[j]-1$
\item [{.}] $Span[j]\gets\langle p[j],p[j]\rangle$
\end{algor}
\item [{endfor}]~
\item [{.}] $t\gets0$
\item [{while}] $\sum_{N\in Need}N\:>0$

\begin{algor}[1]
\item [{.}] $t\gets t+1$
\item [{.}] $i\gets1$
\item [{for}] $j=1,\ldots,|Need|$

\begin{algor}[1]
\item [{if}] $Need[j]=0$ \{\{skip to next span\}\}

\begin{algor}[1]
\item [{.}] \textbf{next} $j$ 
\end{algor}
\item [{elseif}] $Need[j]\geq2$

\begin{algor}[1]
\item [{for}] $k=i,\ldots,Span[j][1]-2$ \{\{visit all elements since last
split up to next split\}\}

\begin{algor}[1]
\item [{.}] \textsc{Remain$(A,k,t,Last)$}
\end{algor}
\item [{endfor}]~
\item [{if}] $Span[j][1]=Span[j][2]$ \{\{for first split, grow in both
directions\}\}

\begin{algor}[1]
\item [{.}] \textsc{SplitLeftRight$(A,V,Span[j][1],t,Last)$}
\item [{.}] $Span[j]\gets\langle Span[j][1]-1,Span[j][2]+1\rangle$ \{\{adjust
both ends of span\}\}
\item [{.}] $Need[j]\gets Need[j]-2$
\item [{.}] $i\gets Span[j][2]+1$ \{\{next element to consider is just
right of the new span end\}\}
\item [{.}] \textbf{next} $j$ \{\{begin work on the next span\}\}
\end{algor}
\item [{else}]~

\begin{algor}[1]
\item [{.}] \textsc{SplitLeft$(A,V,Span[j][1],t,Last)$ }\{\{Grow the left
side\}\}
\item [{.}] $Span[j]\gets\langle Span[j][1]-1,Span[j][2]\rangle$ \{\{adjust
the left span end\}\}
\item [{.}] $Need[j]\gets Need[j]-1$
\item [{.}] $i\gets Span[j][1]+2$
\end{algor}
\item [{endif}]~
\end{algor}
\item [{endif}]~
\item [{.}] \{\{$Need[j]>0$ implies the right end needs to grow, whether
or not left end grows\}\}
\item [{for}] $k=i,\ldots,Span[v][2]-1$ 

\begin{algor}[1]
\item [{.}] \textsc{Remain$(A,k,t,Last)$} \{\{keep current values up to
end of span\}\}
\end{algor}
\item [{endfor}]~
\item [{.}] \textsc{SplitRight$(A,V,Span[j][2],t,Last)$ }\{\{Grow right\}\}
\item [{.}] $Span[j]\gets\langle Span[j][1],Span[j][2]+1\rangle$ \{\{adjust
the right span end\}\}
\item [{.}] $Need[j]\gets Need[j]-1$
\item [{.}] $i\gets Span[j][2]+1$
\end{algor}
\item [{endfor}]~
\item [{for}] $k=i,\ldots,|V|$

\begin{algor}[1]
\item [{.}] \textsc{Remain$(A,k,t,Last)$ }\{\{keep current values to end
of array\}\}
\end{algor}
\item [{endfor}]~
\end{algor}
\item [{endwhile}]~
\item [{.}] \textbf{return} $t$
\end{algor}
\caption{\textsc{\label{alg:GrowSpans}GrowSpans}}
\end{algorithm}

\begin{algorithm}
\textbf{Input:} $\langle\phi\rangle,v$ where:

\quad{}$\langle\phi\rangle$ is the 3-CNF formula, encoded as an
array of literals, whose variables are targets for each cell,

\quad{}$v$ is the number of variables in $\phi$.

\begin{raggedright}
\textbf{Output:} $D$, an array of destinations for each cell.
\par\end{raggedright}
\begin{algor}[1]
\item [{.}] \textbf{for }$\; j=1,\ldots,v\;$ \textbf{do} $\; Dest[j]\gets[\,]$
\item [{for}] $i=1,\ldots,|\langle\phi\rangle|$

\begin{algor}[1]
\item [{.}] $Dest[\langle\phi\rangle[i][1]]\gets Dest[\langle\phi\rangle[i][1]]\;||\:[i]$
\{\{append destinations for each variable to its own list\}\}
\end{algor}
\item [{endfor}]~
\item [{.}] $i\gets0$
\item [{forall}] $d\in Dest$ \{\{concatenate the destinations\}\}

\begin{algor}[1]
\item [{for}] $j=i,\ldots,|d|$ 

\begin{algor}[1]
\item [{.}] $D[i+j]\gets d[j]$
\end{algor}
\item [{endfor}]~
\item [{.}] $i\gets i+|d|$
\end{algor}
\item [{endfor}]~
\item [{.}] \textbf{return} $D$
\end{algor}
\caption{\textsc{\label{alg:FindDestinations}FindDestinations}}
\end{algorithm}

\begin{algorithm}
\textbf{Input:} $A,V,D,t,Last$ where:

\quad{}$A$ is the array of cells, 

\quad{}$V$ is the array of variables held in each cell,

\quad{}$D$ is the array of destinations for each variable in $V$,

\quad{}$t$ is the time step at which distributing variables begins.

\quad{}$Last$ is the array of last states created in each $a_{i}\in A$.

\begin{raggedright}
\textbf{Output:} $t$, the time step at which all variables are at
their destinations.
\par\end{raggedright}
\begin{algor}[1]
\item [{.}] $Sorted\gets\mathtt{false}$ 
\item [{while}] not $Sorted$

\begin{algor}[1]
\item [{.}] $Sorted\gets\mathtt{true}$
\item [{.}] $t\gets t+1$
\item [{for}] $i=1$ to $|V|-1$ in steps of $2$

\begin{algor}[1]
\item [{if}] $D[i]>D[i+1]$

\begin{algor}[1]
\item [{.}] \textsc{Swap$(A,V,D,i,t,Last)$}
\item [{.}] $Sorted\gets\mathtt{false}$
\end{algor}
\item [{else}]~

\begin{algor}[1]
\item [{.}] \textsc{Remain$(A,i,t,Last)$}
\item [{.}] \textsc{Remain$(A,i+1,t,Last)$}
\end{algor}
\item [{endif}]~
\end{algor}
\item [{endfor}]~
\item [{if}] $|V|\equiv0\bmod2$

\begin{algor}[1]
\item [{.}] \textsc{Remain$(A,|V|,t,Last)$}
\end{algor}
\item [{endif}]~
\item [{if}] $Sorted$

\begin{algor}[1]
\item [{.}] \textbf{break}
\end{algor}
\item [{endif}]~
\item [{.}] $Sorted\gets\mathtt{true}$
\item [{.}] $t\gets t+1$
\item [{.}] \textsc{Remain$(A,1,t,Last)$}
\item [{for}] $i=2$ to $|V|-1$ in steps of $2$

\begin{algor}[1]
\item [{if}] $D[i]>D[i+1]$

\begin{algor}[1]
\item [{.}] \textsc{Swap$(A,V,D,i,t,Last)$}
\item [{.}] $Sorted\gets\mathtt{false}$
\end{algor}
\item [{else}]~

\begin{algor}[1]
\item [{.}] \textsc{Remain$(A,i,t,Last)$}
\item [{.}] \textsc{Remain$(A,i+1,t,Last)$}
\end{algor}
\item [{endif}]~
\end{algor}
\item [{endfor}]~
\item [{if}] $|V|\equiv1\bmod2$

\begin{algor}[1]
\item [{.}] \textsc{Remain$(A,|V|,t,Last)$}
\end{algor}
\item [{endif}]~
\end{algor}
\item [{endwhile}]~
\end{algor}
\caption{\textsc{\label{alg:OrderAssignments}OrderAssignments}}
\end{algorithm}

We will develop the proof through a series of lemmata which establish
some necessary properties. 
\begin{lem}
\label{lem:VisPhi}\nameref{alg:3SatToFsca} computes an evolution
of $V$ from an initial condition, having a single element for each
encoded assignment and $\mathtt{0}$s otherwise, to a final condition
at the return of \nameref{alg:OrderAssignments} in which $V$ contains
the variables of $\phi$ matching in order of appearance without regard
to negation. That is, \nameref{alg:3SatToFsca} evolves $V$ to the
value $\langle\langle\phi\rangle[1][1],\langle\phi\rangle[2][1],\ldots,\langle\phi\rangle[n][1]\rangle$.\end{lem}
\begin{proof}
Let the function $\#(j)$ be defined as the number of occurrences
of $x_{j}$ in $\phi$ and so must have value greater than 0 for all
$j$. This is also the number of occurrences of $\alpha_{j}$ which
will be needed to evaluate $\phi$. We first show that $V$'s initialization
consists only of $\mathtt{0}$s and a single encoded element for each
$\alpha_{j}$. It is clear that $Count[j]\equiv\#(j)$ for $1\leq j\leq v$
since $Count[j]$ is initialized to 0 in line 2 and is increased on
line 4 of \nameref{alg:3SatToFsca} by 1 for every occurrence of $j$
(the encoding of $x_{j}$) in $\langle\phi\rangle$\textsc{. }We can
also see that lines 10-14 create a span in $V$ of length $\#(j)$
having the contents 
\[
\mathtt{0}^{\left\lfloor \frac{1}{2}\left(\#(j)-1\right)\right\rfloor }j\mathtt{0}^{\left\lceil \frac{1}{2}\left(\#(j)-1\right)\right\rceil }
\]
since: $V$ is initialized to all 0s on line 5; $j$ ranges over every
variable which was counted in $\langle\phi\rangle$; $p[j]$ is initialized
to the position $i+\left\lfloor (\#(j)-1)/2\right\rfloor $ on line
12, and $V[p[j]]$ is assigned $j$ on line 13. These spans begin
at index 1 in $V$ by line 10, and are all adjacent since each time
through the loop, $i$ is assigned $i+Count[j]$. Therefore, $V$
on completing initialization\textsc{ }has the form
\begin{equation}
\left\langle \mathtt{0}^{\left\lfloor \frac{1}{2}\left(\#(j)-1\right)\right\rfloor }j\mathtt{0}^{\left\lceil \frac{1}{2}\left(\#(j)-1\right)\right\rceil }\right\rangle _{j=1}^{v}\label{eq:VInit}
\end{equation}
so that each assignment appears exactly once. This satisfies the initial
condition of the theorem. 

We next show that on return from \nameref{alg:GrowSpans}, $V$ has
the value $\langle1^{\#(1)},2^{\#(2)},\ldots,v^{\#(v)}\rangle$ where
$1,\ldots,v$ encode $\alpha_{1},\ldots,\alpha_{v}$ and $b^{k}$
denotes $k$ sequential occurrences of $b$ for $k>0$ and the empty
string otherwise. In \nameref{alg:GrowSpans}, clearly $Need[j]$
is initialized to $\#(j)-1$ for $1\leq j\leq v$ by line 2. Since
each span around the assignments in $V$ has exactly one non-zero
element, $Need[j]$ counts the number of $\mathtt{0}$s in each assignment's
span. Notice that if $\#(j)$ is odd, $\#(j)-1$ is even and so evenly
divisible by 2, meaning there are an equal number of $\mathtt{0}$s
on either side of an assignment $j$. Conversely, If $\#(j)$ is even,
there is one more $\mathtt{0}$ on the right side of the assignment
$j$. It is also clear that $Span[j]$ is initialized with a pair
of indices of the first and last occurrence of $j$ in $V$. 

Let $\Delta(\langle x,y\rangle)=y-x+1$ be a function which computes
the number of elements in a span. We can see by induction that $\#(j)=\Delta(Span[j])+Need[j]$
for all $j$ at each pass through the while loop. For the base case
of $\#(j)=1$, we have $Span[j]=\langle p[j],p[j]\rangle$ so $\Delta(Span[j])=1$
and $Need[j]=0$ for all $j$ by line 2. Assume the relation holds
at an arbitrary pass through the loop. On the next pass, there are
three possibilities:
\begin{itemize}
\item $Need[j]\geq2$, for which their are two possibilities: $\Delta(Span[j])=1$
and so \textsc{SplitLeftRight} is called, or $\Delta(Span[j])>1$
and \textsc{SplitLeft} then\textsc{ SplitRight} is called. For the
first of these, a $\mathtt{0}$ is replaced with $j$ on both ends
of the span, $\Delta(Span[j])$ is increased by 2, and $Need[j]$
is decreased by 2. No further action is taken for this case since
the loop is abbreviated on line 19. For the second possibility, the
loop first replaces a single $\mathtt{0}$ on the left end of the
span by calling \textsc{SplitLeft} and then increases $\Delta(Span[j])$
by one and decreases $Need[j]$ by one. The loop then continues until
it calls \textsc{SplitRight,} replacing a single $\mathtt{0}$ on
the right end of the span, and then increases $\Delta(Span[j])$ by
one and decreases $Need[j]$ by one. In either possibility, $\Delta(Span[j])$
is increased by the same amount as $Need[j]$ is decreased, so their
sum is unchanged and the relation holds. Note also that if there were
one $\mathtt{0}$ more on the right side of the span, that will still
be the case after the pass through the loop since growth in this case
is symmetric.
\item $Need[j]=1$, in which case there must be a single $\mathtt{0}$ on
the right end of the span. The algorithm calls\textsc{ SplitRight},
$\Delta(Span[j])$ is increased by one on the right end, and $Need[j]$
is decreased by one. Since $Span[j]$ is increased by the same amount
$Need[j]$ is decreased, their sum is unchanged and the relation holds.
\item $Need[j]=0$, in which case no changes are made to $Span[j]$ or $Need[j]$
and so the relation still holds.
\end{itemize}
The for loop on line 8 repeats this process for all $j\leq v$. Since
the inductive step holds, the relation is true for every iteration
of the while loop.

It is clear that each $Span[j]$ grows by replacing $\mathtt{0}$s
with $j$ each pass through the loop until $\Delta(Span[j])=\#(j)$
and $Need[j]=0$. Recall that $Need[j]$ began as the number of $\mathtt{0}$s
in span $j$. Therefore when $Need[j]=0$, the span must contain $j^{\#(j)}$,
and so $V$ must contain $\langle1^{\#(1)},2^{\#(2)},\ldots,v^{\#(v)}\rangle$
when $\Sigma_{j}Need[j]=0$.

Next, we look at \textsc{$D$} returned by\textsc{ \nameref{alg:FindDestinations}}
and show that $V[i]=j$ if and only if $\langle\phi\rangle[D[i]][1]=j$.
First, Assume $V[i]=j$ and let $i-k$ be the left-most occurrence
of $j$ in $V$ for some $k$. Then there must be $j-1$ other assignment
spans that appear to the left of $V[i-k]$. By the construction of
$D$, $D[i]$ must be in $Dest[j]$ created on line 3 since $D[1,\ldots,i-k-1]$
is the concatenation of the first $j-1$ arrays in $Dest$ with the
$j^{\text{th}}$ array to follow. But $D[i]\in Dest[j]$ implies there
was discovered a literal in $\langle\phi\rangle$ having variable
$j$ at position $i$ by the construction of $Dest[j]$ on line 3.
Now suppose instead that $\langle\phi\rangle[D[i]][1]=j$. Then $Dest[j]$
contains $D[i]$, and so there is exactly $j-1$ assignment spans
in $V$ before index $i$. Being in the $j^{\text{th}}$ span of $V$,
$V[i]$ must have the value $j$. 

Lastly, we examine the result of \nameref{alg:OrderAssignments}.
Ignoring the calls to \nameref{pro:Remain}, which do not affect $D$
or $V$, it is clear that \nameref{alg:OrderAssignments} is exactly
Odd-Even Sort on the elements of $D$ with every swap carried out
by \nameref{pro:Swap}. Considering for the moment only the effect
on $D$ and $V$ in \nameref{pro:Swap}, it is clear that each change
to one is duplicated in the other so that if $D[i]$ and $D[i+1]$
are swapped on any given pass through the array, $V[i]$ and $V[i+1]$
are also swapped and not otherwise. Since each $D[i]$ contains the
index of a literal in $\langle\phi\rangle$ which requires the assignment
$V[i]$, sorting the elements of $D[i]$ so that the corresponding
elements in $V[i]$ are moved in exactly the same way must result
in $D=\langle1,2,\ldots,n\rangle$ and correspondingly $V=\langle\langle\phi\rangle[1][1],\langle\phi\rangle[2][1],\ldots,\langle\phi\rangle[n][1]\rangle$.
This proves the lemma.\end{proof}
\begin{lem}
\label{lem:k-2Passes}Let $(A,k)$ be the result of \nameref{alg:3SatToFsca}
on input $\langle\phi\rangle$. \nameref{alg:3SatToFsca} makes $k-2$
passes through the array $V$ and, for each pass, adds exactly one
state to each $a_{i}\in A$.\end{lem}
\begin{proof}
To see this, it is first useful to note that states are only added
in procedures \nameref{pro:Remain}, \nameref{pro:SplitLeftRight},
\nameref{pro:SplitLeft}, \nameref{pro:SplitRight}, and \nameref{pro:Swap}.
Further, each of these adds exactly one new state to the cells they
affect: \nameref{pro:Remain} affects only the cell $a_{i}$ for the
given parameter $i$, \nameref{pro:SplitLeftRight} affects cells
$a_{i-1}$, $a_{i}$, and $a_{i+1}$, \nameref{pro:SplitLeft} affects
cells $a_{i-1}$ and $a_{i}$, and \nameref{pro:SplitRight} and \nameref{pro:Swap}
both affect cells $a_{i}$ and $a_{i+1}$. In all cases, a single
new state $q_{i,t}$ is added to $Q_{i}$, a transition set from $Last[i]$
to $q_{i,t}$ is defined for all possible inputs to the cell, and
$Last[i]$ is updated to refer to $q_{i,t}$ for all affected $a_{i}$.
Therefore it is sufficient to show that one and only one of these
procedures is called to affect, or \emph{cover}, each cell on every
pass through $V$. For this, we need only to consider only \nameref{alg:GrowSpans}
and \nameref{alg:OrderAssignments} as $V$ is not accessed elsewhere
after initialization. 

In \nameref{alg:GrowSpans}, if $\Sigma_{j}Need[j]=0$, then the number
of passes equals the number of time steps added to $A$ since both
are 0. Otherwise, for each iteration of the while loop, there must
be a next and, independently, a last changed span $j$ in $V$. Assume
$i$ is set to the index of the first position not yet covered by
one of the state-adding procedures. Suppose the next changed span
$j$ has $Need[j]\geq2$. Then all positions from $i$ to $Span[j][1]-2$
remain unchanged in $V$, and \nameref{pro:Remain} is called for
each of those positions in the loop on line 12. The element of $V$
at position $Span[j][1]-1$ is then changed in one of two ways: either
$\Delta(Span[j])=1$ and so \nameref{pro:SplitLeftRight} is called,
covering positions $Span[j][1]-1,\ldots,Span[j][1]+1$ and $i$ is
set one passed the new end of the span; or $\Delta(Span[j]))>1$,
in which case \nameref{pro:SplitLeft} is called to cover positions
$Span[j][1]-1$ and $Span[j][1]$ and $i$ gets assigned $Span[j][1]+1$.
In this second case, the algorithm will call \nameref{pro:Remain}
for every position up to $Span[j][2]-1$, then cover positions $Span[j][2]$
and $Span[j][2]+1$ with a call to \nameref{pro:SplitRight}, and
finally set $i$ to one passed the new end of the span. In either
case, all positions from the starting value of $i$ up to the new
end of the span are covered and $i$ is set to the first position
not yet covered. This is the condition in which we began.

Suppose instead that the next changed span $j$ has $Need[j]=1$.
Then all positions from 1 to $Span[j][2]-1$ remain unchanged in $V$,
and \nameref{pro:Remain} is called for each of those positions in
the loop on line 26. Positions $Span[j][2]$ and $Span[j][2]+1$ will
be covered with a call to \nameref{pro:SplitRight}, and $i$ will
be set to one passed the new end of the span. Again, we are in the
starting condition. 

Since it is clear for a base case where $i=1$ and the next span to
change has any $Need[j]>0$ that all positions are covered from the
starting $i$ up to the new right end of the span, then by induction
we see that all spans are so covered.

Now suppose the last changed span has been covered. Then $i$ is set
to the first position not yet covered and no other spans with $Need[j]>0$
remain. Then all positions from $i$ to $|V|$ are covered by calls
to \nameref{pro:Remain} in the loop on line 32. Notice that no position
was covered more than once. Therefore, each pass through the while
loop adds exactly 1 state to each $a_{i}$.

Since $t$ is incremented each time through the while loop, and each
pass through the loop adds exactly one state to each cell, $t$ counts
the number of states added in \nameref{alg:GrowSpans}.

In \nameref{alg:OrderAssignments}, coverage is easier to see. Either
positions $i$ and $i+1$ are swapped or they are both covered by
a call to \nameref{pro:Remain} in both the odd and even phase of
the sort. In the even phase, position 1 is covered explicitly on line
18, and in both phases, any unpaired element at the end is covered
conditionally in lines 13 and 27 with calls to \nameref{pro:Remain}.
Therefore, each pass of \nameref{alg:OrderAssignments} adds exactly
1 state to each $a_{i}$.

Since $t$ is passed to \nameref{alg:OrderAssignments} holding the
number of states added in \nameref{alg:GrowSpans}, and $t$ is incremented
for each pass of the odd-even sort, and each pass adds exactly one
state to each cell, $t$ counts the number of states added in \nameref{alg:GrowSpans}
and \nameref{alg:OrderAssignments}. Since $k=t+2$ for the states
added to each cell by the loop on line 24, \nameref{alg:GrowSpans}
and \nameref{alg:OrderAssignments} make exactly $k-2$ passes through
$V$.\end{proof}
\begin{cor}
$A$ is simple.\end{cor}
\begin{proof}
Since only one state is added to each cell for each time step; each
newly added state becomes the last state; and only simple transitions
are added from the last state to any new state, it is easy to see
by induction that $A$ is simple for all $t\leq k-2$. Finally, it
is clear that the last three transition sets added by lines 27-35
only take $a_{i}$ from $Last[i]$ to $q_{i,t+1}$, from $q_{i,t+1}$
to $q_{i,t+2}$, and from $q_{i,t+2}$ to $q_{i,t+2}$ respectively.
Therefore, $A$ is simple for all $t$. 
\end{proof}
\begin{figure*}[th!]

\begin{minipage}[t]{0.48\columnwidth}%
\begin{procedure}[H]
\begin{raggedright}
\textbf{Parameters:} $A,V,i,t,Last$
\par\end{raggedright}
\begin{algor}[1]
\item [{.}] $V[i-1]\gets V[i]$
\item [{.}] Create a state $q_{i,t}$ in $a_{i}$
\item [{.}] $Q_{i}\gets Q_{i}\cup\{q_{i,t}\}$
\item [{.}] $\delta_{i}\gets\delta_{i}\,\cup\, T_{\omega}(Last[i],q_{i,t})$
\item [{.}] $Last[i]\gets q_{i,t}$ 
\item [{.}] Create a state $q_{i-1,t}$ in $a_{i-1}$
\item [{.}] $Q_{i-1}\gets Q_{i-1}\cup\{q_{i-1,t}\}$
\item [{.}] $\delta_{i-1}\gets\delta_{i-1}\,\cup\, T_{\rho}(Last[i-1],q_{i-1,t})$
\item [{.}] $Last[i-1]\gets q_{i-1,t}$ 
\end{algor}
\caption{\label{pro:SplitLeft}\textsc{SplitLeft}}
\end{procedure}
\end{minipage}\hfill{}%
\begin{minipage}[t]{0.48\columnwidth}%
\begin{procedure}[H]
\begin{raggedright}
\textbf{Parameters:} $A,V,i,t,Last$
\par\end{raggedright}
\begin{algor}[1]
\item [{.}] $V[i+1]\gets V[i]$
\item [{.}] Create a state $q_{i,t}$ in $a_{i}$
\item [{.}] $Q_{i}\gets Q_{i}\cup\{q_{i,t}\}$
\item [{.}] $\delta_{i}\gets\delta_{i}\,\cup\, T_{\omega}(Last[i],q_{i,t})$
\item [{.}] $Last[i]\gets q_{i,t}$ 
\item [{.}] Create a state $q_{i+1,t}$ in $a_{i+1}$
\item [{.}] $Q_{i+1}\gets Q_{i+1}\cup\{q_{i+1,t}\}$
\item [{.}] $\delta_{i+1}\gets\delta_{i+1}\,\cup\, T_{\lambda}(Last[i+1],q_{i+1,t})$
\item [{.}] $Last[i+1]\gets q_{i+1,t}$ 
\end{algor}
\caption{\label{pro:SplitRight}\textsc{SplitRight}}
\end{procedure}
\end{minipage}

\begin{minipage}[t]{0.48\columnwidth}%
\begin{procedure}[H]
\begin{raggedright}
\textbf{Parameters:} $A,V,i,t,Last$
\par\end{raggedright}
\begin{algor}[1]
\item [{.}] $V[i-1]\gets V[i]$;$V[i+1]\gets V[i]$
\item [{.}] Create a state $q_{i,t}$ in $a_{i}$
\item [{.}] $Q_{i}\gets Q_{i}\cup\{q_{i,t}\}$
\item [{.}] $\delta_{i}\gets\delta_{i}\,\cup\, T_{\omega}(Last[i],q_{i,t})$
\item [{.}] $Last[i]\gets q_{i,t}$ 
\item [{.}] Create a state $q_{i-1,t}$ in $a_{i-1}$
\item [{.}] $Q_{i-1}\gets Q_{i-1}\cup\{q_{i-1,t}\}$
\item [{.}] $\delta_{i-1}\gets\delta_{i-1}\,\cup\, T_{\rho}(Last[i-1],q_{i-1,t})$
\item [{.}] $Last[i-1]\gets q_{i-1,t}$ 
\item [{.}] Create a state $q_{i+1,t}$ in $a_{i+1}$
\item [{.}] $Q_{i+1}\gets Q_{i+1}\cup\{q_{i+1,t}\}$
\item [{.}] $\delta_{i+1}\gets\delta_{i+1}\,\cup\, T_{\lambda}(Last[i+1],q_{i+1,t})$
\item [{.}] $Last[i+1]\gets q_{i+1,t}$ 
\end{algor}
\caption{\label{pro:SplitLeftRight}\textsc{SplitLeftRight}}
\end{procedure}
\end{minipage}\hfill{}%
\begin{minipage}[t]{0.48\columnwidth}%
\begin{procedure}[H]
\begin{raggedright}
\textbf{Parameters:} $A,V,D,i,t,Last$
\par\end{raggedright}
\begin{algor}[1]
\item [{.}] $tmp\gets V[i]$
\item [{.}] $V[i]\gets V[i+1]$
\item [{.}] $V[i+1]\gets tmp$
\item [{.}] $tmp\gets D[i]$
\item [{.}] $D[i]\gets D[i+1]$
\item [{.}] $D[i+1]\gets tmp$
\item [{.}] Create a state $q_{i,t}$ in $a_{i}$
\item [{.}] $Q_{i}\gets Q_{i}\cup\{q_{i,t}\}$
\item [{.}] $\delta_{i}\gets\delta_{i}\,\cup\, T_{\rho}(Last[i],q_{i,t})$
\item [{.}] $Last[i]\gets q_{i,t}$ 
\item [{.}] Create a state $q_{i+1,t}$ in $a_{i+1}$
\item [{.}] $Q_{i+1}\gets Q_{i+1}\cup\{q_{i+1,t}\}$
\item [{.}] $\delta_{i+1}\gets\delta_{i+1}\,\cup\, T_{\lambda}(Last[i+1],q_{i+1,t})$
\item [{.}] $Last[i+1]\gets q_{i+1,t}$ 
\end{algor}
\caption{\label{pro:Swap}\textsc{Swap}}
\end{procedure}
\end{minipage}

\begin{minipage}[t]{0.48\columnwidth}%
\begin{center}
\begin{procedure}[H]
\begin{raggedright}
\textbf{Parameters:} $A,i,t,Last$
\par\end{raggedright}
\begin{algor}[1]
\item [{.}] Create a state $q_{i,t}$ in $a_{i}$
\item [{.}] $Q_{i}\gets Q_{i}\cup\{q_{i,t}\}$
\item [{.}] $\delta_{i}\gets\delta_{i}\,\cup\, T_{\omega}(Last[i],q_{i,t})$
\item [{.}] $Last[i]\gets q_{i,t}$ 
\end{algor}
\caption{\textsc{\label{pro:Remain}Remain}}
\end{procedure}

\par\end{center}%
\end{minipage}\hfill{}

\end{figure*}
\begin{lem}
\label{lem:AhasVars}Let $\alpha_{1},\alpha_{2},\ldots,\alpha_{v}$
be assignments to $x_{1},x_{2},\ldots,x_{v}$ in $\phi$ and let $(A,k)$
be the result of \nameref{alg:3SatToFsca} on input $\langle\phi\rangle$.
If $A$ is provided the initial values in $\mathbf{s}$ defined by
\begin{equation}
\mathbf{s}=\left\langle \mathtt{0}^{\left\lfloor \frac{1}{2}\left(\#(j)-1\right)\right\rfloor }\alpha_{j}\mathtt{0}^{\left\lceil \frac{1}{2}\left(\#(j)-1\right)\right\rceil }\right\rangle _{j=1}^{v}\label{eq:sInit}
\end{equation}
and operated for $k-2$ time steps, the resulting value of $A$ is
$\left\langle \alpha_{\langle\phi\rangle[1][1]},\alpha_{\langle\phi\rangle[2][1]},\ldots,\alpha_{\langle\phi\rangle[n][1]},\right\rangle $. \end{lem}
\begin{proof}
We will denote the value in $V[i]$ at pass $t$ by $V_{i}^{(t)}$
to make clear differences in value of the same location in different
passes. Similarly, we will use $\omega_{i}^{(t)}$ to note the output
of cell $a_{i}$ at time step $t$ . We show by induction that $V_{i}^{(t)}=j$
implies that $\omega_{i}^{(t)}=\alpha_{j}$ for $t\leq k-2$. For
the base case of $t=0,$ the implication holds by comparison of (\ref{eq:VInit})
and (\ref{eq:sInit}). 

Assume the implication holds for arbitrary $t<k-2$. To show the implication
holds for $t+1$, we must consider two cases. Let $t_{E}$ be the
value of $t$ returned from \nameref{alg:GrowSpans}. First, suppose
$t<t_{E}$, in which case pass $t+1$ will be made inside \nameref{alg:GrowSpans}.
We know from \lemref{k-2Passes} that every element in $V$ is either
explicitly changed or it is explicitly not changed as the index $i$
ranges over the positions of $V$. If an element is changed, the change
must happen in a call to one of \nameref{pro:SplitLeftRight}, \nameref{pro:SplitLeft},
or \nameref{pro:SplitRight}. If it is \emph{not} changed, \nameref{pro:Remain}
is called for position $i$. We examine each in turn.

In \nameref{pro:SplitLeftRight}, notice that $V_{i-1}^{(t+1)}$ is
assigned the value of the position to its right, $V_{i+1}^{(t+1)}$
is assigned the value of the position to its left, and $V_{i}^{(t+1)}$
keeps its previous value. So $V_{i-1}^{(t+1)}=V_{i+1}^{(t+1)}=V_{i}^{(t+1)}=V_{i}^{(t)}$.
Correspondingly, the new transition set in $a_{i-1}$, on any input,
takes the value of its right neighbor so that $\omega_{i-1}^{(t+1)}=\omega_{i}^{(t)}$.
We know by the inductive hypothesis that a time step $t,$ $\omega_{i}^{(t)}=\alpha_{V_{i}^{(t)}}$,
so then $\omega_{i-1}^{(t+1)}=\alpha_{V_{i}^{(t)}}$. We also know
that $V_{i-1}^{(t+1)}=V_{i}^{(t)}$, so it must be that $\omega_{i-1}^{(t+1)}=\alpha_{V_{i-1}^{(t+1)}}$.
Likewise, $a_{i+1}$ gets a new state and transition which takes its
left neighbor's value, so $\omega_{i+1}^{(t+1)}=\omega_{i}^{(t)}=\alpha_{V_{i}^{(t)}}$
by hypothesis. But we also have $V_{i+1}^{(t+1)}=V_{i}^{(t)}$, so
$\omega_{i+1}^{(t+1)}=\alpha_{V_{i+1}^{(t+1)}}$. Finally, $a_{i}$
gets a new state and transition which keeps $\omega_{i}$ constant,
and so $\omega_{i}^{(t+1)}=\omega_{i}^{(t)}$ which is by hypothesis
equal to $\alpha_{V_{i}^{(t)}}$. And since $V_{i}^{(t)}=V_{i}^{(t+1)}$,
we have $\omega_{i}^{(t+1)}=\alpha_{V_{i}^{(t+1)}}$. The implication
holds for all cells affected by \nameref{pro:SplitLeftRight}.

Similar arguments show that the implication also holds for \nameref{pro:SplitLeft},
\nameref{pro:SplitRight}, and \nameref{pro:Remain}. Thus the hypothesis
is true for $t<t_{E}$.

Now suppose $t\geq t_{E}$. Then pass $t+1$ will be made in \nameref{alg:OrderAssignments}.
As shown in \lemref{k-2Passes}, every element in $V$ is either changed
in \nameref{pro:Swap} or it is left unchanged, in which case \nameref{pro:Remain}
is called. We have already shown the implication holds in \nameref{pro:Remain},
so we have only \nameref{pro:Swap} to contend with. The argument
is very similar to the one above. Clearly $V_{i}^{(t+1)}=V_{i+1}^{(t)}$
and $V_{i+1}^{(t+1)}=V_{i}^{(t)}$. Since $a_{i}$ adds the the transition
set $T_{\lambda}$ from $Last[i]$ to $q_{i,t+1}$, we know $\omega_{i}^{(t+1)}=\omega_{i+1}^{(t)}$.
Likewise, for $a_{i+1}$, we know $\omega_{i+1}^{(t+1)}=\omega_{i}^{(t)}$.
By the induction hypothesis, we have $\omega_{i}^{(t)}=\alpha_{V_{i}^{(t)}}$
and $\omega_{i+1}^{(t)}=\alpha_{V_{i+1}^{(t)}}$. Therefore, $\omega_{i+1}^{(t+1)}=\omega_{i}^{(t)}=\alpha_{V_{i}^{(t)}}=\alpha_{V_{i+1}^{(t+1)}}$
and $\omega_{i}^{(t+1)}=\omega_{i+1}^{(t)}=\alpha_{V_{i+1}^{(t)}}=\alpha_{V_{i}^{(t+1)}}$,
so the implication holds for \nameref{pro:Swap} as well, and by extension,
for all $t_{E}\leq t\leq k-2$.

Since $(V_{i}^{(t)}=j)\implies(\omega_{i}^{(t)}=\alpha_{j})$, then
$A$ at time $k-2$ has the value $\langle\omega_{1}^{(k-2)},\omega_{2}^{(k-2)},\ldots,\omega_{n}^{(k-2)}\rangle=$
$\langle\alpha_{V_{1}^{(k-2)}},\alpha_{V_{2}^{(k-2)}},\ldots,\alpha_{V_{n}^{(k-2)}}\rangle$.
By \lemref{VisPhi}, $V=\langle\langle\phi\rangle[1][1],\langle\phi\rangle[2][1],\ldots,\langle\phi\rangle[n][1]\rangle$
after its last pass and by \lemref{k-2Passes} there are $k-2$ passes
through $V$. Therefore, $A$ at time $k-2$ has the value $\left\langle \alpha_{\langle\phi\rangle[1][1]},\alpha_{\langle\phi\rangle[2][1]},\ldots,\alpha_{\langle\phi\rangle[n][1]}\right\rangle $.

\end{proof}
We can now complete the proof of \thmref{AevalsPhi}. Let $(A,k)$
be the result of \nameref{alg:3SatToFsca} on input $\langle\phi\rangle$
and let $\alpha_{j}=\mathtt{1}$ encode a $\mathtt{true}$ assignment
and $\alpha_{j}=\mathtt{0}$ encode a $\mathtt{false}$ assignment
to $x_{j}$. Define the initial values for $A$ as 
\[
\mathbf{s}=\left\langle \mathtt{0}^{\left\lfloor \frac{1}{2}\left(\#(j)-1\right)\right\rfloor }\alpha_{j}\mathtt{0}^{\left\lceil \frac{1}{2}\left(\#(j)-1\right)\right\rceil }\right\rangle _{j=1}^{v}.
\]
We claim $(\langle q_{0,0},q_{1,0},\ldots,q_{n,0}\rangle,\mathbf{s})\vdash_{\! A}^{k}(\langle q_{0,k},q_{1,k},\ldots,q_{n,k}\rangle,\langle(\mathtt{010})^{c}\rangle)$
if and only iff $\phi$ is satisfied by the assignments $\alpha_{1},\ldots,\alpha_{v}$.

First, assume the assignments encoded as $\alpha_{1},\ldots,\alpha_{v}$
satisfy $\phi$. By \lemref{AhasVars}, we know the $k-2^{\text{nd}}$
value of $A$ consists of the values $\alpha_{1},\ldots,\alpha_{v}$
ordered as they appear in $\phi.$ Lines 24 through 30 in \nameref{alg:3SatToFsca}
make it clear that each $a_{i}$ will complement its value if and
only if the corresponding literal in $\phi$ is complemented, and
so will have the opposite value at time $k-1.$ Therefore, the value
of $A$ at time $k-1$ is exactly the encoded literals of $\phi$
evaluated for the assignments $\alpha_{1},\ldots,\alpha_{v}$. Lines
31 through 34 show that every third cell starting with the second
cell will compute the OR function, and all other cells will compute
the 0 function at time $k$. So then $\omega_{i}$ will take the value
$\omega_{i-1}\vee\omega_{i}\vee\omega_{i+1}$ for $i\equiv2\bmod3$
and 0 otherwise. But these ORs exactly evaluate the clauses of $\phi$
when considering it's literals as an array of length $n=3c$. Since
$\phi$ is satisfied by $\alpha_{1},\ldots,\alpha_{v}$, each such
ORing of literals in this grouping must result in $\mathtt{true}$,
and so the OR of their encoding must be $\mathtt{1}$. This implies
the value of $a_{i}$ will be $\mathtt{1}$ for $i\equiv2\bmod3$
and $\mathtt{0}$ otherwise. Therefore, after $k$ time steps, $A$
has the configuration $(\langle q_{0,k},q_{1,k},\ldots,q_{n,k}\rangle,\langle(\mathtt{010})^{c}\rangle)$.

Conversely, suppose $A$ with initial value $\mathbf{s}$ has the
configuration $(\langle q_{0,k},q_{1,k},\ldots,q_{n,k}\rangle$, $\langle(\mathtt{010})^{c}\rangle)$
after $k$ time steps. We know by the construction of each cell on
lines 31 through 34 that those cells $a_{i}$ for $i\equiv2\bmod3$
with output value $\omega_{i}=\mathtt{1}$ are the result of the OR
of their three inputs from the previous time step, and so $\omega_{i}^{(k)}=\omega_{i-1}^{(k-1)}\vee\omega_{i}^{(k-1)}\vee\omega_{i+1}^{(k-1)}$.
We also know that $\omega_{i}^{(k-2)}=\alpha_{\langle\phi\rangle[i][1]}$
and that $\omega_{i}^{(k-1)}=\overline{\omega_{i}^{(k-2)}}$ if and
only if the $i^{\text{th}}$ literal in $\phi$ is complemented. Then
the values of $a_{i}$ at time step $k-1$ are exactly the corresponding
literals of $\phi$ when evaluated for the assignments $\alpha_{1},\ldots,\alpha_{v}$.
Since $\omega_{i-1}^{(k-1)}\vee\omega_{i}^{(k-1)}\vee\omega_{i+1}^{(k-1)}=\mathtt{1}$,
the clauses over those literals must also evaluate to $\mathtt{true}$,
and so $\phi$ is satisfied.

Since $a_{i}$ has value $\mathtt{1}$ for $i\equiv2\bmod3$ only
when $\phi$ is satisfied, $A$ correctly evaluates the clauses of
$\phi$.
\end{proof}

\subsection{Comparing Computational Ability}

It can be difficult to get a sense of the computational efficiency
of simple FSCA since it is a parallel construction on one hand, but
a construction of machines much less powerful than Turing machines
on the other. To facilitate a comparison, we will consider the number
of time steps required to perform certain computations relative to
other computational models. A Turing machine with alphabet $\Sigma=\{0,1\}$,
for instance, would require roughly $5n/3$ operations to evaluate
the clauses of a 3-CNF formula having $n$ literals: checking and
inverting each literal ($n$ operations), then performing two OR operations
for every three literals ($2n/3$ operations). This number would grow
at least by a factor of $n$ if we considered individual head movements.
Our FSCA as constructed above, however, does somewhat better (without
using the cyclic boundary property.) We examine this formally in the
following theorem.
\begin{thm}
\label{thm:k<3n/2}Let $(A,k)$ be as returned from \nameref{alg:3SatToFsca}
on input $\langle\phi\rangle$ such that $A$ has $n$ cells. $k\leq3n/2+2$.\end{thm}
\begin{proof}
\lemref{k-2Passes} shows that there is one state in each $a_{i}$
for every pass through $V$, and that there are exactly $k-2$ such
passes. We simply bound $k-2$ as a function of $n$. This is complicated
somewhat by having no fixed relationship between the number of clauses
and the number of variables. However, a coarse bound is still possible.

In \nameref{alg:GrowSpans}, we note that each span can grow by two
at every time step except possibly its last (if its ultimate size
is even.) At worst, there is only one variable that needs to grow,
and that variable is centered in its span by (\ref{eq:VInit}). Since
the variable occupies one element of $V$ at initialization (by \lemref{VisPhi}),
\nameref{alg:GrowSpans} can require no more than $\left\lceil (n-1)/2\right\rceil \leq n/2$
passes to fill the entire array. 

The remaining passes are made by \nameref{alg:OrderAssignments}.
Recall that this algorithm functions exactly as Odd-Even sort. It
is well known than Odd-Even sort can sort $n$ variables in $n$ passes. 

Finally, $k$ is assigned the number of passes made in \nameref{alg:GrowSpans}
and \nameref{alg:OrderAssignments} plus 2. Therefore, $k\leq n/2+n+2=3n/2+2$,
and so the bound holds. 
\end{proof}
It appears in the case of evaluating 3-CNF formulas, even simple FSCA
are capable of reasonably efficient operation.

Comparing simple FSCA to elementary FSCA, it is not clear an elementary
FSCA could be built to evaluate a 3-CNF efficiently. If a cell is
constructed to perform the OR of a clause, then it cannot also invert
a literal or exchange a value with a neighbor as each cell is allowed
only one function. Any such solution would require a more clever mixture
of functions over neighborhoods of cells, and may end up relying on
the ability to simulate a Turing machine, as shown for rule 110. \cite{Cook04}

\subsection{Invertibility of FSCA}

Having formally defined FSCA and examined some of their computational
capability, we return to the issue of invertibility. We first need
a formal definition of the problem, which we provide for the general
case of FSCA. We define the $k$-INVERT decision problem as follows:
\[
k\text{-INVERT}=\{(A,C^{(t)},k)\:|\: A\text{ is an FSCA of }n\text{ cells},\exists\, C^{(t-k)}(C^{(t-k)}\vdash_{\! A}^{k}C^{(t)})\text{ for }k\in\mathbb{N}\}
\]

We would like to know whether deciding $k$-INVERT is NP-Hard. If
so, then an FSCA could serve as a primitive on which to build a provably
secure PRG provided its construction and operation are efficient.
The efficiency condition motivates us to seek the simplest FSCA for
which inversion is provably NP-Hard. 

Elementary FSCAs do not inspire great confidence in this regard. Apart
from the attacks already shown, the difficulty in evaluating a simple
Boolean formula suggests a fundamental lack of ability to withstand
analysis from an opponent armed with a Turing machine.

Simple FSCA, however, seem to hold more promise. If running such a
machine forward from an assignment produces the evaluation of a formula,
then running it backwards (i.e. inverting its operation) from an evaluation
must produce an assignment. If the evaluation were a satisfying one,
such an ability could be used to decide 3-SAT. This might be a little
surprising. Notice that the simple restriction reduces the $k$-INVERT
problem to just finding an appropriate $\mathbf{s}^{(t-k)}$, since
$\mathbf{q}^{(t-k)}$ is easily deduced from $\mathbf{q}^{(t)}$:
simply follow the transitions backwards from $\mathbf{q}^{(t)}$ for
$k$ states. However, it may be that the mixing behavior of the 3-neighbor
construction over sufficiently many time steps provides the hardness
we need. Therefore, we will examine simple FSCAs to determine if there
is a $B(n)$ for which they are NP-Hard to $k$-invert.

We now formalize this intuition in a reduction from 3-SAT to a $k$-INVERT
variant for simple FSCA. 
\begin{thm}
\label{thm:kSiNpComplete}Let $k$-SIMPLE-INVERT, abbreviated $k$SI,
be defined as
\[
\{(A,C^{(t)},k)\:|\: A\text{ is a simple FSCA of }n\text{ cells},\exists\, C^{(t-k)}(C^{(t-k)}\vdash_{\! A}^{k}C^{(t)}),k\in\mathbb{N}\}.
\]
$k$SI is NP-Complete for $k\geq3n/2+2$.\end{thm}
\begin{proof}
We show that $k$SI $\in$ NP and that 3-SAT is polynomial-time reducible
to $k$SI. Since 3-SAT is NP-Complete, this will prove $k$SI is NP-Complete. 

The first condition is easy to see. Recall that $\mathbf{q}^{(t-k)}$
is easy to deduce from $\mathbf{q}^{(t)}$ and $A$ which are encoded
in the input string. We can create a decider for $k$SI which, given
a certificate $\mathbf{s}^{(t-k)}$, performs the following: run $A$
with configuration $(\mathbf{q}^{(t-k)},\mathbf{s}^{(t-k)})$ for
$k$ time steps. If the result is $C^{(t)}$, accept. Otherwise, reject.
Since this decider requires only $O(kn)$ steps, we can verify $k$SI
in polynomial time, and so $k$SI is in NP. 

Now we show 3-SAT $\leq_{\text{P}}$ $k$SI. Let $\phi$ be a 3-CNF
formula having $c$ unique clauses over $v$ variables such that no
clause repeats. Let $n=3c$. Consider \algref{3-SatToSi}, \nameref{alg:3-SatToSi}.

\begin{algorithm}
\textbf{Input:} $\langle\phi\rangle$, an encoding of $\phi$.

\textbf{Output:} FSCA $A$, a configuration $C$, and $\kappa\in\Nat$ 
\begin{algor}[1]
\item [{.}] $(A,\kappa)\gets$\nameref{alg:3SatToFsca}$(\langle\phi\rangle)$
\item [{.}] $C\gets(\mathbf{q}^{(\kappa)},\langle(\mathtt{010})^{c}\rangle)$
\item [{.}] \textbf{return} $(A,C,\kappa)$ 
\end{algor}
\caption{\label{alg:3-SatToSi}3-\textsc{SatTo$k$Si}}
\end{algorithm}

First, note that \nameref{alg:3-SatToSi} runs in time $O(n^{2})$
since, by \lemref{k-2Passes}, \nameref{alg:3SatToFsca} makes $\kappa$
(which is $O(n)$) passes over an array of length $n$ and $C$ can
be created in time $O(n)$. 

Next, we show that the returned $(A,C,\kappa)\in k\text{SI}\iff\phi\in$
3-SAT. Suppose $(A,C,\kappa)\in$ $k$SI and that $k=\kappa$. Then
there exists a $C^{(0)}$ such that $C^{(0)}\vdash_{\! A}^{\kappa}C$.
As $C$ has a value where each $\mathbf{s}_{i}=\mathtt{1}$ for $i\equiv2\bmod3$,
and since by \thmref{AevalsPhi} $A$ evaluates $\phi$, $\phi$ must
be satisfiable and so $\phi\in$ 3-SAT. Now consider any $k>\kappa$.
By the construction of $A,$ the value after time step $\kappa$ never
changes since all $q_{i,\kappa}$ transition only to themselves, keeping
the same value. If $A$ has value $\langle(\mathtt{010})^{c}\rangle$
at time step $k$, it must have had the same value at time step $\kappa$,
and so $\phi$ is satisfiable and in 3-SAT.

Suppose $\phi\in$ 3-SAT and $k=\kappa$. Then there is a satisfying
assignment, $\alpha_{1},\ldots,\alpha_{v}$ for $\phi$. Let 
\[
\mathbf{s}=\left\langle \mathtt{0}^{\left\lfloor \frac{1}{2}\left(\#(j)-1\right)\right\rfloor }\alpha_{j}\mathtt{0}^{\left\lceil \frac{1}{2}\left(\#(j)-1\right)\right\rceil }\right\rangle _{j=1}^{v}
\]
Since by \thmref{AevalsPhi} $A$ evaluates $\phi$, then it must
be that $(\mathbf{q}^{(0)},\mathbf{s}))\vdash_{\! A}^{\kappa}(\mathbf{q}^{(\kappa)},\langle(\mathtt{010})^{c}\rangle)=C$,
and so $(A,C,\kappa)\in k\text{SI}$. Since $C^{(\kappa)}\vdash_{\! A}^{*}C^{(\kappa)}$,
$(A,C,\kappa)\in k\text{SI}$ for all $k>\kappa$.

Finally, by \thmref{k<3n/2}, we know that $\kappa\leq3n/2+2$. 
\end{proof}
We can generalize this a bit further using the technique of padding
as often done for other NP-Complete problems. Notice that \nameref{alg:3SatToFsca}
constructs a cyclic boundary FSCA, but never makes use of the boundary
connections. We can therefore break the boundary connections, insert
dummy cells, and connect those cyclically without affecting the operation
of the FSCA in the original $n$ cells. This changes the number of
cells while leaving the number of time steps constant, allowing the
ratio between the two to be an arbitrary one. This gives us the following
theorem.
\begin{thm}
\label{thm:kSigSiHard}$k$SI is NP-Complete for $k\geq n/\sigma$
for any arbitrary $\sigma\in\Nat.$\end{thm}
\begin{proof}
We again reduce 3-SAT to $k$SI. Let $\phi$ be a 3-CNF formula having
$c$ clauses. Our goal is to construct an FSCA of $n$ cells which
evaluates $\phi$ and a $k\in\Nat$ such that $n/\sigma\leq k$. Consider
the following algorithm:

\begin{algorithm}
\textbf{Input:} $\langle\phi\rangle$, an encoding of $\phi$,

\qquad{} $\sigma\in\Nat$.

\textbf{Output:} FSCA $A$, a configuration $C$, and $k\in\Nat$ 
\begin{algor}[1]
\item [{.}] $(A',\kappa)\gets$\nameref{alg:3SatToFsca}$(\langle\phi\rangle)$
\item [{.}] $n\gets(9c\sigma/2)+2\sigma$
\item [{for}] $i=1,\ldots,n-3c$ \{\{create $n-3c$ dummy cells\}\}

\begin{algor}[1]
\item [{.}] Create a cell $a_{3c+i}$ 
\item [{.}] Add states $q_{3c+i,0},\ldots,q_{3c+i,\kappa}$ with transition
sets $T_{0}$$(q_{3c+i,t},q_{3c+i,t+1})$, $1\leq t<\kappa$
\item [{.}] Add transition set $T_{0}$$(q_{3c+i,\kappa},q_{3c+i,\kappa})$
\item [{.}] $\rho_{3c+i-1}\gets\omega_{3c+i}$
\item [{.}] $\lambda_{3c+i}\gets\omega_{3c+i-1}$
\end{algor}
\item [{endfor}]~
\item [{.}] $\rho_{1}\gets\omega_{3c+n}$
\item [{.}] $\lambda_{3c+n}\gets\omega_{1}$
\item [{.}] $A\gets\langle a_{1},\ldots,a_{3c},\ldots,a_{3c+n}\rangle$
\item [{.}] $C\gets(\mathbf{q}^{(\kappa)},\langle(\mathtt{010})^{c}\mathtt{0}^{n-3c}\rangle)$
\item [{.}] \textbf{return} $(A,C,\kappa)$ 
\end{algor}
\caption{\label{alg:3-SatToSigKSi}3\textsc{-SatTo$\sigma k$Si}}
\end{algorithm}

\nameref{alg:3-SatToSigKSi} performs $O\left(9\sigma c\right)+O(3c)$
operations, and so runs in time polynomial in $c$. Further, $A$
clearly evaluates $\phi$ since $A'$ evaluates $\phi$ without any
communication from $\lambda_{1}$ or $\rho_{3c}$ by \thmref{AevalsPhi}.
Thus $A$ reaches the $\kappa^{\text{th}}$ time step with a value
of $(\mathtt{010})^{c}\mathtt{0}^{n-3c}$ iff there is an assignment
$\alpha_{1},\ldots,\alpha_{v}$ to the $v$ variables of $\phi$ which
satisfy it. Lastly, by \thmref{k<3n/2}, we know that $\kappa\leq3(3c)/2+2$.
Since $n=9c\sigma/2+2\sigma=\sigma\kappa$, and $A$ has a total of
$3c$ cells (from $A'$) plus $n-3c$ dummy cells, $A$ has $n=\sigma\kappa$
cells. Therefore, $(A,C,\kappa)$ obeys $\kappa\geq n/\sigma$. As
$A$'s value never changes after time step $\kappa,$ the theorem
holds for all $k\geq\kappa\geq n/\sigma$.
\end{proof}
\pagebreak{}

\section{\label{sec:FscaBasedPrg}A PRG Based on FSCA}
\begin{quotation}
\begin{flushright}
\medskip{}
\textsl{\small \textquotedbl{}One of the most singular characteristics
of the art of deciphering is the strong conviction possessed by every
person, even moderately acquainted with it, that he is able to construct
a cipher which nobody else can decipher.\textquotedbl{}--Charles Babbage.}
\par\end{flushright}{\small \par}
\end{quotation}
\smallskip{}

\subsection{Design of an FSCA-based PRG}

Recall our thought experiment from \secref{NewConstruction}, where
the rule of each cell is selected uniformly at random. Under this
scheme, we cannot know anything about the previous output values even
one time step back, nor can we guess what functions will be used next
and so have no ability to guess the next output. We would like to
approach this ideal behavior to the extent possible in a PRG based
on FSCA. 

While our approximation of the ideal must be lacking in some aspects,
uniform distribution of each cell's outputs is a required property.
This suggests that the functions applied by a cell over time must
on the whole have a uniform distribution as well. Suppose each cell
generates an 8-bit value to specify the function at each time step.
Any single such function may be clearly biased, but if each cell cycles
through all possible 8-bit values, its output distribution is uniform
in the aggregate. This is because the number of \textquoteright{}1\textquoteright{}
bits in any fixed bit position across all possible 8-bit strings is
128, meaning there are 128 functions that yield a '1' on any fixed
3-bit input and 128 that yield a '0' for that same input. We can thus
guarantee uniform distribution by having each cell cycle through all
possible 256 time step functions.

Another desired property from our thought experiment is that of independent,
randomly selected functions. While we do not aspire to truly random
selection and we are limited to cycling through the 256 possible functions,
we can at a minimum apply the functions in the order of some permutation
which resists cryptanalysis. The S-Boxes of cryptographic primitives
provide such permutations, and are commonly studied for resistance
to linear and differential attacks. We choose the S-Box from AES as
a basis for selecting the time step functions \cite{FIPS197}.

Following such a permutation does not achieve independence, though.
At some point, the permutation will choose the 0 function. If all
cells do so at the same time, the value would be fixed at $\mathbf{0}$.
We prevent this by having each cell cycle through the permutation
in a different order and from a different starting position. We assign
to each cell an 8-bit offset value which is taken from the seed. This
offset specifies the starting index in the S-Box when expressed as
a lookup table. The function for the cell is then taken as the offset
value XORed with the S-Box value at the cell's current index. XORing
in a constant causes the time step function to use values from the
S-Box in an order unique to that constant.

One difficulty in requiring the time step function to cycle through
the 256 possible functions is the issue of biased functions at the
end of the cycle. Suppose the last function for a cell in the cycle
is the 0 function. On restarting the cycle, that cell will have its
value stuck at 0. Even moderately biased functions at the end of the
cycle lead to stuck bits in the value for the next cycle. Early experiments
showed that most seeds became periodic after just 1 or 2 cycles. While
the uniform distribution property requires the PRG to operate in cycles
of 256, there is no requirement that they be the \emph{same} cycle.
We can change the cycle simply by changing the offsets. At the end
of each cycle, we generate new offsets by: 
\begin{enumerate}
\item Saving away the current value. 
\item Stepping the FSCA for 8 time steps, saving the value away at each
one to generate 8 bits for each cell.
\item XORing those 8 bits into the offsets for each cell.
\item Restoring the current value.
\end{enumerate}
This process effectively changes the simple FSCA each 256 time steps,
giving each cell a new time step function for each of its 256 states.

There are practical concerns in generating output bits from this construction.
First, a cell's path through its cycle of functions may have periods
of extreme bias, affecting the distribution of its output and that
of its neighbors. If we simply return the FSCA's value during such
a period, we may notice a bias in the resulting output. Further, we
would also be revealing a considerable portion of the stored state
(i.e. the value) of the PRG and so weaken its cryptographic strength.

Generating output by XORing FSCA values separated by a number of time
steps seems to solve these problems. XORing two FSCA values leaves
a slight bias in long sequences, but four seems to be sufficient in
practice to remove bias. Different trade-offs between security and
efficiency can be made here. Outputting the XOR of four FSCA values
would also avoid directly revealing the internals of the PRG. Since
$n$-cell simple FSCA are only hard to $k$-invert for some $k$ as
a function of $n$, we would like to choose a number of time steps
which is derived from $n$ and also balances efficiency and security.
We note that at $t=t_{0}+n/2$, every cell is affected by every value
at $t_{0}$, and so the value $\mathbf{s}^{(t_{0})}$ is fully diffused
over the FSCA. Paranoia inspires us to choose $n$ rather than $n/2$
time steps as a minimum between generating outputs. Combining this
with the necessary conditions to avoid bias, we generate output bits
by XORing together four FSCA values, each separated by $n/4$ time
steps.

Another practical consideration is allowing for regular or even constant
(e.g. all 0) seeds. In this case, the offset is the same for all cells,
resulting in a short period for the FSCA. In a fit of irony, we ensure
entropy and chaos in the offsets by first XORing in the bytes of the
most%
\footnote{Based on an incomplete survey.%
} harmonious algebraic number, $\varphi$, the golden ratio. This ensures
each cell has some disordered offset from the S-Box values.

We name the resulting PRG ``Chasm'', owing to the etymology of the
Greek word ``chaos,'' originally meaning ``void'' or ``chasm.''
We represent the current state of a Chasm PRG as a 5-tuple $(n,\mathbf{s},\mathbf{o},\mathbf{i},c)$
where:
\begin{itemize}
\item $n\in\Nat$ is the number of cells in the PRG
\item $\mathbf{s}\in\{0,1\}^{n}$ is the current value of the underlying
FSCA
\item $\mathbf{o}\in\{0,1\}^{8n}$ is the vector of 8-bit offsets for each
cell, indexed as $\langle\mathbf{o}_{1},\ldots,\mathbf{o}_{8n}\rangle$
where each $\mathbf{o}_{i}\in\{0,1\}^{8}$
\item $\mathbf{i}\in\{0,1\}^{8n}$ is the vector of 8-bit indexes of each
cell in the S-Box permutation, indexed as $\langle\mathbf{i}_{1},\ldots,\mathbf{i}_{n}\rangle$
where each $\mathbf{i}_{i}\in\{0,1\}^{8}$
\item $c\in\Nat$ is the cycle step counter
\end{itemize}
The algorithms \nameref{alg:ChasmInitialize} (\algref{ChasmInitialize}),
\nameref{pro:ChasmTimeStep} (\proref{ChasmTimeStep}), and \nameref{alg:ChasmNext}
(\algref{ChasmNext}) fully specify the Chasm PRG operation. We use
the notation $\mathbf{s}\lll k$ and $\mathbf{s}\ggg k$ to denote
the cyclic shifting (rotating) of $\mathbf{s}$ by $k$ bits to the
left and right respectively. We use $\gg$ as right shift, $\oplus$
as XOR, $\cdot$ as AND, and $+$ as Boolean OR. These operations
are done element-wise when the terms are vectors. We will also use
$\varphi$ to mean the golden ratio $(1+\sqrt{5})/2$ and $\varphi_{i}$
to mean the $i^{\text{th}}$ bit of $\varphi$ when represented in
binary. $B[\mathbf{i}]$ denotes the vector of values $B_{j}$ stored
in the AES S-Box lookup table $B$ for each index $j\in\mathbf{i}$.
See \tabref{AesSBox} for the table values.

\begin{figure*}[th]

\begin{minipage}[t]{0.45\columnwidth}%
\begin{algorithm}[H]
\begin{raggedright}
\textbf{Input:} $\sigma\in\{0,1\}^{9n}$, the seed.
\par\end{raggedright}

\begin{raggedright}
\textbf{Output:} a Chasm PRG
\par\end{raggedright}
\begin{algor}[1]
\item [{.}] \begin{raggedright}
$\mathbf{s}\gets\langle\sigma_{1},\ldots,\sigma_{n}\rangle\oplus\langle\varphi_{1},\ldots,\varphi_{n}\rangle$
\par\end{raggedright}
\item [{.}] \begin{raggedright}
$\mathbf{o}\gets\langle\sigma_{n+1},\ldots,\sigma_{9n}\rangle\oplus\langle\varphi_{n+1},\ldots,\varphi_{9n}\rangle$
\par\end{raggedright}
\item [{.}] \begin{raggedright}
$\mathbf{i}\gets\mathbf{o}$
\par\end{raggedright}
\item [{.}] \begin{raggedright}
$c\gets0$
\par\end{raggedright}
\item [{.}] \begin{raggedright}
\textbf{return} $(n,\mathbf{s},\mathbf{o},\mathbf{i},c)$ 
\par\end{raggedright}
\end{algor}
\caption{\textsc{\label{alg:ChasmInitialize}ChasmInitialize}}
\end{algorithm}
\end{minipage}\hfill{}%
\begin{minipage}[t]{0.5\columnwidth}%
\begin{procedure}[H]
\begin{raggedright}
\textbf{Parameters:} $g=(n,\mathbf{s},\mathbf{o},\mathbf{i},c)$,
a Chasm PRG
\par\end{raggedright}
\begin{algor}[1]
\item [{.}] \begin{raggedright}
$\mathbf{x}\gets((\mathbf{s}\ggg1)\ll\langle2^{n}\rangle)+(\mathbf{s}\ll\langle1^{n}\rangle)+\mathbf{s}\lll1$
\{\{Collect the neighbors for each cell.\}\}
\par\end{raggedright}
\item [{.}] $\mathbf{f}\gets B[\mathbf{i}]\oplus\mathbf{o}$
\item [{.}] $\mathbf{s}\gets(\mathbf{f}\gg\mathbf{x})\cdot\langle1^{n}\rangle$
\item [{for}] $j=1,\ldots,n$ 

\begin{algor}[1]
\item [{.}] $\mathbf{i}_{j}\gets\mathbf{i}_{j}+1$
\end{algor}
\item [{endfor}]~
\item [{.}] $c\gets c+1$
\end{algor}
\caption{\textsc{\label{pro:ChasmTimeStep}ChasmTimeStep}}
\end{procedure}
\end{minipage}

\end{figure*}

\begin{algorithm}
\textbf{Input:} $g=(n,\mathbf{s},\mathbf{o},\mathbf{i},c)$, a Chasm
PRG

\textbf{Output:} $\mathbf{v}\in\{0,1\}^{n}$
\begin{algor}[1]
\item [{.}] $\mathbf{v}\gets\langle0^{n}\rangle$
\item [{for}] $j=1,\ldots,4$

\begin{algor}[1]
\item [{for}] $k=1,\ldots,n/4$

\begin{algor}[1]
\item [{.}] \nameref{pro:ChasmTimeStep}$(g)$ \{\{update $\mathbf{s}$\}\}
\end{algor}
\item [{endfor}]~
\item [{.}] $\mathbf{v}\gets\mathbf{v}\oplus\mathbf{s}$
\end{algor}
\item [{endfor}]~
\item [{if}] $c\geq256$ \{\{check end of cycle\}\}

\begin{algor}[1]
\item [{.}] $\mathbf{t}\gets\mathbf{s}$
\item [{.}] $\mathbf{u}\gets\langle\rangle$
\item [{for}] $j=1,\ldots,8$

\begin{algor}[1]
\item [{.}] \nameref{pro:ChasmTimeStep}$(g)$ \{\{update $\mathbf{s}$\}\}
\item [{.}] $\mathbf{u}\gets\mathbf{u}||\mathbf{s}$
\end{algor}
\item [{endfor}]~
\item [{.}] $\mathbf{o}\gets\mathbf{o}\oplus\mathbf{u}$
\item [{.}] $\mathbf{s}\gets\mathbf{t}$
\item [{.}] $c\gets0$
\end{algor}
\item [{endif}]~
\item [{.}] \textbf{return} $\mathbf{v}$ 
\end{algor}
\caption{\textsc{\label{alg:ChasmNext}ChasmNext}}
\end{algorithm}

Note that \nameref{alg:ChasmNext} requires modification when $n\nmid256$
or when $n>256$ to ensure the cycle is restarted at the correct time
step. This presentation is simplified for clarity.

\subsection{Security of Chasm}

Recall that two things are required to meet the definition of a forward-secure
PRG:
\begin{enumerate}
\item The next output is hard to predict given previous outputs, and
\item The previous output is hard to compute given the current (stored)
state of the generator.
\end{enumerate}
Proof of either of these properties based on reasonable assumptions
remains open. We conjecture, however, that computing previous outputs
given the current state is closely related to the $k$-SIMPLE-INVERT
problem.

\subsection{Statistical Testing of Chasm}

We have submitted Chasm to the Statistical Test Suite (STS) version
2.1.1 of the National Institute of Standards and Technology \cite{Rukhin2010}.
STS takes a number of sequences generated by the PRG in question and
runs a battery of statistical tests. Each statistic has a distribution
of possible values assuming the null hypothesis that the sequence
is random is true. STS analyzes the observed distribution of statistics
for each test and uses that to draw a conclusion about the null hypothesis.

A significance level $\alpha$ is chosen to help in drawing this conclusion.
STS determines a $p$-value for each statistic run over a single test
sequence. This $p$-value is the probability that a perfect random
generator would produce a seemingly less random (or ``worse'') sequence
than the one tested. The more extreme the $p$-value, the further
out in the ``tails'' the test sequence is in the distribution of
the possible statistic values when those values are computed using
truly random sequences. If the $p$-value is less than our significance
level $\alpha$, STS rejects the null hypothesis that the test sequence
is random; it is simply too unlikely that it came from a random source. 

While a single test sequence may result in a statistic outside the
expected natural range, we must remember that randomness is a probabilistic
property. A certain amount of failing sequences should be expected
from any truly random generator; they should just happen according
to the distribution of the statistic. For this reason, STS runs a
group of sequences through each test and looks for a proportion of
sequences passing a given statistic according to a confidence interval
and also a uniform distribution of $p$-values in the group for that
statistic. These two metrics provide for a high-level conclusion to
be made on the PRG when tested with sufficient seeds, test sequences,
and sequence length. 

The 15 tests employed by the STS are as follows:
\begin{enumerate}
\item The Frequency (Monobit) Test, 
\item Frequency Test within a Block, 
\item The Runs Test, 
\item Tests for the Longest-Run-of-Ones in a Block, 
\item The Binary Matrix Rank Test, 
\item The Discrete Fourier Transform (Spectral) Test, 
\item The Non-overlapping Template Matching Test, 
\item The Overlapping Template Matching Test, 
\item Maurer's \textquotedbl{}Universal Statistical\textquotedbl{} Test, 
\item The Linear Complexity Test, 
\item The Serial Test, 
\item The Approximate Entropy Test, 
\item The Cumulative Sums (Cusums) Test, 
\item The Random Excursions Test, and 
\item The Random Excursions Variant Test. 
\end{enumerate}
Several of these tests are run a number of times with slightly different
parameters, leading to a total of 188 individual tests. Each group
comprises 100 sequences, each 1,000,000 bits in length, giving a total
of 18,800 individual test results. STS reports the proportion of sequences
which pass each individual statistic and compares that to the expected
proportion computed at a chosen significance level. If the observed
proportion is less than expected, we say the group has a \emph{proportion}
failure. STS also reports a $\chi^{2}$ test on the distribution of
the $p$-values to check that each statistic is uniformly distributed.
If the resulting $p$-value of the distribution of $p$-values is
less than 0.0001, we say the group has a \emph{uniformity} failure.
All testing uses a significance level $\alpha=0.01$.

We wish to examine the behavior of Chasm using a number of different
seeds and bit length configurations in order to build general confidence
in the scheme. We chose 18 ``structured'' seeds having various regular
patterns and 20 random seeds for Chasm PRGs having bit lengths of
$n=8$, 16, and 128. Structured seed patterns are listed in \tabref{StructuredSeeds}.
Random seeds were obtained from the Hotbits radioactive decay generator
\cite{hotbits}. We generated 12.5 MB of data from each seed for each
$n$ and used that data for the 100$\times$1,000,000 bit test outlined
above. Sequences were generated from a python implementation using
numpy. For each group of 100 sequences, we recorded the STS uniformity
$p$-value and the proportion of passed tests.

\begin{table}
\begin{centering}
\begin{tabular}{|c|c|c|c|c|c|}
\hline 
Seed & Pattern & Seed & Pattern & Seed & Pattern\tabularnewline
\hline 
\hline 
1 & $\mathtt{0}^{*}$ & 7 & $(\mathtt{1100})^{*}$ & 13 & $(\mathtt{0^{\mathrm{7}}1})^{*}$\tabularnewline
\hline 
2 & $(\mathtt{01})^{*}$ & 8 & $\mathtt{0}^{n/2}\mathtt{1}^{n/2}$ & 14 & $(\mathtt{1^{\mathrm{7}}0})^{*}$\tabularnewline
\hline 
3 & $(\mathtt{010})^{*}$ & 9 & $(\mathtt{1110})^{*}$ & 15 & $(\mathtt{10^{\mathrm{7}}})^{*}$\tabularnewline
\hline 
4 & $(\mathtt{101})^{*}$ & 10 & $(\mathtt{0001})^{*}$ & 16 & $(\mathtt{01^{\mathrm{7}}})^{*}$\tabularnewline
\hline 
5 & $\mathtt{1}^{*}$ & 11 & $(\mathtt{1011})^{*}$ & 17 & $(\mathtt{1^{\mathrm{4}}0^{\mathrm{4}}})^{*}$\tabularnewline
\hline 
6 & $(\mathtt{111000})^{*}$ & 12 & $(\mathtt{0100})^{*}$ & 18 & $(\mathtt{0^{\mathrm{2}}1^{\mathrm{4}}0^{\mathrm{2}}})^{*}$\tabularnewline
\hline 
\end{tabular}
\par\end{centering}

\caption{\label{tab:StructuredSeeds}\textsf{Structured Seed patterns. These
patterns are repeated as necessary to provide a $9n$-bit seed for
each test.}}
\end{table}

\subsection{Test Results}

We will examine the test results from a number of viewpoints. Overall,
there were 2,025,992 individual test statistics computed over 11,400
sequences of 1,000,000 bits each. Of these, 2,005,166 passed their
individual test which is 141 less than the 2,005,307 we would expect
at $\alpha=0.01$. In straight percentages, this is a 98.972\% passing
rate. For comparison, we ran a 1,000-sequence test of 1,000,000 bits
each on the Blum-Blum-Shub (BBS) generator, widely regarded as a strong
PRG. This test passed 98.996\% of individual tests.

\figref{IndividuaFailures} charts how each tested Chasm generator
configuration performed. The left bar in each category shows the number
of observed failures of individual tests, while the right bar shows
the maximum number of failures expected at $\alpha=0.01$. Only the
128-bit configuration using random seeds goes beyond the expected
value, failing 5 tests more than allowed. Notice there is a slightly
different number of tests in each category. This is due to the variable
nature of the Random Excursion tests, which use sequences of different
lengths and so can produce a variable number of sequences to test.
The lower numbers for structured configurations is simply due to testing
only 18 structured seeds vs. 20 random seeds.

\begin{figure}
\begin{centering}
\includegraphics[bb=120bp 270bp 610bp 500bp,clip,scale=0.85]{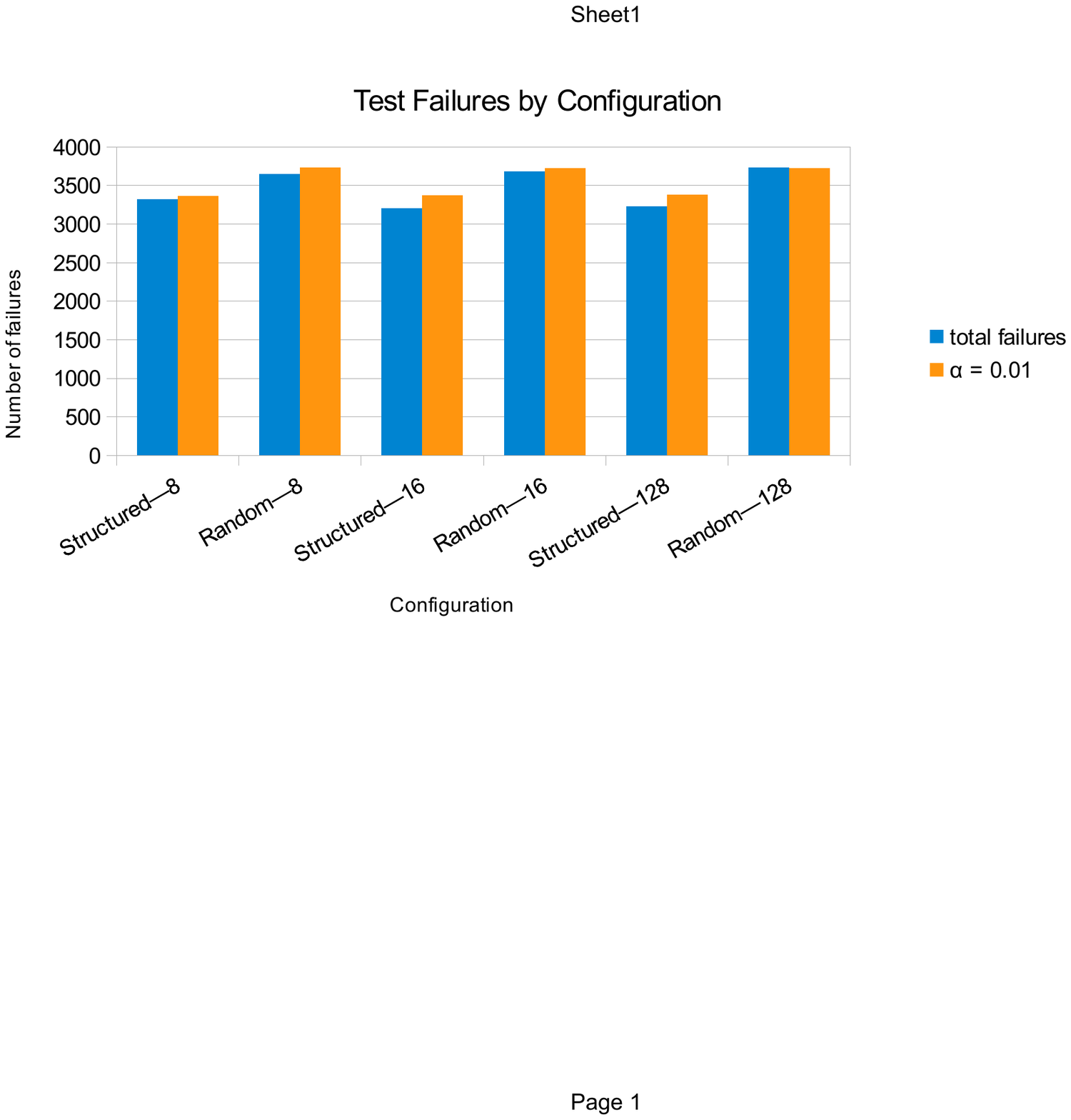}
\par\end{centering}

\caption{\label{fig:IndividuaFailures}\textsf{Individual test failures by
Configuration}}
\end{figure}

These sequences were tested in 21,432 groups of 100 sequences each.
Of these, 20 groups had uniformity failures and 92 had proportion
failures. \tabref{GroupByCat} shows how these failures were distributed
among the test configurations. For a closer look at the effect of
various seed patterns, \figref{GroupBySeed} shows the number of each
failure type for each seed. Seed numbers 19 and above are random,
and are different for each $n$. Seed numbers 18 and below are of
different lengths, but use the same pattern as described in \tabref{StructuredSeeds}
for all $n$. No clear pattern seems to emerge from this data.

\begin{table}
\begin{centering}
\begin{tabular}{ccccc}
$n$ & seed type & uniformity failures & proportion failures & \% of groups\tabularnewline
\hline 
\hline 
\multirow{2}{*}{8} & structured & 11 & 27 & 0.13\%\tabularnewline
 & random & 4 & 18 & 0.08\%\tabularnewline
\hline 
\multirow{2}{*}{16} & structured & 1 & 11 & 0.05\%\tabularnewline
 & random & 1 & 8 & 0.04\%\tabularnewline
\hline 
\multirow{2}{*}{128} & structured & 3 & 12 & 0.06\%\tabularnewline
 & random & 0 & 16 & 0.07\%\tabularnewline
\hline 
\end{tabular}
\par\end{centering}

\caption{\label{tab:GroupByCat}\textsf{Group test results by category}}
\end{table}

\begin{figure}
\begin{centering}
\includegraphics[bb=80bp 175bp 650bp 490bp,clip,scale=0.8]{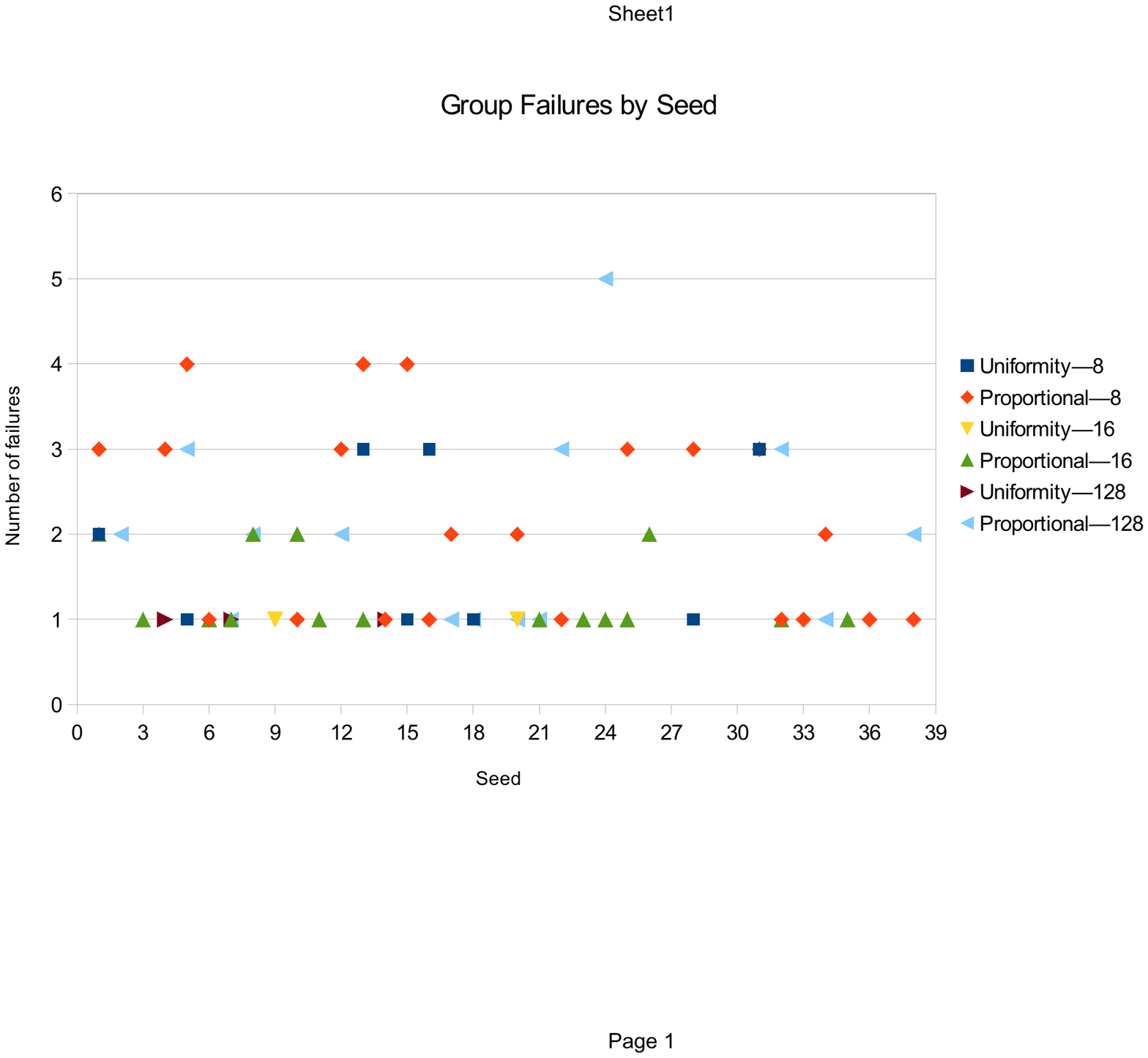}
\par\end{centering}

\caption{\label{fig:GroupBySeed}\textsf{Group failures by seed}}
\end{figure}

We made a few other observations of a 4-bit configuration of Chasm.
We generated 3,750,000 bytes of data from this generator seeded with
0s and subjected it to STS on 30 sequences of 1,000,000 bits each.
The output had an average 0.49\% bias towards '1' bits, and so failed
13 out of 30 groups in the frequency test, 10 and 14 groups in the
two Cumulative Sums tests, and 8 groups in the Runs test. However,
all other tests passed with no proportion or uniformity errors. These
are impressive results for a 4-cell construction, especially considering
the similar results of 18-cell constructions in \cite{GuanT04}. Further,
another run was observed to generate over 50 megabytes of data without
exhibiting strictly periodic behavior. No statistical measurements
were made on the result.

Finally, we note the python implementation was able to generate 475
megabytes in roughly 12 hours while the STS C implementation of the
BBS generator took close to 72 hours to generate 125 megabytes.

\section{Conclusions}

The results of \secref{AnalysisNonLinearCA} and \secref{Analysis-of-Linear-CAs}
make it hard to have confidence in the ability of one-dimensional
two-state three-neighbor cyclic CA to provide secure cryptographic
primitives according to modern definitions. When the rule set is known
and linear and the entire state vector is known, inverting the CA
is of course straight forward. \algref{InvertToggleRule} extends
this condition to uniform CA whose rule is a non-linear toggle rule,
the only kind shown to perform well in statistical testing. When a
known temporal sequence is produced by a known rule set, linear or
not, the techniques presented in \subbref{ImproveMeierStaffel} seem
likely to provide a deterministic algorithm for recovering the seed
of that sequence, improving on the previously known probabilistic
algorithm. Only CA which choose their rules dynamically have no impending
fatal cryptanalytic attack. \secref{Analysis-of-Linear-CAs} shows
these CA are at best very inefficient when the dynamic rules are limited
to linear rules. It may be that dynamically choosing from some set
of non-linear rules provides better efficiency.

Returning to cellular automaton's historical roots by allowing multiple
states per cell, on the other hand, provides a computational model
whose power is determined by whether P=NP. Simple FSCA seem to offer
a candidate one-way function for use in cryptographic primitives.
We may imagine secure PRGs, hash functions, block ciphers, and stream
ciphers based on carefully designed FSCA. These may also allow a security
parameter which lets designers choose a security level appropriate
for the application.

The Chasm family of PRGs presented in \secref{FscaBasedPrg} approaches
one such primitive. Chasm allows for a security parameter whereas
block-based primitives, such as AES- or SHA-1-based PRGs require fixed
sizes. An application can use Chasm at 3 bits and up. If the quality
of data produced by low-$n$ implementations can be shored up, this
may be an attractive option for resource-constrained applications
such as VLSI testing, Bluetooth/wireless applications, RFID readers,
Key fobs, etc. The hardware requirements are a 256-byte lookup table,
$n$-byte offset table (for the bytes of $\varphi$), $n$ 8-to-1
MUXes, and an XOR accumulator. While demonstration that Chasm satisfies
the definition of a forward-secure PRG assuming the one-wayness of
FSCA remains open, it does not seem too far off. If shown, Chasm would
offer a nicely parallel PRG suitable for hardware, GPU, and vector
register implementations. Its performance even in single threaded
interpreted languages seems far superior to the BBS generator, so
a parallel implementation of Chasm with security proofs would be a
very attractive primitive. Setting these proofs aside, our test results
suggest that Chasm is certainly a viable option for a statistical
pseudorandom generator. While there are slightly more failures in
the STS suite than one would like to see, it is quite conceivable
that small adjustments to the Chasm algorithm can correct this. Simply
mixing more values per output may be adequate.

Beyond cryptography, it may be worth considering other applications
of more or less bounded FSCA. For instance, are their problems in
NC0 which might be modeled as FSCA computations and examined from
a different perspective? What would a poly($n$) bounded FSCA be capable
of? Are there applications to problems in PSPACE or EXPTIME? It seems
there are new questions for the adherents of cellular automata to
tackle.

\section{Future Work}

There are quite a few open questions raised above. First, it may be
interesting to consider a decision problem related to \propref{2}:
Given the state vector for a uniform cyclic CA over any rule, is there
a deterministic polynomial-time algorithm to compute its predecessor
state if one exists? We've shown that, when certain patterns exist
in the state vector for certain rules, such an algorithm exists. It
may be interesting to consider the possibilities left when those patterns
\emph{do not} exist.

Second, the algorithm in \subbref{ImproveMeierStaffel} to improve
the bounds on the Meier/Staffelbach algorithm remains to be fully
developed and tested. Such an improvement would have an impact on
much of the literature related to current CA-based cryptosystems.
Related to this are all the open problems discussed in \subbref{CaOpenProblems}.

With respect to FSCA and Chasm, proofs of forward and backward security
would be greatly beneficial to instill confidence in a new primitive.
Also, further study of the linearity and differential properties of
the generator is needed, as well as a more complete assessment of
the potential for weak keys. Finally, applying FSCA to other cryptographic
primitives seems likely to yield interesting results.

\bibliographystyle{plainnat}
\bibliography{bibliography}

\appendix

\section{AES S-Box}

The Values of the AES S-Box are provided for reference. See \cite{FIPS197}
for the full AES specification.

\begin{table}[H]
\begin{tabular}{|c||c|c|c|c|c|c|c|c|c|c|c|c|c|c|c|c|}
\cline{2-17} 
\multicolumn{1}{c||}{} & 0 & 1 & 2 & 3 & 4 & 5 & 6 & 7 & 8 & 9 & a & b & c & d & e & f\tabularnewline
\hline 
\hline 
0 & 63 & 7c & 77 & 7b & f2 & 6b & 6f & c5 & 30 & 01 & 67 & 2b & fe & d7 & ab & 76\tabularnewline
\hline 
10 & ca & 82 & c9 & 7d & fa & 59 & 47 & f0 & ad & d4 & a2 & af & 9c & a4 & 72 & c0\tabularnewline
\hline 
20 & b7 & fd & 93 & 26 & 36 & 3f & f7 & cc & 34 & a5 & e5 & f1 & 71 & d8 & 31 & 15\tabularnewline
\hline 
30 & 04 & c7 & 23 & c3 & 18 & 96 & 5 & 9a & 07 & 12 & 80 & e2 & eb & 27 & b2 & 75\tabularnewline
\hline 
40 & 09 & 83 & 2c & 1a & 1b & 6e & 5a & a0 & 52 & 3b & d6 & b3 & 29 & e3 & 2f & 84\tabularnewline
\hline 
50 & 53 & d1 & 00 & ed & 20 & fc & b1 & 5b & 6a & cb & be & 39 & 4a & 4c & 58 & cf\tabularnewline
\hline 
60 & d0 & ef & aa & fb & 43 & 4d & 33 & 85 & 45 & f9 & 02 & 7f & 50 & 3c & 9f & a8\tabularnewline
\hline 
70 & 51 & a3 & 40 & 8f & 92 & 9d & 38 & f5 & bc & b6 & da & 21 & 10 & ff & f3 & d2\tabularnewline
\hline 
80 & cd & 0c & 13 & ec & 5f & 97 & 44 & 17 & c4 & a7 & 7e & 3d & 64 & 5d & 19 & 73\tabularnewline
\hline 
90 & 60 & 81 & 4f & dc & 22 & 2a & 90 & 88 & 46 & ee & b8 & 14 & de & 5e & 0b & db\tabularnewline
\hline 
a0 & e0 & 32 & 3a & 0a & 49 & 06 & 24 & 5c & c2 & d3 & ac & 62 & 91 & 95 & e4 & 79\tabularnewline
\hline 
b0 & e7 & c8 & 37 & 6d & 8d & d5 & 4e & a9 & 6c & 56 & f4 & ea & 65 & 7a & ae & 08\tabularnewline
\hline 
c0 & ba & 78 & 25 & 2e & 1c & a6 & b4 & c6 & e8 & dd & 74 & 1f & 4b & bd & 8b & 8a\tabularnewline
\hline 
d0 & 70 & 3e & b5 & 66 & 48 & 03 & f6 & 0e & 61 & 35 & 57 & b9 & 86 & c1 & 1d & 9e\tabularnewline
\hline 
e0 & e1 & f8 & 98 & 11 & 69 & d9 & 8e & 94 & 9b & 1e & 87 & e9 & ce & 55 & 28 & df\tabularnewline
\hline 
f0 & 8c & a1 & 89 & 0d & bf & e6 & 42 & 68 & 41 & 99 & 2d & 0f & b0 & 54 & bb & 16\tabularnewline
\hline 
\end{tabular}

\caption{\label{tab:AesSBox}\textsf{Hexadecimal values of the AES S-Box as
a lookup table.}}

\end{table}

\section{Experiments}

The sections in this appendix demonstrate representative experiments
showing some property or other. These experiments were conducted with
python 2.7.1 using Numpy 1.6.1. In some cases, actual python interpreter
session are captured.

\subsection{\label{apn:SymEffects}Effects of Rule Symmetry on Multiple Seeds
for a Given Sequence}

The following experiments show that symmetry in the rule set seems
to affect the number of seeds that can generate the same temporal
sequence. The function OtherMatchingSeeds finds all seeds that produce
the same temporal sequence over $n$ steps as a given seed under a
given rule set.
\begin{codesmall}
>\textcompwordmark{}>\textcompwordmark{}>~OtherMatchingSeeds({[}150,150,150,150,90,150,150,150,150{]},~SeedFromStr('010110110'))

{[}0~0~0~0~1~1~1~0~0{]}

{[}0~0~0~1~1~0~1~0~0{]}

{[}0~0~1~0~1~1~0~0~0{]}

{[}0~0~1~1~1~0~0~0~0{]}

{[}0~1~0~0~1~1~1~1~0{]}

{[}0~1~0~1~1~0~1~1~0{]}

{[}0~1~1~0~1~1~0~1~0{]}

{[}0~1~1~1~1~0~0~1~0{]}

{[}1~0~0~0~1~1~1~0~1{]}

{[}1~0~0~1~1~0~1~0~1{]}

{[}1~0~1~0~1~1~0~0~1{]}

{[}1~0~1~1~1~0~0~0~1{]}

{[}1~1~0~0~1~1~1~1~1{]}

{[}1~1~0~1~1~0~1~1~1{]}

{[}1~1~1~0~1~1~0~1~1{]}

{[}1~1~1~1~1~0~0~1~1{]}

>\textcompwordmark{}>\textcompwordmark{}>~OtherMatchingSeeds({[}150,150,150,150,150,150,150,90,150{]},~SeedFromStr('110010110'))

{[}0~1~1~1~1~1~0~1~1{]}

{[}1~1~0~0~1~0~1~1~0{]}

>\textcompwordmark{}>\textcompwordmark{}>~OtherMatchingSeeds({[}150,150,150,150,150,150,150,150,90{]},~SeedFromStr('110010110'))

{[}1~1~0~0~1~0~1~1~0{]}
\end{codesmall}
Severely asymmetrical rules limit the matches severly--one rule 90
at the edge means there's a single seed that generates the temporal
sequence.
\begin{codesmall}
>\textcompwordmark{}>\textcompwordmark{}>~OtherMatchingSeeds({[}150,150,150,150,150,150,150,150,90{]},~SeedFromStr('011001011'))

{[}0~1~1~0~0~1~0~1~1{]}

>\textcompwordmark{}>\textcompwordmark{}>~OtherMatchingSeeds({[}150,150,150,150,150,150,150,150,90{]},~SeedFromStr('101100101'))

{[}1~0~1~1~0~0~1~0~1{]}

>\textcompwordmark{}>\textcompwordmark{}>~OtherMatchingSeeds({[}150,150,150,150,150,150,150,150,165{]},~SeedFromStr('101100101'))

{[}1~0~1~1~0~0~1~0~1{]}
\end{codesmall}
As long as that mismatched rule is a two-input rule, there's only
one seed. If it's changed to a 3-input rule, even complementary, many
matching seeds exist.
\begin{codesmall}
>\textcompwordmark{}>\textcompwordmark{}>~OtherMatchingSeeds({[}150,150,150,150,150,150,150,150,105{]},~SeedFromStr('101100101'))

{[}0~0~0~0~0~1~0~0~0{]}

{[}0~0~0~1~0~0~0~0~0{]}

{[}0~0~1~0~0~1~1~0~0{]}

{[}0~0~1~1~0~0~1~0~0{]}

{[}0~1~0~0~0~1~0~1~0{]}

{[}0~1~0~1~0~0~0~1~0{]}

{[}0~1~1~0~0~1~1~1~0{]}

{[}0~1~1~1~0~0~1~1~0{]}

{[}1~0~0~0~0~1~0~0~1{]}

{[}1~0~0~1~0~0~0~0~1{]}

{[}1~0~1~0~0~1~1~0~1{]}

{[}1~0~1~1~0~0~1~0~1{]}

{[}1~1~0~0~0~1~0~1~1{]}

{[}1~1~0~1~0~0~0~1~1{]}

{[}1~1~1~0~0~1~1~1~1{]}

{[}1~1~1~1~0~0~1~1~1{]}

>\textcompwordmark{}>\textcompwordmark{}>~OtherMatchingSeeds({[}150,150,150,150,150,150,\uline{105},150,150{]},~SeedFromStr('101100101'))

{[}0~0~0~0~0~1~0~0~0{]}

{[}0~0~0~1~0~0~0~0~0{]}

{[}0~0~1~0~0~1~1~0~0{]}

{[}0~0~1~1~0~0~1~0~0{]}

{[}0~1~0~0~0~1~0~1~0{]}

{[}0~1~0~1~0~0~0~1~0{]}

{[}0~1~1~0~0~1~1~1~0{]}

{[}0~1~1~1~0~0~1~1~0{]}

{[}1~0~0~0~0~1~0~0~1{]}

{[}1~0~0~1~0~0~0~0~1{]}

{[}1~0~1~0~0~1~1~0~1{]}

{[}1~0~1~1~0~0~1~0~1{]}

{[}1~1~0~0~0~1~0~1~1{]}

{[}1~1~0~1~0~0~0~1~1{]}

{[}1~1~1~0~0~1~1~1~1{]}

{[}1~1~1~1~0~0~1~1~1{]}

>\textcompwordmark{}>\textcompwordmark{}>~OtherMatchingSeeds({[}150,150,150,150,\uline{105},150,\uline{105},150,150{]},~SeedFromStr('101100101'))

{[}0~0~0~0~0~1~0~0~0{]}

{[}0~0~0~1~0~0~0~0~0{]}

{[}0~0~1~0~0~1~1~0~0{]}

{[}0~0~1~1~0~0~1~0~0{]}

{[}0~1~0~0~0~1~0~1~0{]}

{[}0~1~0~1~0~0~0~1~0{]}

{[}0~1~1~0~0~1~1~1~0{]}

{[}0~1~1~1~0~0~1~1~0{]}

{[}1~0~0~0~0~1~0~0~1{]}

{[}1~0~0~1~0~0~0~0~1{]}

{[}1~0~1~0~0~1~1~0~1{]}

{[}1~0~1~1~0~0~1~0~1{]}

{[}1~1~0~0~0~1~0~1~1{]}

{[}1~1~0~1~0~0~0~1~1{]}

{[}1~1~1~0~0~1~1~1~1{]}

{[}1~1~1~1~0~0~1~1~1{]}

>\textcompwordmark{}>\textcompwordmark{}>~OtherMatchingSeeds({[}\uline{105},150,150,150,\uline{105},150,\uline{105},150,150{]},~SeedFromStr('101100101'))

{[}0~0~0~0~0~1~0~0~0{]}

{[}0~0~0~1~0~0~0~0~0{]}

{[}0~0~1~0~0~1~1~0~0{]}

{[}0~0~1~1~0~0~1~0~0{]}

{[}0~1~0~0~0~1~0~1~0{]}

{[}0~1~0~1~0~0~0~1~0{]}

{[}0~1~1~0~0~1~1~1~0{]}

{[}0~1~1~1~0~0~1~1~0{]}

{[}1~0~0~0~0~1~0~0~1{]}

{[}1~0~0~1~0~0~0~0~1{]}

{[}1~0~1~0~0~1~1~0~1{]}

{[}1~0~1~1~0~0~1~0~1{]}

{[}1~1~0~0~0~1~0~1~1{]}

{[}1~1~0~1~0~0~0~1~1{]}

{[}1~1~1~0~0~1~1~1~1{]}

{[}1~1~1~1~0~0~1~1~1{]}
\end{codesmall}
So placing complementary 3-input rules in various asymetric locations
appears not to change the results. Here, different seeds are shown
to produce the same temporal sequence.
\begin{codesmall}
>\textcompwordmark{}>\textcompwordmark{}>~TempSeqFromSeed({[}\uline{105},150,150,150,\uline{105},150,\uline{105},150,150{]},~SeedFromStr('100001001'),~18)

array({[}0,~0,~0,~0,~0,~1,~1,~0,~1,~1,~1,~1,~0,~0,~1,~0,~0,~0,~0{]},~dtype=uint8)

>\textcompwordmark{}>\textcompwordmark{}>~TempSeqFromSeed({[}\uline{105},150,150,150,\uline{105},150,\uline{105},150,150{]},~SeedFromStr('101100101'),~18)

array({[}0,~0,~0,~0,~0,~1,~1,~0,~1,~1,~1,~1,~0,~0,~1,~0,~0,~0,~0{]},~dtype=uint8)
\end{codesmall}
The location of the asymmetry appears to affect the number of matching
seeds.  The following shows how the number of matches drops as we
go from perfectly symmetrical through different asymmetry patterns.
Notice that in the extreme cases, the only seed that matches is the
original seed, stored as s1.
\begin{codesmall}
>\textcompwordmark{}>\textcompwordmark{}>~s1~=~SeedFromStr('101100101')

>\textcompwordmark{}>\textcompwordmark{}>~rules~=~{[}150,~90,~150,~150,~150,~150,~150,~90,~150{]}

>\textcompwordmark{}>\textcompwordmark{}>~OtherMatchingSeeds(rules,~s1)

{[}0~0~0~0~0~1~0~0~0{]}

{[}0~0~0~1~0~0~0~0~0{]}

{[}0~0~1~0~0~1~1~0~0{]}

{[}0~0~1~1~0~0~1~0~0{]}

{[}0~1~0~0~0~1~0~1~0{]}

{[}0~1~0~1~0~0~0~1~0{]}

{[}0~1~1~0~0~1~1~1~0{]}

{[}0~1~1~1~0~0~1~1~0{]}

{[}1~0~0~0~0~1~0~0~1{]}

{[}1~0~0~1~0~0~0~0~1{]}

{[}1~0~1~0~0~1~1~0~1{]}

{[}1~0~1~1~0~0~1~0~1{]}

{[}1~1~0~0~0~1~0~1~1{]}

{[}1~1~0~1~0~0~0~1~1{]}

{[}1~1~1~0~0~1~1~1~1{]}

{[}1~1~1~1~0~0~1~1~1{]}

>\textcompwordmark{}>\textcompwordmark{}>~rules~=~{[}\uline{90,}~\uline{150},~150,~150,~150,~150,~150,~90,~150{]}

>\textcompwordmark{}>\textcompwordmark{}>~OtherMatchingSeeds(rules,~s1)

{[}1~0~1~1~0~0~1~0~1{]}

>\textcompwordmark{}>\textcompwordmark{}>~rules~=~{[}90,~~\uline{90},~150,~150,~150,~150,~150,~90,~150{]}

>\textcompwordmark{}>\textcompwordmark{}>~OtherMatchingSeeds(rules,~s1)

{[}1~0~1~1~0~0~1~0~1{]}

>\textcompwordmark{}>\textcompwordmark{}>~rules~=~{[}\uline{150,}~90,~~\uline{90},~150,~150,~150,~150,~90,~150{]}

>\textcompwordmark{}>\textcompwordmark{}>~OtherMatchingSeeds(rules,~s1)

{[}0~1~0~1~0~0~0~0~1{]}

{[}0~1~1~0~0~1~1~1~0{]}

{[}1~0~0~0~0~1~0~1~0{]}

{[}1~0~1~1~0~0~1~0~1{]}

>\textcompwordmark{}>\textcompwordmark{}>~rules~=~{[}150,~90,~~90,~150,~150,~150,~~\uline{90},~90,~150{]}

>\textcompwordmark{}>\textcompwordmark{}>~OtherMatchingSeeds(rules,~s1)

{[}0~0~0~0~0~1~0~0~0{]}

{[}0~0~0~1~0~0~0~0~0{]}

{[}0~0~1~0~0~1~1~0~0{]}

{[}0~0~1~1~0~0~1~0~0{]}

{[}0~1~0~0~0~1~0~1~0{]}

{[}0~1~0~1~0~0~0~1~0{]}

{[}0~1~1~0~0~1~1~1~0{]}

{[}0~1~1~1~0~0~1~1~0{]}

{[}1~0~0~0~0~1~0~0~1{]}

{[}1~0~0~1~0~0~0~0~1{]}

{[}1~0~1~0~0~1~1~0~1{]}

{[}1~0~1~1~0~0~1~0~1{]}

{[}1~1~0~0~0~1~0~1~1{]}

{[}1~1~0~1~0~0~0~1~1{]}

{[}1~1~1~0~0~1~1~1~1{]}

{[}1~1~1~1~0~0~1~1~1{]}

>\textcompwordmark{}>\textcompwordmark{}>~rules~=~{[}150,~90,~~90,~150,~150,~150,~90,~\uline{150},~150{]}

>\textcompwordmark{}>\textcompwordmark{}>~OtherMatchingSeeds(rules,~s1)

{[}0~1~0~0~0~1~0~1~1{]}

{[}1~0~1~1~0~0~1~0~1{]}

>\textcompwordmark{}>\textcompwordmark{}>~rules~=~{[}150,150,150,150,150,90,150,150,150{]}

>\textcompwordmark{}>\textcompwordmark{}>~OtherMatchingSeeds(rules,~s1)

{[}1~0~1~1~0~0~1~0~1{]}~\\
~\\

>\textcompwordmark{}>\textcompwordmark{}>~rules~=~{[}150,90,90,150,150,150,90,150,150,90,150,150,150,150,90,150,150{]}

>\textcompwordmark{}>\textcompwordmark{}>~s3~=~SeedFromStr('00101001110111101')

>\textcompwordmark{}>\textcompwordmark{}>~OtherMatchingSeeds(rules,~s3)

{[}0~0~1~0~1~0~0~1~1~1~0~1~1~1~1~0~1{]}
\end{codesmall}

\subsection{\label{apn:SolveTempSeq}Solving Temporal Sequences Under Arbitrary
Symmetrical Rulesets}

This experiment examines how local rule differences affect a three-cell
neighbor hood by looking at difference patterns in the output under
different rule sets from a uniform rule 150 CA. The function prints
$^{H}S_{i-1:i+1}^{t}$, $^{H}S_{i-2:i+2}^{t-1}$ for each predecessor,
and $^{150}S_{i-2:i+2}^{t-1}\oplus{}^{H}S_{i-2:i+2}^{t-1}$ for all
predecessors under a hybrid rule set $H$. The notation $^{R}S_{i:j}$
denotes the state vector values from cell $i$ to $j$ under rule
set $R$. 

This information shows how to get the same successor values under
a different rule set by changing the predecessor state values. Notice
that those rules that differ only in their complementarity have constant
difference patterns for all successor values. There is no rule set
with unique differences in all 8 successor values--4 seems to be the
maximum.
\begin{codesmall}

Key:

s\_t~|-{}-~S\_i-2:i+2~at~t-1~-|~|-{}-xor~with~preds(150)-{}-|

>\textcompwordmark{}>\textcompwordmark{}>~rs.hybridPredDiffs()

Hybrid~Diff~Rules:~90~90~90~:

000~00000~01010~10101~11111~00000~00111~00011~00100

001~00001~01011~10100~11110~00000~00111~00011~00100

010~00010~01000~10111~11101~00001~00110~00010~00101

011~00011~01001~10110~11100~00001~00110~00010~00101

100~00101~01111~10000~11010~00011~00100~00000~00111

101~00100~01110~10001~11011~00011~00100~00000~00111

110~00111~01101~10010~11000~00010~00101~00001~00110

111~00110~01100~10011~11001~00010~00101~00001~00110

Hybrid~Diff~Rules:~90~90~150~:

000~00000~01011~10101~11110~00000~00110~00011~00101

001~00001~01010~10100~11111~00000~00110~00011~00101

010~00011~01000~10110~11101~00000~00110~00011~00101

011~00010~01001~10111~11100~00000~00110~00011~00101

100~00101~01110~10000~11011~00011~00101~00000~00110

101~00100~01111~10001~11010~00011~00101~00000~00110

110~00110~01101~10011~11000~00011~00101~00000~00110

111~00111~01100~10010~11001~00011~00101~00000~00110

Hybrid~Diff~Rules:~90~90~105~:

000~00001~01010~10100~11111~00001~00111~00010~00100

001~00000~01011~10101~11110~00001~00111~00010~00100

010~00010~01001~10111~11100~00001~00111~00010~00100

011~00011~01000~10110~11101~00001~00111~00010~00100

100~00100~01111~10001~11010~00010~00100~00001~00111

101~00101~01110~10000~11011~00010~00100~00001~00111

110~00111~01100~10010~11001~00010~00100~00001~00111

111~00110~01101~10011~11000~00010~00100~00001~00111

Hybrid~Diff~Rules:~90~90~165~:

000~00001~01011~10100~11110~00001~00110~00010~00101

001~00000~01010~10101~11111~00001~00110~00010~00101

010~00011~01001~10110~11100~00000~00111~00011~00100

011~00010~01000~10111~11101~00000~00111~00011~00100

100~00100~01110~10001~11011~00010~00101~00001~00110

101~00101~01111~10000~11010~00010~00101~00001~00110

110~00110~01100~10011~11001~00011~00100~00000~00111

111~00111~01101~10010~11000~00011~00100~00000~00111

Hybrid~Diff~Rules:~90~150~90~:

000~00000~01010~10111~11101~00000~00111~00001~00110

001~00001~01011~10110~11100~00000~00111~00001~00110

010~00010~01000~10101~11111~00001~00110~00000~00111

011~00011~01001~10100~11110~00001~00110~00000~00111

100~00111~01101~10000~11010~00001~00110~00000~00111

101~00110~01100~10001~11011~00001~00110~00000~00111

110~00101~01111~10010~11000~00000~00111~00001~00110

111~00100~01110~10011~11001~00000~00111~00001~00110

Hybrid~Diff~Rules:~90~150~150~:

000~00000~01011~10110~11101~00000~00110~00000~00110

001~00001~01010~10111~11100~00000~00110~00000~00110

010~00011~01000~10101~11110~00000~00110~00000~00110

011~00010~01001~10100~11111~00000~00110~00000~00110

100~00110~01101~10000~11011~00000~00110~00000~00110

101~00111~01100~10001~11010~00000~00110~00000~00110

110~00101~01110~10011~11000~00000~00110~00000~00110

111~00100~01111~10010~11001~00000~00110~00000~00110

Hybrid~Diff~Rules:~90~150~105~:

000~00001~01010~10111~11100~00001~00111~00001~00111

001~00000~01011~10110~11101~00001~00111~00001~00111

010~00010~01001~10100~11111~00001~00111~00001~00111

011~00011~01000~10101~11110~00001~00111~00001~00111

100~00111~01100~10001~11010~00001~00111~00001~00111

101~00110~01101~10000~11011~00001~00111~00001~00111

110~00100~01111~10010~11001~00001~00111~00001~00111

111~00101~01110~10011~11000~00001~00111~00001~00111

Hybrid~Diff~Rules:~90~150~165~:

000~00001~01011~10110~11100~00001~00110~00000~00111

001~00000~01010~10111~11101~00001~00110~00000~00111

010~00011~01001~10100~11110~00000~00111~00001~00110

011~00010~01000~10101~11111~00000~00111~00001~00110

100~00110~01100~10001~11011~00000~00111~00001~00110

101~00111~01101~10000~11010~00000~00111~00001~00110

110~00100~01110~10011~11001~00001~00110~00000~00111

111~00101~01111~10010~11000~00001~00110~00000~00111

Hybrid~Diff~Rules:~90~105~90~:

000~00010~01000~10101~11111~00010~00101~00011~00100

001~00011~01001~10100~11110~00010~00101~00011~00100

010~00000~01010~10111~11101~00011~00100~00010~00101

011~00001~01011~10110~11100~00011~00100~00010~00101

100~00101~01111~10010~11000~00011~00100~00010~00101

101~00100~01110~10011~11001~00011~00100~00010~00101

110~00111~01101~10000~11010~00010~00101~00011~00100

111~00110~01100~10001~11011~00010~00101~00011~00100

Hybrid~Diff~Rules:~90~105~150~:

000~00011~01000~10101~11110~00011~00101~00011~00101

001~00010~01001~10100~11111~00011~00101~00011~00101

010~00000~01011~10110~11101~00011~00101~00011~00101

011~00001~01010~10111~11100~00011~00101~00011~00101

100~00101~01110~10011~11000~00011~00101~00011~00101

101~00100~01111~10010~11001~00011~00101~00011~00101

110~00110~01101~10000~11011~00011~00101~00011~00101

111~00111~01100~10001~11010~00011~00101~00011~00101

Hybrid~Diff~Rules:~90~105~105~:

000~00010~01001~10100~11111~00010~00100~00010~00100

001~00011~01000~10101~11110~00010~00100~00010~00100

010~00001~01010~10111~11100~00010~00100~00010~00100

011~00000~01011~10110~11101~00010~00100~00010~00100

100~00100~01111~10010~11001~00010~00100~00010~00100

101~00101~01110~10011~11000~00010~00100~00010~00100

110~00111~01100~10001~11010~00010~00100~00010~00100

111~00110~01101~10000~11011~00010~00100~00010~00100

Hybrid~Diff~Rules:~90~105~165~:

000~00011~01001~10100~11110~00011~00100~00010~00101

001~00010~01000~10101~11111~00011~00100~00010~00101

010~00001~01011~10110~11100~00010~00101~00011~00100

011~00000~01010~10111~11101~00010~00101~00011~00100

100~00100~01110~10011~11001~00010~00101~00011~00100

101~00101~01111~10010~11000~00010~00101~00011~00100

110~00110~01100~10001~11011~00011~00100~00010~00101

111~00111~01101~10000~11010~00011~00100~00010~00101

Hybrid~Diff~Rules:~90~165~90~:

000~00010~01000~10111~11101~00010~00101~00001~00110

001~00011~01001~10110~11100~00010~00101~00001~00110

010~00000~01010~10101~11111~00011~00100~00000~00111

011~00001~01011~10100~11110~00011~00100~00000~00111

100~00111~01101~10010~11000~00001~00110~00010~00101

101~00110~01100~10011~11001~00001~00110~00010~00101

110~00101~01111~10000~11010~00000~00111~00011~00100

111~00100~01110~10001~11011~00000~00111~00011~00100

Hybrid~Diff~Rules:~90~165~150~:

000~00011~01000~10110~11101~00011~00101~00000~00110

001~00010~01001~10111~11100~00011~00101~00000~00110

010~00000~01011~10101~11110~00011~00101~00000~00110

011~00001~01010~10100~11111~00011~00101~00000~00110

100~00110~01101~10011~11000~00000~00110~00011~00101

101~00111~01100~10010~11001~00000~00110~00011~00101

110~00101~01110~10000~11011~00000~00110~00011~00101

111~00100~01111~10001~11010~00000~00110~00011~00101

Hybrid~Diff~Rules:~90~165~105~:

000~00010~01001~10111~11100~00010~00100~00001~00111

001~00011~01000~10110~11101~00010~00100~00001~00111

010~00001~01010~10100~11111~00010~00100~00001~00111

011~00000~01011~10101~11110~00010~00100~00001~00111

100~00111~01100~10010~11001~00001~00111~00010~00100

101~00110~01101~10011~11000~00001~00111~00010~00100

110~00100~01111~10001~11010~00001~00111~00010~00100

111~00101~01110~10000~11011~00001~00111~00010~00100

Hybrid~Diff~Rules:~90~165~165~:

000~00011~01001~10110~11100~00011~00100~00000~00111

001~00010~01000~10111~11101~00011~00100~00000~00111

010~00001~01011~10100~11110~00010~00101~00001~00110

011~00000~01010~10101~11111~00010~00101~00001~00110

100~00110~01100~10011~11001~00000~00111~00011~00100

101~00111~01101~10010~11000~00000~00111~00011~00100

110~00100~01110~10001~11011~00001~00110~00010~00101

111~00101~01111~10000~11010~00001~00110~00010~00101

Hybrid~Diff~Rules:~150~90~90~:

000~00000~01111~10101~11010~00000~00010~00011~00001

001~00001~01110~10100~11011~00000~00010~00011~00001

010~00010~01101~10111~11000~00001~00011~00010~00000

011~00011~01100~10110~11001~00001~00011~00010~00000

100~00101~01010~10000~11111~00011~00001~00000~00010

101~00100~01011~10001~11110~00011~00001~00000~00010

110~00111~01000~10010~11101~00010~00000~00001~00011

111~00110~01001~10011~11100~00010~00000~00001~00011

Hybrid~Diff~Rules:~150~90~150~:

000~00000~01110~10101~11011~00000~00011~00011~00000

001~00001~01111~10100~11010~00000~00011~00011~00000

010~00011~01101~10110~11000~00000~00011~00011~00000

011~00010~01100~10111~11001~00000~00011~00011~00000

100~00101~01011~10000~11110~00011~00000~00000~00011

101~00100~01010~10001~11111~00011~00000~00000~00011

110~00110~01000~10011~11101~00011~00000~00000~00011

111~00111~01001~10010~11100~00011~00000~00000~00011

Hybrid~Diff~Rules:~150~90~105~:

000~00001~01111~10100~11010~00001~00010~00010~00001

001~00000~01110~10101~11011~00001~00010~00010~00001

010~00010~01100~10111~11001~00001~00010~00010~00001

011~00011~01101~10110~11000~00001~00010~00010~00001

100~00100~01010~10001~11111~00010~00001~00001~00010

101~00101~01011~10000~11110~00010~00001~00001~00010

110~00111~01001~10010~11100~00010~00001~00001~00010

111~00110~01000~10011~11101~00010~00001~00001~00010

Hybrid~Diff~Rules:~150~90~165~:

000~00001~01110~10100~11011~00001~00011~00010~00000

001~00000~01111~10101~11010~00001~00011~00010~00000

010~00011~01100~10110~11001~00000~00010~00011~00001

011~00010~01101~10111~11000~00000~00010~00011~00001

100~00100~01011~10001~11110~00010~00000~00001~00011

101~00101~01010~10000~11111~00010~00000~00001~00011

110~00110~01001~10011~11100~00011~00001~00000~00010

111~00111~01000~10010~11101~00011~00001~00000~00010

Hybrid~Diff~Rules:~150~150~90~:

000~00000~01101~10111~11010~00000~00000~00001~00001

001~00001~01100~10110~11011~00000~00000~00001~00001

010~00010~01111~10101~11000~00001~00001~00000~00000

011~00011~01110~10100~11001~00001~00001~00000~00000

100~00111~01010~10000~11101~00001~00001~00000~00000

101~00110~01011~10001~11100~00001~00001~00000~00000

110~00101~01000~10010~11111~00000~00000~00001~00001

111~00100~01001~10011~11110~00000~00000~00001~00001

Hybrid~Diff~Rules:~150~150~150~:

000~00000~01101~10110~11011~00000~00000~00000~00000

001~00001~01100~10111~11010~00000~00000~00000~00000

010~00011~01110~10101~11000~00000~00000~00000~00000

011~00010~01111~10100~11001~00000~00000~00000~00000

100~00110~01011~10000~11101~00000~00000~00000~00000

101~00111~01010~10001~11100~00000~00000~00000~00000

110~00101~01000~10011~11110~00000~00000~00000~00000

111~00100~01001~10010~11111~00000~00000~00000~00000

Hybrid~Diff~Rules:~150~150~105~:

000~00001~01100~10111~11010~00001~00001~00001~00001

001~00000~01101~10110~11011~00001~00001~00001~00001

010~00010~01111~10100~11001~00001~00001~00001~00001

011~00011~01110~10101~11000~00001~00001~00001~00001

100~00111~01010~10001~11100~00001~00001~00001~00001

101~00110~01011~10000~11101~00001~00001~00001~00001

110~00100~01001~10010~11111~00001~00001~00001~00001

111~00101~01000~10011~11110~00001~00001~00001~00001

Hybrid~Diff~Rules:~150~150~165~:

000~00001~01100~10110~11011~00001~00001~00000~00000

001~00000~01101~10111~11010~00001~00001~00000~00000

010~00011~01110~10100~11001~00000~00000~00001~00001

011~00010~01111~10101~11000~00000~00000~00001~00001

100~00110~01011~10001~11100~00000~00000~00001~00001

101~00111~01010~10000~11101~00000~00000~00001~00001

110~00100~01001~10011~11110~00001~00001~00000~00000

111~00101~01000~10010~11111~00001~00001~00000~00000

Hybrid~Diff~Rules:~150~105~90~:

000~00010~01111~10101~11000~00010~00010~00011~00011

001~00011~01110~10100~11001~00010~00010~00011~00011

010~00000~01101~10111~11010~00011~00011~00010~00010

011~00001~01100~10110~11011~00011~00011~00010~00010

100~00101~01000~10010~11111~00011~00011~00010~00010

101~00100~01001~10011~11110~00011~00011~00010~00010

110~00111~01010~10000~11101~00010~00010~00011~00011

111~00110~01011~10001~11100~00010~00010~00011~00011

Hybrid~Diff~Rules:~150~105~150~:

000~00011~01110~10101~11000~00011~00011~00011~00011

001~00010~01111~10100~11001~00011~00011~00011~00011

010~00000~01101~10110~11011~00011~00011~00011~00011

011~00001~01100~10111~11010~00011~00011~00011~00011

100~00101~01000~10011~11110~00011~00011~00011~00011

101~00100~01001~10010~11111~00011~00011~00011~00011

110~00110~01011~10000~11101~00011~00011~00011~00011

111~00111~01010~10001~11100~00011~00011~00011~00011

Hybrid~Diff~Rules:~150~105~105~:

000~00010~01111~10100~11001~00010~00010~00010~00010

001~00011~01110~10101~11000~00010~00010~00010~00010

010~00001~01100~10111~11010~00010~00010~00010~00010

011~00000~01101~10110~11011~00010~00010~00010~00010

100~00100~01001~10010~11111~00010~00010~00010~00010

101~00101~01000~10011~11110~00010~00010~00010~00010

110~00111~01010~10001~11100~00010~00010~00010~00010

111~00110~01011~10000~11101~00010~00010~00010~00010

Hybrid~Diff~Rules:~150~105~165~:

000~00011~01110~10100~11001~00011~00011~00010~00010

001~00010~01111~10101~11000~00011~00011~00010~00010

010~00001~01100~10110~11011~00010~00010~00011~00011

011~00000~01101~10111~11010~00010~00010~00011~00011

100~00100~01001~10011~11110~00010~00010~00011~00011

101~00101~01000~10010~11111~00010~00010~00011~00011

110~00110~01011~10001~11100~00011~00011~00010~00010

111~00111~01010~10000~11101~00011~00011~00010~00010

Hybrid~Diff~Rules:~150~165~90~:

000~00010~01101~10111~11000~00010~00000~00001~00011

001~00011~01100~10110~11001~00010~00000~00001~00011

010~00000~01111~10101~11010~00011~00001~00000~00010

011~00001~01110~10100~11011~00011~00001~00000~00010

100~00111~01000~10010~11101~00001~00011~00010~00000

101~00110~01001~10011~11100~00001~00011~00010~00000

110~00101~01010~10000~11111~00000~00010~00011~00001

111~00100~01011~10001~11110~00000~00010~00011~00001

Hybrid~Diff~Rules:~150~165~150~:

000~00011~01101~10110~11000~00011~00000~00000~00011

001~00010~01100~10111~11001~00011~00000~00000~00011

010~00000~01110~10101~11011~00011~00000~00000~00011

011~00001~01111~10100~11010~00011~00000~00000~00011

100~00110~01000~10011~11101~00000~00011~00011~00000

101~00111~01001~10010~11100~00000~00011~00011~00000

110~00101~01011~10000~11110~00000~00011~00011~00000

111~00100~01010~10001~11111~00000~00011~00011~00000

Hybrid~Diff~Rules:~150~165~105~:

000~00010~01100~10111~11001~00010~00001~00001~00010

001~00011~01101~10110~11000~00010~00001~00001~00010

010~00001~01111~10100~11010~00010~00001~00001~00010

011~00000~01110~10101~11011~00010~00001~00001~00010

100~00111~01001~10010~11100~00001~00010~00010~00001

101~00110~01000~10011~11101~00001~00010~00010~00001

110~00100~01010~10001~11111~00001~00010~00010~00001

111~00101~01011~10000~11110~00001~00010~00010~00001

Hybrid~Diff~Rules:~150~165~165~:

000~00011~01100~10110~11001~00011~00001~00000~00010

001~00010~01101~10111~11000~00011~00001~00000~00010

010~00001~01110~10100~11011~00010~00000~00001~00011

011~00000~01111~10101~11010~00010~00000~00001~00011

100~00110~01001~10011~11100~00000~00010~00011~00001

101~00111~01000~10010~11101~00000~00010~00011~00001

110~00100~01011~10001~11110~00001~00011~00010~00000

111~00101~01010~10000~11111~00001~00011~00010~00000

Hybrid~Diff~Rules:~105~90~90~:

000~00101~01010~10000~11111~00101~00111~00110~00100

001~00100~01011~10001~11110~00101~00111~00110~00100

010~00111~01000~10010~11101~00100~00110~00111~00101

011~00110~01001~10011~11100~00100~00110~00111~00101

100~00000~01111~10101~11010~00110~00100~00101~00111

101~00001~01110~10100~11011~00110~00100~00101~00111

110~00010~01101~10111~11000~00111~00101~00100~00110

111~00011~01100~10110~11001~00111~00101~00100~00110

Hybrid~Diff~Rules:~105~90~150~:

000~00101~01011~10000~11110~00101~00110~00110~00101

001~00100~01010~10001~11111~00101~00110~00110~00101

010~00110~01000~10011~11101~00101~00110~00110~00101

011~00111~01001~10010~11100~00101~00110~00110~00101

100~00000~01110~10101~11011~00110~00101~00101~00110

101~00001~01111~10100~11010~00110~00101~00101~00110

110~00011~01101~10110~11000~00110~00101~00101~00110

111~00010~01100~10111~11001~00110~00101~00101~00110

Hybrid~Diff~Rules:~105~90~105~:

000~00100~01010~10001~11111~00100~00111~00111~00100

001~00101~01011~10000~11110~00100~00111~00111~00100

010~00111~01001~10010~11100~00100~00111~00111~00100

011~00110~01000~10011~11101~00100~00111~00111~00100

100~00001~01111~10100~11010~00111~00100~00100~00111

101~00000~01110~10101~11011~00111~00100~00100~00111

110~00010~01100~10111~11001~00111~00100~00100~00111

111~00011~01101~10110~11000~00111~00100~00100~00111

Hybrid~Diff~Rules:~105~90~165~:

000~00100~01011~10001~11110~00100~00110~00111~00101

001~00101~01010~10000~11111~00100~00110~00111~00101

010~00110~01001~10011~11100~00101~00111~00110~00100

011~00111~01000~10010~11101~00101~00111~00110~00100

100~00001~01110~10100~11011~00111~00101~00100~00110

101~00000~01111~10101~11010~00111~00101~00100~00110

110~00011~01100~10110~11001~00110~00100~00101~00111

111~00010~01101~10111~11000~00110~00100~00101~00111

Hybrid~Diff~Rules:~105~150~90~:

000~00111~01010~10000~11101~00111~00111~00110~00110

001~00110~01011~10001~11100~00111~00111~00110~00110

010~00101~01000~10010~11111~00110~00110~00111~00111

011~00100~01001~10011~11110~00110~00110~00111~00111

100~00000~01101~10111~11010~00110~00110~00111~00111

101~00001~01100~10110~11011~00110~00110~00111~00111

110~00010~01111~10101~11000~00111~00111~00110~00110

111~00011~01110~10100~11001~00111~00111~00110~00110

Hybrid~Diff~Rules:~105~150~150~:

000~00110~01011~10000~11101~00110~00110~00110~00110

001~00111~01010~10001~11100~00110~00110~00110~00110

010~00101~01000~10011~11110~00110~00110~00110~00110

011~00100~01001~10010~11111~00110~00110~00110~00110

100~00000~01101~10110~11011~00110~00110~00110~00110

101~00001~01100~10111~11010~00110~00110~00110~00110

110~00011~01110~10101~11000~00110~00110~00110~00110

111~00010~01111~10100~11001~00110~00110~00110~00110

Hybrid~Diff~Rules:~105~150~105~:

000~00111~01010~10001~11100~00111~00111~00111~00111

001~00110~01011~10000~11101~00111~00111~00111~00111

010~00100~01001~10010~11111~00111~00111~00111~00111

011~00101~01000~10011~11110~00111~00111~00111~00111

100~00001~01100~10111~11010~00111~00111~00111~00111

101~00000~01101~10110~11011~00111~00111~00111~00111

110~00010~01111~10100~11001~00111~00111~00111~00111

111~00011~01110~10101~11000~00111~00111~00111~00111

Hybrid~Diff~Rules:~105~150~165~:

000~00110~01011~10001~11100~00110~00110~00111~00111

001~00111~01010~10000~11101~00110~00110~00111~00111

010~00100~01001~10011~11110~00111~00111~00110~00110

011~00101~01000~10010~11111~00111~00111~00110~00110

100~00001~01100~10110~11011~00111~00111~00110~00110

101~00000~01101~10111~11010~00111~00111~00110~00110

110~00011~01110~10100~11001~00110~00110~00111~00111

111~00010~01111~10101~11000~00110~00110~00111~00111

Hybrid~Diff~Rules:~105~105~90~:

000~00101~01000~10010~11111~00101~00101~00100~00100

001~00100~01001~10011~11110~00101~00101~00100~00100

010~00111~01010~10000~11101~00100~00100~00101~00101

011~00110~01011~10001~11100~00100~00100~00101~00101

100~00010~01111~10101~11000~00100~00100~00101~00101

101~00011~01110~10100~11001~00100~00100~00101~00101

110~00000~01101~10111~11010~00101~00101~00100~00100

111~00001~01100~10110~11011~00101~00101~00100~00100

Hybrid~Diff~Rules:~105~105~150~:

000~00101~01000~10011~11110~00101~00101~00101~00101

001~00100~01001~10010~11111~00101~00101~00101~00101

010~00110~01011~10000~11101~00101~00101~00101~00101

011~00111~01010~10001~11100~00101~00101~00101~00101

100~00011~01110~10101~11000~00101~00101~00101~00101

101~00010~01111~10100~11001~00101~00101~00101~00101

110~00000~01101~10110~11011~00101~00101~00101~00101

111~00001~01100~10111~11010~00101~00101~00101~00101

Hybrid~Diff~Rules:~105~105~105~:

000~00100~01001~10010~11111~00100~00100~00100~00100

001~00101~01000~10011~11110~00100~00100~00100~00100

010~00111~01010~10001~11100~00100~00100~00100~00100

011~00110~01011~10000~11101~00100~00100~00100~00100

100~00010~01111~10100~11001~00100~00100~00100~00100

101~00011~01110~10101~11000~00100~00100~00100~00100

110~00001~01100~10111~11010~00100~00100~00100~00100

111~00000~01101~10110~11011~00100~00100~00100~00100

Hybrid~Diff~Rules:~105~105~165~:

000~00100~01001~10011~11110~00100~00100~00101~00101

001~00101~01000~10010~11111~00100~00100~00101~00101

010~00110~01011~10001~11100~00101~00101~00100~00100

011~00111~01010~10000~11101~00101~00101~00100~00100

100~00011~01110~10100~11001~00101~00101~00100~00100

101~00010~01111~10101~11000~00101~00101~00100~00100

110~00001~01100~10110~11011~00100~00100~00101~00101

111~00000~01101~10111~11010~00100~00100~00101~00101

Hybrid~Diff~Rules:~105~165~90~:

000~00111~01000~10010~11101~00111~00101~00100~00110

001~00110~01001~10011~11100~00111~00101~00100~00110

010~00101~01010~10000~11111~00110~00100~00101~00111

011~00100~01011~10001~11110~00110~00100~00101~00111

100~00010~01101~10111~11000~00100~00110~00111~00101

101~00011~01100~10110~11001~00100~00110~00111~00101

110~00000~01111~10101~11010~00101~00111~00110~00100

111~00001~01110~10100~11011~00101~00111~00110~00100

Hybrid~Diff~Rules:~105~165~150~:

000~00110~01000~10011~11101~00110~00101~00101~00110

001~00111~01001~10010~11100~00110~00101~00101~00110

010~00101~01011~10000~11110~00110~00101~00101~00110

011~00100~01010~10001~11111~00110~00101~00101~00110

100~00011~01101~10110~11000~00101~00110~00110~00101

101~00010~01100~10111~11001~00101~00110~00110~00101

110~00000~01110~10101~11011~00101~00110~00110~00101

111~00001~01111~10100~11010~00101~00110~00110~00101

Hybrid~Diff~Rules:~105~165~105~:

000~00111~01001~10010~11100~00111~00100~00100~00111

001~00110~01000~10011~11101~00111~00100~00100~00111

010~00100~01010~10001~11111~00111~00100~00100~00111

011~00101~01011~10000~11110~00111~00100~00100~00111

100~00010~01100~10111~11001~00100~00111~00111~00100

101~00011~01101~10110~11000~00100~00111~00111~00100

110~00001~01111~10100~11010~00100~00111~00111~00100

111~00000~01110~10101~11011~00100~00111~00111~00100

Hybrid~Diff~Rules:~105~165~165~:

000~00110~01001~10011~11100~00110~00100~00101~00111

001~00111~01000~10010~11101~00110~00100~00101~00111

010~00100~01011~10001~11110~00111~00101~00100~00110

011~00101~01010~10000~11111~00111~00101~00100~00110

100~00011~01100~10110~11001~00101~00111~00110~00100

101~00010~01101~10111~11000~00101~00111~00110~00100

110~00001~01110~10100~11011~00100~00110~00111~00101

111~00000~01111~10101~11010~00100~00110~00111~00101

Hybrid~Diff~Rules:~165~90~90~:

000~00101~01111~10000~11010~00101~00010~00110~00001

001~00100~01110~10001~11011~00101~00010~00110~00001

010~00111~01101~10010~11000~00100~00011~00111~00000

011~00110~01100~10011~11001~00100~00011~00111~00000

100~00000~01010~10101~11111~00110~00001~00101~00010

101~00001~01011~10100~11110~00110~00001~00101~00010

110~00010~01000~10111~11101~00111~00000~00100~00011

111~00011~01001~10110~11100~00111~00000~00100~00011

Hybrid~Diff~Rules:~165~90~150~:

000~00101~01110~10000~11011~00101~00011~00110~00000

001~00100~01111~10001~11010~00101~00011~00110~00000

010~00110~01101~10011~11000~00101~00011~00110~00000

011~00111~01100~10010~11001~00101~00011~00110~00000

100~00000~01011~10101~11110~00110~00000~00101~00011

101~00001~01010~10100~11111~00110~00000~00101~00011

110~00011~01000~10110~11101~00110~00000~00101~00011

111~00010~01001~10111~11100~00110~00000~00101~00011

Hybrid~Diff~Rules:~165~90~105~:

000~00100~01111~10001~11010~00100~00010~00111~00001

001~00101~01110~10000~11011~00100~00010~00111~00001

010~00111~01100~10010~11001~00100~00010~00111~00001

011~00110~01101~10011~11000~00100~00010~00111~00001

100~00001~01010~10100~11111~00111~00001~00100~00010

101~00000~01011~10101~11110~00111~00001~00100~00010

110~00010~01001~10111~11100~00111~00001~00100~00010

111~00011~01000~10110~11101~00111~00001~00100~00010

Hybrid~Diff~Rules:~165~90~165~:

000~00100~01110~10001~11011~00100~00011~00111~00000

001~00101~01111~10000~11010~00100~00011~00111~00000

010~00110~01100~10011~11001~00101~00010~00110~00001

011~00111~01101~10010~11000~00101~00010~00110~00001

100~00001~01011~10100~11110~00111~00000~00100~00011

101~00000~01010~10101~11111~00111~00000~00100~00011

110~00011~01001~10110~11100~00110~00001~00101~00010

111~00010~01000~10111~11101~00110~00001~00101~00010

Hybrid~Diff~Rules:~165~150~90~:

000~00111~01101~10000~11010~00111~00000~00110~00001

001~00110~01100~10001~11011~00111~00000~00110~00001

010~00101~01111~10010~11000~00110~00001~00111~00000

011~00100~01110~10011~11001~00110~00001~00111~00000

100~00000~01010~10111~11101~00110~00001~00111~00000

101~00001~01011~10110~11100~00110~00001~00111~00000

110~00010~01000~10101~11111~00111~00000~00110~00001

111~00011~01001~10100~11110~00111~00000~00110~00001

Hybrid~Diff~Rules:~165~150~150~:

000~00110~01101~10000~11011~00110~00000~00110~00000

001~00111~01100~10001~11010~00110~00000~00110~00000

010~00101~01110~10011~11000~00110~00000~00110~00000

011~00100~01111~10010~11001~00110~00000~00110~00000

100~00000~01011~10110~11101~00110~00000~00110~00000

101~00001~01010~10111~11100~00110~00000~00110~00000

110~00011~01000~10101~11110~00110~00000~00110~00000

111~00010~01001~10100~11111~00110~00000~00110~00000

Hybrid~Diff~Rules:~165~150~105~:

000~00111~01100~10001~11010~00111~00001~00111~00001

001~00110~01101~10000~11011~00111~00001~00111~00001

010~00100~01111~10010~11001~00111~00001~00111~00001

011~00101~01110~10011~11000~00111~00001~00111~00001

100~00001~01010~10111~11100~00111~00001~00111~00001

101~00000~01011~10110~11101~00111~00001~00111~00001

110~00010~01001~10100~11111~00111~00001~00111~00001

111~00011~01000~10101~11110~00111~00001~00111~00001

Hybrid~Diff~Rules:~165~150~165~:

000~00110~01100~10001~11011~00110~00001~00111~00000

001~00111~01101~10000~11010~00110~00001~00111~00000

010~00100~01110~10011~11001~00111~00000~00110~00001

011~00101~01111~10010~11000~00111~00000~00110~00001

100~00001~01011~10110~11100~00111~00000~00110~00001

101~00000~01010~10111~11101~00111~00000~00110~00001

110~00011~01001~10100~11110~00110~00001~00111~00000

111~00010~01000~10101~11111~00110~00001~00111~00000

Hybrid~Diff~Rules:~165~105~90~:

000~00101~01111~10010~11000~00101~00010~00100~00011

001~00100~01110~10011~11001~00101~00010~00100~00011

010~00111~01101~10000~11010~00100~00011~00101~00010

011~00110~01100~10001~11011~00100~00011~00101~00010

100~00010~01000~10101~11111~00100~00011~00101~00010

101~00011~01001~10100~11110~00100~00011~00101~00010

110~00000~01010~10111~11101~00101~00010~00100~00011

111~00001~01011~10110~11100~00101~00010~00100~00011

Hybrid~Diff~Rules:~165~105~150~:

000~00101~01110~10011~11000~00101~00011~00101~00011

001~00100~01111~10010~11001~00101~00011~00101~00011

010~00110~01101~10000~11011~00101~00011~00101~00011

011~00111~01100~10001~11010~00101~00011~00101~00011

100~00011~01000~10101~11110~00101~00011~00101~00011

101~00010~01001~10100~11111~00101~00011~00101~00011

110~00000~01011~10110~11101~00101~00011~00101~00011

111~00001~01010~10111~11100~00101~00011~00101~00011

Hybrid~Diff~Rules:~165~105~105~:

000~00100~01111~10010~11001~00100~00010~00100~00010

001~00101~01110~10011~11000~00100~00010~00100~00010

010~00111~01100~10001~11010~00100~00010~00100~00010

011~00110~01101~10000~11011~00100~00010~00100~00010

100~00010~01001~10100~11111~00100~00010~00100~00010

101~00011~01000~10101~11110~00100~00010~00100~00010

110~00001~01010~10111~11100~00100~00010~00100~00010

111~00000~01011~10110~11101~00100~00010~00100~00010

Hybrid~Diff~Rules:~165~105~165~:

000~00100~01110~10011~11001~00100~00011~00101~00010

001~00101~01111~10010~11000~00100~00011~00101~00010

010~00110~01100~10001~11011~00101~00010~00100~00011

011~00111~01101~10000~11010~00101~00010~00100~00011

100~00011~01001~10100~11110~00101~00010~00100~00011

101~00010~01000~10101~11111~00101~00010~00100~00011

110~00001~01011~10110~11100~00100~00011~00101~00010

111~00000~01010~10111~11101~00100~00011~00101~00010

Hybrid~Diff~Rules:~165~165~90~:

000~00111~01101~10010~11000~00111~00000~00100~00011

001~00110~01100~10011~11001~00111~00000~00100~00011

010~00101~01111~10000~11010~00110~00001~00101~00010

011~00100~01110~10001~11011~00110~00001~00101~00010

100~00010~01000~10111~11101~00100~00011~00111~00000

101~00011~01001~10110~11100~00100~00011~00111~00000

110~00000~01010~10101~11111~00101~00010~00110~00001

111~00001~01011~10100~11110~00101~00010~00110~00001

Hybrid~Diff~Rules:~165~165~150~:

000~00110~01101~10011~11000~00110~00000~00101~00011

001~00111~01100~10010~11001~00110~00000~00101~00011

010~00101~01110~10000~11011~00110~00000~00101~00011

011~00100~01111~10001~11010~00110~00000~00101~00011

100~00011~01000~10110~11101~00101~00011~00110~00000

101~00010~01001~10111~11100~00101~00011~00110~00000

110~00000~01011~10101~11110~00101~00011~00110~00000

111~00001~01010~10100~11111~00101~00011~00110~00000

Hybrid~Diff~Rules:~165~165~105~:

000~00111~01100~10010~11001~00111~00001~00100~00010

001~00110~01101~10011~11000~00111~00001~00100~00010

010~00100~01111~10001~11010~00111~00001~00100~00010

011~00101~01110~10000~11011~00111~00001~00100~00010

100~00010~01001~10111~11100~00100~00010~00111~00001

101~00011~01000~10110~11101~00100~00010~00111~00001

110~00001~01010~10100~11111~00100~00010~00111~00001

111~00000~01011~10101~11110~00100~00010~00111~00001

Hybrid~Diff~Rules:~165~165~165~:

000~00110~01100~10011~11001~00110~00001~00101~00010

001~00111~01101~10010~11000~00110~00001~00101~00010

010~00100~01110~10001~11011~00111~00000~00100~00011

011~00101~01111~10000~11010~00111~00000~00100~00011

100~00011~01001~10110~11100~00101~00010~00110~00001

101~00010~01000~10111~11101~00101~00010~00110~00001

110~00001~01011~10100~11110~00100~00011~00111~00000

111~00000~01010~10101~11111~00100~00011~00111~00000
\end{codesmall}

\subsection{\label{apn:CorrSeqSeeds}Correlation of Sequences to Seeds}

These experiments show the relationship between seeds and the sequences
they generate under various rule configurations. The SeedDiffFromTemp
function below creates a dictionary with all possible sequences as
keys and a list of the seeds that generate them as values. Numbers
here are often stored as native integers for performance reasons;
results showing decimal numbers are meant to be interpreted as the
binary equivalent of that number.

\medskip{}

This first experiment shows that rule sets having rules of the same
arity in each position generate the same temporal sequences with only
a constant factor added into the seed to account for rules that differ
only in their complimentarity. Rulesets having different arity in
corresponding rules generate the same sequences only from seeds with
no constant relationship across the rule sets. Indeed, depending on
the difference in rules, the same sequence may not even be possible
in another rule set. 
\begin{codesmall}
>\textcompwordmark{}>\textcompwordmark{}>~hr1~=~{[}90,~90,~90,~90,~90,~90,~90,~90,~150{]}

>\textcompwordmark{}>\textcompwordmark{}>~hr2~=~{[}90,~90,~90,~90,~90,~90,~90,~90,~105{]}

>\textcompwordmark{}>\textcompwordmark{}>~hr3~=~{[}90,~90,~90,~90,~90,~90,~150,~90,~150{]}

>\textcompwordmark{}>\textcompwordmark{}>~d1~=~hca.SeedDiffFromTemp(hr1,~9)

>\textcompwordmark{}>\textcompwordmark{}>~d2~=~hca.SeedDiffFromTemp(hr2,~9)

>\textcompwordmark{}>\textcompwordmark{}>~d3~=~hca.SeedDiffFromTemp(hr3,~9)

>\textcompwordmark{}>\textcompwordmark{}>~len(d3)

511

>\textcompwordmark{}>\textcompwordmark{}>~len(d2)

511

>\textcompwordmark{}>\textcompwordmark{}>~len(d1)

511

>\textcompwordmark{}>\textcompwordmark{}>~for~i~in~range(1,~10):

	print~'d1~\textasciicircum{}~d2:',

	pb(d1{[}i{]}{[}0{]}~\textasciicircum{}~d2{[}i{]}{[}0{]},~9)

	print~'d2~\textasciicircum{}~d3:',

	pb(d2{[}i{]}{[}0{]}~\textasciicircum{}~d3{[}i{]}{[}0{]},~9)

	print~'d1~\textasciicircum{}~d3:',

	pb(d1{[}i{]}{[}0{]}~\textasciicircum{}~d3{[}i{]}{[}0{]},~9)

d1~\textasciicircum{}~d2:~101101101

d1~\textasciicircum{}~d2:~101101101

d1~\textasciicircum{}~d2:~101101101

d1~\textasciicircum{}~d2:~101101101

d1~\textasciicircum{}~d2:~101101101

d1~\textasciicircum{}~d2:~101101101

d1~\textasciicircum{}~d2:~101101101

d1~\textasciicircum{}~d2:~101101101

d1~\textasciicircum{}~d2:~101101101~\\
~\\
d2~\textasciicircum{}~d3:~111101111

d2~\textasciicircum{}~d3:~000000000

d2~\textasciicircum{}~d3:~010000010

d2~\textasciicircum{}~d3:~000101001

d2~\textasciicircum{}~d3:~010101011

d2~\textasciicircum{}~d3:~101000100

d2~\textasciicircum{}~d3:~111000110

d2~\textasciicircum{}~d3:~101101110

d2~\textasciicircum{}~d3:~111101100~\\
~\\
d1~\textasciicircum{}~d3:~010000010

d1~\textasciicircum{}~d3:~101101101

d1~\textasciicircum{}~d3:~111101111

d1~\textasciicircum{}~d3:~101000100

d1~\textasciicircum{}~d3:~111000110

d1~\textasciicircum{}~d3:~000101001

d1~\textasciicircum{}~d3:~010101011

d1~\textasciicircum{}~d3:~000000011

d1~\textasciicircum{}~d3:~010000001
\end{codesmall}
This holds regardless of mixing arity or complimentarity, as shown
by the following rule sets.
\begin{codesmall}
>\textcompwordmark{}>\textcompwordmark{}>~hr4~=~{[}90{]}~{*}~8~+~{[}105{]}

>\textcompwordmark{}>\textcompwordmark{}>~d4~=~SeedDiffFromTemp(hr4,~9)

>\textcompwordmark{}>\textcompwordmark{}>~hr5~=~{[}165{]}~{*}~8~+~{[}150{]}

>\textcompwordmark{}>\textcompwordmark{}>~d5~=~SeedDiffFromTemp(hr5,~9)

>\textcompwordmark{}>\textcompwordmark{}>~hr1~=~{[}90{]}~{*}~8~+~{[}150{]}

>\textcompwordmark{}>\textcompwordmark{}>~for~j,~d~in~enumerate({[}d1,~d4,~d5{]}):

	for~k,~e~in~enumerate({[}d1,~d4,~d5{]}):

		if~d~!=~e:

			print~'d{[}'+str(j)+'{]}~\textasciicircum{}~e{[}'+str(k)+'{]}'

			for~i~in~range(1,~10):

				if~i~in~d~and~i~in~e:

					pb(d{[}i{]}{[}0{]}~\textasciicircum{}~e{[}i{]}{[}0{]},~9)

				elif~i~in~d:

					print~'right~does~not~have',

					pb(i,~9)

				else:

					print~'left~does~not~have',

					pb(i,9)

d{[}0{]}~\textasciicircum{}~e{[}1{]}

101101101

101101101

101101101

101101101

101101101

101101101

101101101

101101101

101101101

d{[}0{]}~\textasciicircum{}~e{[}2{]}

001001001

001001001

001001001

001001001

001001001

001001001

001001001

001001001

001001001

d{[}1{]}~\textasciicircum{}~e{[}0{]}

<...>

d{[}1{]}~\textasciicircum{}~e{[}2{]}

100100100

100100100

100100100

100100100

100100100

100100100

100100100

100100100

100100100

d{[}2{]}~\textasciicircum{}~e{[}0{]}

<...>

001001001

d{[}2{]}~\textasciicircum{}~e{[}1{]}

<...>

100100100
\end{codesmall}
It is possible for two rule sets to generate the same sequence up
to $t=n$, then diverge. This can happen when one rule set ``dead-ends''
(i.e. zero's out) or has a different period than another. Continuing
the environment from above, we see both of these conditions in two
rule sets:
\begin{codesmall}
>\textcompwordmark{}>\textcompwordmark{}>~hr2~=~{[}90,~90,~90,~90,~90,~90,~90,~90,~105{]}

>\textcompwordmark{}>\textcompwordmark{}>~hr3~=~{[}90,~90,~90,~90,~90,~90,~\uline{150},~90,~\uline{150}{]}

>\textcompwordmark{}>\textcompwordmark{}>~d2{[}141{]}

{[}417{]}

>\textcompwordmark{}>\textcompwordmark{}>~d3{[}141{]}

{[}137{]}

>\textcompwordmark{}>\textcompwordmark{}>~t2~=~hca.TempSeqFromSeed(hr2,~hca.SeedFromInt(417,9),18)

>\textcompwordmark{}>\textcompwordmark{}>~t3~=~hca.TempSeqFromSeed(hr3,~hca.SeedFromInt(137,9),18)

>\textcompwordmark{}>\textcompwordmark{}>~t2

array({[}0,~1,~0,~0,~0,~1,~1,~0,~1,~0,~0,~0,~1,~1,~0,~1,~0,~0,~0{]},~dtype=int8)

>\textcompwordmark{}>\textcompwordmark{}>~t3

array({[}0,~1,~0,~0,~0,~1,~1,~0,~1,~0,~0,~0,~\uline{0},~\uline{0},~0,~\uline{0},~0,~0,~0{]},~dtype=int8)
\end{codesmall}
The next experiment examines the hamming distance between seeds in
two asymmetric CA which generate the same temporal sequence. Here,
the sequence is printed in decimal next to the difference in seeds
for that sequence between the two rule sets. Some editing of the results
has been done to gather similar results. In the resulting table below,
all decimal sequences in each row share the same low 6 bits. The top
two sequences of each group of four in a column share the same upper
3 bits, and again for the bottom two of four. The middle two of each
group of four have the same last 3 bits, and the outer two have the
same last 4 bits.

The fact that these are groups of four seems to be related to the
rule set at play. For rule set hr1 (same as below) and hr2 = $\{90\}^{6}+\{150,90,150\}$
, the groups have 8 members of the same hamming distance. hr2 has
an eventual period of 0, hr1 of 28. Both have unique seeds for all
sequences.
\begin{codesmall}
>\textcompwordmark{}>\textcompwordmark{}>~hr1~=~{[}90,~90,~90,~90,~90,~90,~90,~90,~150{]}

>\textcompwordmark{}>\textcompwordmark{}>~(p,d)~=~MaxPeriod(hr1);~print~p

28

>\textcompwordmark{}>\textcompwordmark{}>~len(d1)

511

>\textcompwordmark{}>\textcompwordmark{}>~hr3~=~{[}150{]}~{*}~5~+~{[}90{]}~{*}~4

>\textcompwordmark{}>\textcompwordmark{}>~(p,d)~=~MaxPeriod(hr3);~print~p

30

>\textcompwordmark{}>\textcompwordmark{}>~d3~=~SeedDiffFromTemp(hr3,~9)

>\textcompwordmark{}>\textcompwordmark{}>~len(d3)

511

>\textcompwordmark{}>\textcompwordmark{}>~for~i~in~range(1,len(d1)):

	if~d1{[}i{]}~==~d3{[}i{]}:

		print~i,

59~195~248

>\textcompwordmark{}>\textcompwordmark{}>~for~i~in~range(1,len(d1)):

	print~i,~pb(d1{[}i{]}{[}0{]}~\textasciicircum{}~d3{[}i{]}{[}0{]})

~~0~000000000~~~~~~64~011101010~~~~~256~001001100~~~~~320~010100110

~59~000000000~~~~~123~011101010~~~~~315~001001100~~~~~379~010100110

195~000000000~~~~~131~011101010~~~~~451~001001100~~~~~387~010100110

258~000000000~~~~~184~011101010~~~~~504~001001100~~~~~440~010100110

	~~~~~~~~~~~~~~~~~~~~~~~~~~~~~~~~~~~~~~~~~~~~~~~~~~~

~~1~011101110~~~~~~65~000000100~~~~~257~010100010~~~~~321~001001000

~58~011101110~~~~~122~000000100~~~~~314~010100010~~~~~378~001001000

194~011101110~~~~~130~000000100~~~~~450~010100010~~~~~386~001001000

249~011101110~~~~~185~000000100~~~~~505~010100010~~~~~441~001001000

~~~~~~~~~~~~~~~~~~~~~~~~~~~~~~~~~~~~~~~~~~~~~~~~~~~~~~

~~2~110000101~~~~~~66~101101111~~~~~258~111001001~~~~~322~100100011

~57~110000101~~~~~121~101101111~~~~~313~111001001~~~~~377~100100011

193~110000101~~~~~129~101101111~~~~~449~111001001~~~~~385~100100011

250~110000101~~~~~186~101101111~~~~~506~111001001~~~~~442~100100011

~~~~~~~~~~~~~~~~~~~~~~~~~~~~~~~~~~~~~~~~~~~~~~~~~~~~~~

~~3~101101011~~~~~~67~110000001~~~~~259~100100111~~~~~323~111001101

~56~101101011~~~~~120~110000001~~~~~312~100100111~~~~~376~111001101

192~101101011~~~~~128~110000001~~~~~448~100100111~~~~~384~111001101

251~101101011~~~~~187~110000001~~~~~507~100100111~~~~~443~111001101

~~~~~~~~~~~~~~~~~~~~~~~~~~~~~~~~~~~~~~~~~~~~~~~~~~~~~~

~~4~111000111~~~~~~68~100101101~~~~~260~110001011~~~~~324~101100001

~63~111000111~~~~~127~100101101~~~~~319~110001011~~~~~383~101100001

199~111000111~~~~~135~100101101~~~~~455~110001011~~~~~391~101100001

252~111000111~~~~~188~100101101~~~~~508~110001011~~~~~444~101100001

~~~~~~~~~~~~~~~~~~~~~~~~~~~~~~~~~~~~~~~~~~~~~~~~~~~~~~

~~5~100101001~~~~~~69~111000011~~~~~261~101100101~~~~~325~110001111

~62~100101001~~~~~126~111000011~~~~~318~101100101~~~~~382~110001111

198~100101001~~~~~134~111000011~~~~~454~101100101~~~~~390~110001111

253~100101001~~~~~189~111000011~~~~~509~101100101~~~~~445~110001111

~~~~~~~~~~~~~~~~~~~~~~~~~~~~~~~~~~~~~~~~~~~~~~~~~~~~~~

~~6~001000010~~~~~~70~010101000~~~~~262~000001110~~~~~326~011100100

~61~001000010~~~~~125~010101000~~~~~317~000001110~~~~~381~011100100

197~001000010~~~~~133~010101000~~~~~453~000001110~~~~~389~011100100

254~001000010~~~~~190~010101000~~~~~510~000001110~~~~~446~011100100

~~~~~~~~~~~~~~~~~~~~~~~~~~~~~~~~~~~~~~~~~~~~~~~~~~~~~~

~~7~010101100~~~~~~71~001000110~~~~~263~011100000~~~~~327~000001010

~60~010101100~~~~~124~001000110~~~~~316~011100000~~~~~380~000001010

196~010101100~~~~~132~001000110~~~~~452~011100000~~~~~388~000001010

255~010101100~~~~~191~001000110~~~?\_511~011100000\_?~~~447~000001010~~

~~~~~~~~~~~~~~~~~~~~~~~~~~~~~~~~~~~~~~~~~~~~~~~~~~~~~~

~~8~110000010~~~~~~72~101101000~~~~~264~111001110~~~~~328~100100100

~51~110000010~~~~~115~101101000~~~~~307~111001110~~~~~371~100100100

203~110000010~~~~~139~101101000~~~~~459~111001110~~~~~395~100100100

240~110000010~~~~~176~101101000~~~~~496~111001110~~~~~432~100100100

~~~~~~~~~~~~~~~~~~~~~~~~~~~~~~~~~~~~~~~~~~~~~~~~~~~~~~

~~9~101101100~~~~~~73~110000110~~~~~265~100100000~~~~~329~111001010

~50~101101100~~~~~114~110000110~~~~~306~100100000~~~~~370~111001010

202~101101100~~~~~138~110000110~~~~~458~100100000~~~~~394~111001010

241~101101100~~~~~177~110000110~~~~~497~100100000~~~~~433~111001010

~~~~~~~~~~~~~~~~~~~~~~~~~~~~~~~~~~~~~~~~~~~~~~~~~~~~~~

~10~000000111~~~~~~74~011101101~~~~~266~001001011~~~~~330~010100001

~49~000000111~~~~~113~011101101~~~~~305~001001011~~~~~369~010100001

201~000000111~~~~~137~011101101~~~~~457~001001011~~~~~393~010100001

242~000000111~~~~~178~011101101~~~~~498~001001011~~~~~434~010100001

~~~~~~~~~~~~~~~~~~~~~~~~~~~~~~~~~~~~~~~~~~~~~~~~~~~~~~

~11~011101001~~~~~~75~000000011~~~~~267~010100101~~~~~331~001001111

~48~011101001~~~~~112~000000011~~~~~304~010100101~~~~~368~001001111

200~011101001~~~~~136~000000011~~~~~456~010100101~~~~~392~001001111

243~011101001~~~~~179~000000011~~~~~499~010100101~~~~~435~001001111

~~~~~~~~~~~~~~~~~~~~~~~~~~~~~~~~~~~~~~~~~~~~~~~~~~~~~~

~12~001000101~~~~~~76~010101111~~~~~268~000001001~~~~~332~011100011

~55~001000101~~~~~119~010101111~~~~~311~000001001~~~~~375~011100011

207~001000101~~~~~143~010101111~~~~~463~000001001~~~~~399~011100011

244~001000101~~~~~180~010101111~~~~~500~000001001~~~~~436~011100011

~~~~~~~~~~~~~~~~~~~~~~~~~~~~~~~~~~~~~~~~~~~~~~~~~~~~~~

~13~010101011~~~~~~77~001000001~~~~~269~011100111~~~~~333~000001101

~54~010101011~~~~~118~001000001~~~~~310~011100111~~~~~374~000001101

206~010101011~~~~~142~001000001~~~~~462~011100111~~~~~398~000001101

245~010101011~~~~~181~001000001~~~~~501~011100111~~~~~437~000001101

~~~~~~~~~~~~~~~~~~~~~~~~~~~~~~~~~~~~~~~~~~~~~~~~~~~~~~

~14~111000000~~~~~~78~100101010~~~~~270~110001100~~~~~334~101100110

~53~111000000~~~~~117~100101010~~~~~309~110001100~~~~~373~101100110

205~111000000~~~~~141~100101010~~~~~461~110001100~~~~~397~101100110

246~111000000~~~~~182~100101010~~~~~502~110001100~~~~~438~101100110

~~~~~~~~~~~~~~~~~~~~~~~~~~~~~~~~~~~~~~~~~~~~~~~~~~~~~~

~15~100101110~~~~~~79~111000100~~~~~271~101100010~~~~~335~110001000

~52~100101110~~~~~116~111000100~~~~~308~101100010~~~~~372~110001000

204~100101110~~~~~140~111000100~~~~~460~101100010~~~~~396~110001000

247~100101110~~~~~183~111000100~~~~~503~101100010~~~~~439~110001000

~~~~~~~~~~~~~~~~~~~~~~~~~~~~~~~~~~~~~~~~~~~~~~~~~~~~~~

~16~110101101~~~~~~80~101000111~~~~~272~111100001~~~~~336~100001011

~43~110101101~~~~~107~101000111~~~~~299~111100001~~~~~363~100001011

211~110101101~~~~~147~101000111~~~~~467~111100001~~~~~403~100001011

232~110101101~~~~~168~101000111~~~~~488~111100001~~~~~424~100001011

~~~~~~~~~~~~~~~~~~~~~~~~~~~~~~~~~~~~~~~~~~~~~~~~~~~~~~

~17~101000011~~~~~~81~110101001~~~~~273~100001111~~~~~337~111100101

~42~101000011~~~~~106~110101001~~~~~298~100001111~~~~~362~111100101

210~101000011~~~~~146~110101001~~~~~466~100001111~~~~~402~111100101

233~101000011~~~~~169~110101001~~~~~489~100001111~~~~~425~111100101

~~~~~~~~~~~~~~~~~~~~~~~~~~~~~~~~~~~~~~~~~~~~~~~~~~~~~~

~18~000101000~~~~~~82~011000010~~~~~274~001100100~~~~~338~010001110

~41~000101000~~~~~105~011000010~~~~~297~001100100~~~~~361~010001110

209~000101000~~~~~145~011000010~~~~~465~001100100~~~~~401~010001110

234~000101000~~~~~170~011000010~~~~~490~001100100~~~~~426~010001110

~~~~~~~~~~~~~~~~~~~~~~~~~~~~~~~~~~~~~~~~~~~~~~~~~~~~~~

~19~011000110~~~~~~83~000101100~~~~~275~010001010~~~~~339~001100000

~40~011000110~~~~~104~000101100~~~~~296~010001010~~~~~360~001100000

208~011000110~~~~~144~000101100~~~~~464~010001010~~~~~400~001100000

235~011000110~~~~~171~000101100~~~~~491~010001010~~~~~427~001100000

~~~~~~~~~~~~~~~~~~~~~~~~~~~~~~~~~~~~~~~~~~~~~~~~~~~~~~

~20~001101010~~~~~~84~010000000~~~~~276~000100110~~~~~340~011001100

~47~001101010~~~~~111~010000000~~~~~303~000100110~~~~~367~011001100

215~001101010~~~~~151~010000000~~~~~471~000100110~~~~~407~011001100

236~001101010~~~~~172~010000000~~~~~492~000100110~~~~~428~011001100

~~~~~~~~~~~~~~~~~~~~~~~~~~~~~~~~~~~~~~~~~~~~~~~~~~~~~~

~21~010000100~~~~~~85~001101110~~~~~277~011001000~~~~~341~000100010

~46~010000100~~~~~110~001101110~~~~~302~011001000~~~~~366~000100010

214~010000100~~~~~150~001101110~~~~~470~011001000~~~~~406~000100010

237~010000100~~~~~173~001101110~~~~~493~011001000~~~~~429~000100010

~~~~~~~~~~~~~~~~~~~~~~~~~~~~~~~~~~~~~~~~~~~~~~~~~~~~~~

~22~111101111~~~~~~86~100000101~~~~~278~110100011~~~~~342~101001001

~45~111101111~~~~~109~100000101~~~~~301~110100011~~~~~365~101001001

213~111101111~~~~~149~100000101~~~~~469~110100011~~~~~405~101001001

238~111101111~~~~~174~100000101~~~~~494~110100011~~~~~430~101001001

~~~~~~~~~~~~~~~~~~~~~~~~~~~~~~~~~~~~~~~~~~~~~~~~~~~~~~

~23~100000001~~~~~~87~111101011~~~~~279~101001101~~~~~343~110100111

~44~100000001~~~~~108~111101011~~~~~300~101001101~~~~~364~110100111

212~100000001~~~~~148~111101011~~~~~468~101001101~~~~~404~110100111

239~100000001~~~~~175~111101011~~~~~495~101001101~~~~~431~110100111

~~~~~~~~~~~~~~~~~~~~~~~~~~~~~~~~~~~~~~~~~~~~~~~~~~~~~~

~24~000101111~~~~~~88~011000101~~~~~280~001100011~~~~~344~010001001

~35~000101111~~~~~~99~011000101~~~~~291~001100011~~~~~355~010001001

219~000101111~~~~~155~011000101~~~~~475~001100011~~~~~411~010001001

224~000101111~~~~~160~011000101~~~~~480~001100011~~~~~416~010001001

~~~~~~~~~~~~~~~~~~~~~~~~~~~~~~~~~~~~~~~~~~~~~~~~~~~~~~

~25~011000001~~~~~~89~000101011~~~~~281~010001101~~~~~345~001100111

~34~011000001~~~~~~98~000101011~~~~~290~010001101~~~~~354~001100111

218~011000001~~~~~154~000101011~~~~~474~010001101~~~~~410~001100111

225~011000001~~~~~161~000101011~~~~~481~010001101~~~~~417~001100111

~~~~~~~~~~~~~~~~~~~~~~~~~~~~~~~~~~~~~~~~~~~~~~~~~~~~~~

~26~110101010~~~~~~90~101000000~~~~~282~111100110~~~~~346~100001100

~33~110101010~~~~~~97~101000000~~~~~289~111100110~~~~~353~100001100

217~110101010~~~~~153~101000000~~~~~473~111100110~~~~~409~100001100

226~110101010~~~~~162~101000000~~~~~482~111100110~~~~~418~100001100

~~~~~~~~~~~~~~~~~~~~~~~~~~~~~~~~~~~~~~~~~~~~~~~~~~~~~~

~27~101000100~~~~~~91~110101110~~~~~283~100001000~~~~~347~111100010

~32~101000100~~~~~~96~110101110~~~~~288~100001000~~~~~352~111100010

216~101000100~~~~~152~110101110~~~~~472~100001000~~~~~408~111100010

227~101000100~~~~~163~110101110~~~~~483~100001000~~~~~419~111100010

~~~~~~~~~~~~~~~~~~~~~~~~~~~~~~~~~~~~~~~~~~~~~~~~~~~~~~

~28~111101000~~~~~~92~100000010~~~~~284~110100100~~~~~348~101001110

~39~111101000~~~~~103~100000010~~~~~295~110100100~~~~~359~101001110

223~111101000~~~~~159~100000010~~~~~479~110100100~~~~~415~101001110

228~111101000~~~~~164~100000010~~~~~484~110100100~~~~~420~101001110

~~~~~~~~~~~~~~~~~~~~~~~~~~~~~~~~~~~~~~~~~~~~~~~~~~~~~~

~29~100000110~~~~~~93~111101100~~~~~285~101001010~~~~~349~110100000

~38~100000110~~~~~102~111101100~~~~~294~101001010~~~~~358~110100000

222~100000110~~~~~158~111101100~~~~~478~101001010~~~~~414~110100000

229~100000110~~~~~165~111101100~~~~~485~101001010~~~~~421~110100000

~~~~~~~~~~~~~~~~~~~~~~~~~~~~~~~~~~~~~~~~~~~~~~~~~~~~~~

~30~001101101~~~~~~94~010000111~~~~~286~000100001~~~~~350~011001011

~37~001101101~~~~~101~010000111~~~~~293~000100001~~~~~357~011001011

221~001101101~~~~~157~010000111~~~~~477~000100001~~~~~413~011001011

230~001101101~~~~~166~010000111~~~~~486~000100001~~~~~422~011001011

~~~~~~~~~~~~~~~~~~~~~~~~~~~~~~~~~~~~~~~~~~~~~~~~~~~~~~

~31~010000011~~~~~~95~001101001~~~~~287~011001111~~~~~351~000100101

~36~010000011~~~~~100~001101001~~~~~292~011001111~~~~~356~000100101

220~010000011~~~~~156~001101001~~~~~476~011001111~~~~~412~000100101

231~010000011~~~~~167~001101001~~~~~487~011001111~~~~~423~000100101
\end{codesmall}

\subsection{\label{apn:MapSeedsExpr}Mapping Seeds from Symmetric to Asymmetric
Rulesets}

These experiments investigate the possibility of using information
about a sequence under a symmetric rule set to learn something about
the seed or the rule set that originally produced it.

First, we go through all sequences produced by a uniform rule 150
CA, and find the hamming distance from each seed that produces that
sequence to the seed under an asymmetric CA which produces the same
sequence.
\begin{codesmall}
>\textcompwordmark{}>\textcompwordmark{}>~hr3~=~{[}90,~90,~90,~90,~90,~90,~150,~90,~105{]}

>\textcompwordmark{}>\textcompwordmark{}>~d3~=~SeedDiffFromTemp(hr3,~9)

>\textcompwordmark{}>\textcompwordmark{}>~hr5~=~{[}150{]}~{*}~9

>\textcompwordmark{}>\textcompwordmark{}>~d5~=~SeedDiffFromTemp(hr5,~9)

>\textcompwordmark{}>\textcompwordmark{}>~intersect~=~{[}{]}

>\textcompwordmark{}>\textcompwordmark{}>~for~k~in~d5.keys():

	if~k~==~0:

		continue

	vl~=~d5{[}k{]}

	intersect.append(k)

	print~k,

	for~v~in~vl:

		print~pb(v~\textasciicircum{}~d3{[}k{]}{[}0{]}),

	print

\uline{256}~\uline{110101011}~110000011~111101111~111000111~100101001~100000001~101101101~101000101~

~~~~010101010~010000010~011101110~011000110~000101000~000000000~001101100~001000100

278~001000100~001101100~000000000~000101000~011000110~011101110~010000010~010101010~

~~~~101000101~101101101~100000001~100101001~111000111~111101111~110000011~110101011

300~101000101~101101101~100000001~100101001~111000111~111101111~110000011~110101011~

~~~~001000100~001101100~000000000~000101000~011000110~011101110~010000010~010101010

395~110000011~110101011~111000111~111101111~100000001~100101001~101000101~101101101~

~~~~010000010~010101010~011000110~011101110~000000000~000101000~001000100~001101100

344~011000110~011101110~010000010~010101010~001000100~001101100~000000000~000101000~

~~~~111000111~111101111~110000011~110101011~101000101~101101101~100000001~100101001

\uline{22}~~\uline{111101111}~111000111~110101011~110000011~101101101~101000101~100101001~100000001~

~~~~011101110~011000110~010101010~010000010~001101100~001000100~000101000~000000000

413~001101100~001000100~000101000~000000000~011101110~011000110~010101010~010000010~

~~~~101101101~101000101~100101001~100000001~111101111~111000111~110101011~110000011

453~100000001~100101001~101000101~101101101~110000011~110101011~111000111~111101111~

~~~~000000000~000101000~001000100~001101100~010000010~010101010~011000110~011101110

167~011000110~011101110~010000010~010101010~001000100~001101100~000000000~000101000~

~~~~111000111~111101111~110000011~110101011~101000101~101101101~100000001~100101001

44~~011101110~011000110~010101010~010000010~001101100~001000100~000101000~000000000~

~~~~111101111~111000111~110101011~110000011~101101101~101000101~100101001~100000001

157~111000111~111101111~110000011~110101011~101000101~101101101~100000001~100101001~

~~~~011000110~011101110~010000010~010101010~001000100~001101100~000000000~000101000

177~100101001~100000001~101101101~101000101~110101011~110000011~111101111~111000111~

~~~~000101000~000000000~001101100~001000100~010101010~010000010~011101110~011000110

372~000101000~000000000~001101100~001000100~010101010~010000010~011101110~011000110~

~~~~100101001~100000001~101101101~101000101~110101011~110000011~111101111~111000111

58~~100000001~100101001~101000101~101101101~110000011~110101011~111000111~111101111~

~~~~000000000~000101000~001000100~001101100~010000010~010101010~011000110~011101110

139~000101000~000000000~001101100~001000100~010101010~010000010~011101110~011000110~

~~~~100101001~100000001~101101101~101000101~110101011~110000011~111101111~111000111

197~010101010~010000010~011101110~011000110~000101000~000000000~001101100~001000100~

~~~~110101011~110000011~111101111~111000111~100101001~100000001~101101101~101000101

354~111000111~111101111~110000011~110101011~101000101~101101101~100000001~100101001~

~~~~011000110~011101110~010000010~010101010~001000100~001101100~000000000~000101000

78~~010000010~010101010~011000110~011101110~000000000~000101000~001000100~001101100~

~~~~110000011~110101011~111000111~111101111~100000001~100101001~101000101~101101101

433~010000010~010101010~011000110~011101110~000000000~000101000~001000100~001101100~

~~~~110000011~110101011~111000111~111101111~100000001~100101001~101000101~101101101

467~011101110~011000110~010101010~010000010~001101100~001000100~000101000~000000000~

~~~~111101111~111000111~110101011~110000011~101101101~101000101~100101001~100000001

334~100101001~100000001~101101101~101000101~110101011~110000011~111101111~111000111~

~~~~000101000~000000000~001101100~001000100~010101010~010000010~011101110~011000110

88~~101101101~101000101~100101001~100000001~111101111~111000111~110101011~110000011~

~~~~001101100~001000100~000101000~000000000~011101110~011000110~010101010~010000010

314~010101010~010000010~011101110~011000110~000101000~000000000~001101100~001000100~

~~~~110101011~110000011~111101111~111000111~100101001~100000001~101101101~101000101

98~~001101100~001000100~000101000~000000000~011101110~011000110~010101010~010000010~

~~~~101101101~101000101~100101001~100000001~111101111~111000111~110101011~110000011

233~001000100~001101100~000000000~000101000~011000110~011101110~010000010~010101010~

~~~~101000101~101101101~100000001~100101001~111000111~111101111~110000011~110101011

423~101101101~101000101~100101001~100000001~111101111~111000111~110101011~110000011~

~~~~001101100~001000100~000101000~000000000~011101110~011000110~010101010~010000010

211~101000101~101101101~100000001~100101001~111000111~111101111~110000011~110101011~

~~~~001000100~001101100~000000000~000101000~011000110~011101110~010000010~010101010

116~110000011~110101011~111000111~111101111~100000001~100101001~101000101~101101101~

~~~~010000010~010101010~011000110~011101110~000000000~000101000~001000100~001101100

\uline{489}~\uline{111101111}~111000111~110101011~110000011~101101101~101000101~100101001~100000001~

~~~~011101110~011000110~010101010~010000010~001101100~001000100~000101000~000000000

\uline{255}~\uline{110101011}~110000011~111101111~111000111~100101001~100000001~101101101~101000101~

~~~~010101010~010000010~011101110~011000110~000101000~000000000~001101100~001000100

511~000000000~000101000~001000100~001101100~010000010~010101010~011000110~011101110~

~~~~100000001~100101001~101000101~101101101~110000011~110101011~111000111~111101111

>\textcompwordmark{}>\textcompwordmark{}>~len(intersect)

31

~

>\textcompwordmark{}>\textcompwordmark{}>~pb(22)

'000010110'

>\textcompwordmark{}>\textcompwordmark{}>~pb(489)

'111101001'

>\textcompwordmark{}>\textcompwordmark{}>~pb(255)

'011111111'

>\textcompwordmark{}>\textcompwordmark{}>~pb(256)

'100000000'

\end{codesmall}
Notice that complimentary sequences have the same pattern of differences.
Seed differences for 22 and 489 above both start with 111101111, sequences
for 255 and 256 start with 110101011, etc.

\subsection{\label{apn:MapSeqExpr}Mapping Sequences from Symmetric to Asymmetric
CAs}

In the following experiment, continuing the above environment, all
seeds for all sequences from a 9-cell uniform rule 150 CA are looked
up in the reverse dictionary (where the key is the seed and the resulting
sequence is the value) of an asymmetric CA. The differences between
the two sequences are shown. The interesting observation is that there
are only differences in the last four bits. This pattern is dependent
on the two rules (or more likely, just the asymmetric CA rule set)
at play. Other CA pairs show differences in earlier bit positions
of the temporal sequence, but those cases have a more predictable
arrangement of differences across seed in the symmetric CA.

Again, we note that complementary sequences have the same pattern
of differences between symmetric and asymmetric CAs. See sequences
010001011 and 101110100 below.
\begin{codesmall}
>\textcompwordmark{}>\textcompwordmark{}>~hr3~=~{[}90,~90,~90,~90,~90,~90,~150,~90,~105{]}

>\textcompwordmark{}>\textcompwordmark{}>~d3~=~SeedDiffFromTemp(hr3,~9)

>\textcompwordmark{}>\textcompwordmark{}>~hr5~=~{[}150{]}~{*}~9

>\textcompwordmark{}>\textcompwordmark{}>~d5~=~SeedDiffFromTemp(hr5,~9)

>\textcompwordmark{}>\textcompwordmark{}>~d3r~=~\{\}

>\textcompwordmark{}>\textcompwordmark{}>~for~k~in~d3:

	d6r{[}d6{[}k{]}{[}0{]}{]}~=~k

>\textcompwordmark{}>\textcompwordmark{}>~for~t~in~d5:

	print~'sym~seq='~+~pb(t)~+~':'

	for~j,~i~in~enumerate(d5{[}t{]}):

		print~'{[}'~+~pb(i)~+~':'~+~pb(d3r{[}i{]})~+~':'~+~pb(d3r{[}i{]}\textasciicircum{}t)~+~'{]}',

		if~j~\&~1:

			print

key:~~{[}~Ssym~:~Tasym~:~T~\textasciicircum{}~Tasym{]}~~where

~~~T~=~temporal~sequence~generated~by~symmetrical~CA,

~~~Ssym~=~seed~of~symmetrical~CA~that~generates~T

~~~Tasym~=~temporal~sequence~of~asymmetrical~CA~on~seed~Ssym

sym~seq=100000000:

{[}000011101:100001000:000001000{]}~{[}000110101:100001001:000001001{]}

{[}001011001:100001011:000001011{]}~{[}001110001:100001010:000001010{]}

{[}010011111:100001101:000001101{]}~{[}010110111:100001100:000001100{]}

{[}011011011:100001110:000001110{]}~{[}011110011:100001111:000001111{]}

{[}100011100:100000100:000000100{]}~{[}100110100:100000101:000000101{]}

{[}101011000:100000111:000000111{]}~{[}101110000:100000110:000000110{]}

{[}110011110:100000001:000000001{]}~{[}110110110:100000000:000000000{]}

{[}111011010:100000010:000000010{]}~{[}111110010:100000011:000000011{]}

sym~seq=000000000:

{[}000101000:000000001:000000001{]}~{[}001000100:000000011:000000011{]}

{[}001101100:000000010:000000010{]}~{[}010000010:000000101:000000101{]}

{[}010101010:000000100:000000100{]}~{[}011000110:000000110:000000110{]}

{[}011101110:000000111:000000111{]}~{[}100000001:000001100:000001100{]}

{[}100101001:000001101:000001101{]}~{[}101000101:000001111:000001111{]}

{[}101101101:000001110:000001110{]}~{[}110000011:000001001:000001001{]}

{[}110101011:000001000:000001000{]}~{[}111000111:000001010:000001010{]}

{[}111101111:000001011:000001011{]}~

sym~seq=100010110:

{[}000011100:100010101:000000011{]}~{[}000110100:100010100:000000010{]}

{[}001011000:100010110:000000000{]}~{[}001110000:100010111:000000001{]}

{[}010011110:100010000:000000110{]}~{[}010110110:100010001:000000111{]}

{[}011011010:100010011:000000101{]}~{[}011110010:100010010:000000100{]}

{[}100011101:100011001:000001111{]}~{[}100110101:100011000:000001110{]}

{[}101011001:100011010:000001100{]}~{[}101110001:100011011:000001101{]}

{[}110011111:100011100:000001010{]}~{[}110110111:100011101:000001011{]}

{[}111011011:100011111:000001001{]}~{[}111110011:100011110:000001000{]}

sym~seq=100101100:

{[}000011111:100100011:000001111{]}~{[}000110111:100100010:000001110{]}

{[}001011011:100100000:000001100{]}~{[}001110011:100100001:000001101{]}

{[}010011101:100100110:000001010{]}~{[}010110101:100100111:000001011{]}

{[}011011001:100100101:000001001{]}~{[}011110001:100100100:000001000{]}

{[}100011110:100101111:000000011{]}~{[}100110110:100101110:000000010{]}

{[}101011010:100101100:000000000{]}~{[}101110010:100101101:000000001{]}

{[}110011100:100101010:000000110{]}~{[}110110100:100101011:000000111{]}

{[}111011000:100101001:000000101{]}~{[}111110000:100101000:000000100{]}

sym~seq=110001011:

{[}000010101:110000010:000001001{]}~{[}000111101:110000011:000001000{]}

{[}001010001:110000001:000001010{]}~{[}001111001:110000000:000001011{]}

{[}010010111:110000111:000001100{]}~{[}010111111:110000110:000001101{]}

{[}011010011:110000100:000001111{]}~{[}011111011:110000101:000001110{]}

{[}100010100:110001110:000000101{]}~{[}100111100:110001111:000000100{]}

{[}101010000:110001101:000000110{]}~{[}101111000:110001100:000000111{]}

{[}110010110:110001011:000000000{]}~{[}110111110:110001010:000000001{]}

{[}111010010:110001000:000000011{]}~{[}111111010:110001001:000000010{]}

sym~seq=101011000:

{[}000011010:101011110:000000110{]}~{[}000110010:101011111:000000111{]}

{[}001011110:101011101:000000101{]}~{[}001110110:101011100:000000100{]}

{[}010011000:101011011:000000011{]}~{[}010110000:101011010:000000010{]}

{[}011011100:101011000:000000000{]}~{[}011110100:101011001:000000001{]}

{[}100011011:101010010:000001010{]}~{[}100110011:101010011:000001011{]}

{[}101011111:101010001:000001001{]}~{[}101110111:101010000:000001000{]}

{[}110011001:101010111:000001111{]}~{[}110110001:101010110:000001110{]}

{[}111011101:101010100:000001100{]}~{[}111110101:101010101:000001101{]}

sym~seq=000010110:

{[}000000001:000011101:000001011{]}~{[}000101001:000011100:000001010{]}

{[}001000101:000011110:000001000{]}~{[}001101101:000011111:000001001{]}

{[}010000011:000011000:000001110{]}~{[}010101011:000011001:000001111{]}

{[}011000111:000011011:000001101{]}~{[}011101111:000011010:000001100{]}

{[}100000000:000010001:000000111{]}~{[}100101000:000010000:000000110{]}

{[}101000100:000010010:000000100{]}~{[}101101100:000010011:000000101{]}

{[}110000010:000010100:000000010{]}~{[}110101010:000010101:000000011{]}

{[}111000110:000010111:000000001{]}~{[}111101110:000010110:000000000{]}

sym~seq=110011101:

{[}000010100:110011111:000000010{]}~{[}000111100:110011110:000000011{]}

{[}001010000:110011100:000000001{]}~{[}001111000:110011101:000000000{]}

{[}010010110:110011010:000000111{]}~{[}010111110:110011011:000000110{]}

{[}011010010:110011001:000000100{]}~{[}011111010:110011000:000000101{]}

{[}100010101:110010011:000001110{]}~{[}100111101:110010010:000001111{]}

{[}101010001:110010000:000001101{]}~{[}101111001:110010001:000001100{]}

{[}110010111:110010110:000001011{]}~{[}110111111:110010111:000001010{]}

{[}111010011:110010101:000001000{]}~{[}111111011:110010100:000001001{]}

sym~seq=111000101:

{[}000010011:111001001:000001100{]}~{[}000111011:111001000:000001101{]}

{[}001010111:111001010:000001111{]}~{[}001111111:111001011:000001110{]}

{[}010010001:111001100:000001001{]}~{[}010111001:111001101:000001000{]}

{[}011010101:111001111:000001010{]}~{[}011111101:111001110:000001011{]}

{[}100010010:111000101:000000000{]}~{[}100111010:111000100:000000001{]}

{[}101010110:111000110:000000011{]}~{[}101111110:111000111:000000010{]}

{[}110010000:111000000:000000101{]}~{[}110111000:111000001:000000100{]}

{[}111010100:111000011:000000110{]}~{[}111111100:111000010:000000111{]}

sym~seq=010100111:

{[}000001010:010100001:000000110{]}~{[}000100010:010100000:000000111{]}

{[}001001110:010100010:000000101{]}~{[}001100110:010100011:000000100{]}

{[}010001000:010100100:000000011{]}~{[}010100000:010100101:000000010{]}

{[}011001100:010100111:000000000{]}~{[}011100100:010100110:000000001{]}

{[}100001011:010101101:000001010{]}~{[}100100011:010101100:000001011{]}

{[}101001111:010101110:000001001{]}~{[}101100111:010101111:000001000{]}

{[}110001001:010101000:000001111{]}~{[}110100001:010101001:000001110{]}

{[}111001101:010101011:000001100{]}~{[}111100101:010101010:000001101{]}

sym~seq=000101100:

{[}000000010:000101011:000000111{]}~{[}000101010:000101010:000000110{]}

{[}001000110:000101000:000000100{]}~{[}001101110:000101001:000000101{]}

{[}010000000:000101110:000000010{]}~{[}010101000:000101111:000000011{]}

{[}011000100:000101101:000000001{]}~{[}011101100:000101100:000000000{]}

{[}100000011:000100111:000001011{]}~{[}100101011:000100110:000001010{]}

{[}101000111:000100100:000001000{]}~{[}101101111:000100101:000001001{]}

{[}110000001:000100010:000001110{]}~{[}110101001:000100011:000001111{]}

{[}111000101:000100001:000001101{]}~{[}111101101:000100000:000001100{]}

sym~seq=010011101:

{[}000001001:010010111:000001010{]}~{[}000100001:010010110:000001011{]}

{[}001001101:010010100:000001001{]}~{[}001100101:010010101:000001000{]}

{[}010001011:010010010:000001111{]}~{[}010100011:010010011:000001110{]}

{[}011001111:010010001:000001100{]}~{[}011100111:010010000:000001101{]}

{[}100001000:010011011:000000110{]}~{[}100100000:010011010:000000111{]}

{[}101001100:010011000:000000101{]}~{[}101100100:010011001:000000100{]}

{[}110001010:010011110:000000011{]}~{[}110100010:010011111:000000010{]}

{[}111001110:010011101:000000000{]}~{[}111100110:010011100:000000001{]}

sym~seq=010110001:

{[}000001011:010111100:000001101{]}~{[}000100011:010111101:000001100{]}

{[}001001111:010111111:000001110{]}~{[}001100111:010111110:000001111{]}

{[}010001001:010111001:000001000{]}~{[}010100001:010111000:000001001{]}

{[}011001101:010111010:000001011{]}~{[}011100101:010111011:000001010{]}

{[}100001010:010110000:000000001{]}~{[}100100010:010110001:000000000{]}

{[}101001110:010110011:000000010{]}~{[}101100110:010110010:000000011{]}

{[}110001000:010110101:000000100{]}~{[}110100000:010110100:000000101{]}

{[}111001100:010110110:000000111{]}~{[}111100100:010110111:000000110{]}

sym~seq=\uline{101110100}:

{[}000011000:101110101:000000001{]}~{[}000110000:101110100:000000000{]}

{[}001011100:101110110:000000010{]}~{[}001110100:101110111:000000011{]}

{[}010011010:101110000:000000100{]}~{[}010110010:101110001:000000101{]}

{[}011011110:101110011:000000111{]}~{[}011110110:101110010:000000110{]}

{[}100011001:101111001:000001101{]}~{[}100110001:101111000:000001100{]}

{[}101011101:101111010:000001110{]}~{[}101110101:101111011:000001111{]}

{[}110011011:101111100:000001000{]}~{[}110110011:101111101:000001001{]}

{[}111011111:101111111:000001011{]}~{[}111110111:101111110:000001010{]}

sym~seq=000111010:

{[}000000011:000110110:000001100{]}~{[}000101011:000110111:000001101{]}

{[}001000111:000110101:000001111{]}~{[}001101111:000110100:000001110{]}

{[}010000001:000110011:000001001{]}~{[}010101001:000110010:000001000{]}

{[}011000101:000110000:000001010{]}~{[}011101101:000110001:000001011{]}

{[}100000010:000111010:000000000{]}~{[}100101010:000111011:000000001{]}

{[}101000110:000111001:000000011{]}~{[}101101110:000111000:000000010{]}

{[}110000000:000111111:000000101{]}~{[}110101000:000111110:000000100{]}

{[}111000100:000111100:000000110{]}~{[}111101100:000111101:000000111{]}

sym~seq=\uline{010001011}:

{[}000001000:010001010:000000001{]}~{[}000100000:010001011:000000000{]}

{[}001001100:010001001:000000010{]}~{[}001100100:010001000:000000011{]}

{[}010001010:010001111:000000100{]}~{[}010100010:010001110:000000101{]}

{[}011001110:010001100:000000111{]}~{[}011100110:010001101:000000110{]}

{[}100001001:010000110:000001101{]}~{[}100100001:010000111:000001100{]}

{[}101001101:010000101:000001110{]}~{[}101100101:010000100:000001111{]}

{[}110001011:010000011:000001000{]}~{[}110100011:010000010:000001001{]}

{[}111001111:010000000:000001011{]}~{[}111100111:010000001:000001010{]}

sym~seq=011000101:

{[}000001110:011000001:000000100{]}~{[}000100110:011000000:000000101{]}

{[}001001010:011000010:000000111{]}~{[}001100010:011000011:000000110{]}

{[}010001100:011000100:000000001{]}~{[}010100100:011000101:000000000{]}

{[}011001000:011000111:000000010{]}~{[}011100000:011000110:000000011{]}

{[}100001111:011001101:000001000{]}~{[}100100111:011001100:000001001{]}

{[}101001011:011001110:000001011{]}~{[}101100011:011001111:000001010{]}

{[}110001101:011001000:000001101{]}~{[}110100101:011001001:000001100{]}

{[}111001001:011001011:000001110{]}~{[}111100001:011001010:000001111{]}

sym~seq=101100010:

{[}000011001:101101000:000001010{]}~{[}000110001:101101001:000001011{]}

{[}001011101:101101011:000001001{]}~{[}001110101:101101010:000001000{]}

{[}010011011:101101101:000001111{]}~{[}010110011:101101100:000001110{]}

{[}011011111:101101110:000001100{]}~{[}011110111:101101111:000001101{]}

{[}100011000:101100100:000000110{]}~{[}100110000:101100101:000000111{]}

{[}101011100:101100111:000000101{]}~{[}101110100:101100110:000000100{]}

{[}110011010:101100001:000000011{]}~{[}110110010:101100000:000000010{]}

{[}111011110:101100010:000000000{]}~{[}111110110:101100011:000000001{]}

sym~seq=001001110:

{[}000000110:001001011:000000101{]}~{[}000101110:001001010:000000100{]}

{[}001000010:001001000:000000110{]}~{[}001101010:001001001:000000111{]}

{[}010000100:001001110:000000000{]}~{[}010101100:001001111:000000001{]}

{[}011000000:001001101:000000011{]}~{[}011101000:001001100:000000010{]}

{[}100000111:001000111:000001001{]}~{[}100101111:001000110:000001000{]}

{[}101000011:001000100:000001010{]}~{[}101101011:001000101:000001011{]}

{[}110000101:001000010:000001100{]}~{[}110101101:001000011:000001101{]}

{[}111000001:001000001:000001111{]}~{[}111101001:001000000:000001110{]}

sym~seq=110110001:

{[}000010110:110110100:000000101{]}~{[}000111110:110110101:000000100{]}

{[}001010010:110110111:000000110{]}~{[}001111010:110110110:000000111{]}

{[}010010100:110110001:000000000{]}~{[}010111100:110110000:000000001{]}

{[}011010000:110110010:000000011{]}~{[}011111000:110110011:000000010{]}

{[}100010111:110111000:000001001{]}~{[}100111111:110111001:000001000{]}

{[}101010011:110111011:000001010{]}~{[}101111011:110111010:000001011{]}

{[}110010101:110111101:000001100{]}~{[}110111101:110111100:000001101{]}

{[}111010001:110111110:000001111{]}~{[}111111001:110111111:000001110{]}

sym~seq=111010011:

{[}000010010:111010100:000000111{]}~{[}000111010:111010101:000000110{]}

{[}001010110:111010111:000000100{]}~{[}001111110:111010110:000000101{]}

{[}010010000:111010001:000000010{]}~{[}010111000:111010000:000000011{]}

{[}011010100:111010010:000000001{]}~{[}011111100:111010011:000000000{]}

{[}100010011:111011000:000001011{]}~{[}100111011:111011001:000001010{]}

{[}101010111:111011011:000001000{]}~{[}101111111:111011010:000001001{]}

{[}110010001:111011101:000001110{]}~{[}110111001:111011100:000001111{]}

{[}111010101:111011110:000001101{]}~{[}111111101:111011111:000001100{]}

sym~seq=101001110:

{[}000011011:101000011:000001101{]}~{[}000110011:101000010:000001100{]}

{[}001011111:101000000:000001110{]}~{[}001110111:101000001:000001111{]}

{[}010011001:101000110:000001000{]}~{[}010110001:101000111:000001001{]}

{[}011011101:101000101:000001011{]}~{[}011110101:101000100:000001010{]}

{[}100011010:101001111:000000001{]}~{[}100110010:101001110:000000000{]}

{[}101011110:101001100:000000010{]}~{[}101110110:101001101:000000011{]}

{[}110011000:101001010:000000100{]}~{[}110110000:101001011:000000101{]}

{[}111011100:101001001:000000111{]}~{[}111110100:101001000:000000110{]}

sym~seq=001011000:

{[}000000111:001010110:000001110{]}~{[}000101111:001010111:000001111{]}

{[}001000011:001010101:000001101{]}~{[}001101011:001010100:000001100{]}

{[}010000101:001010011:000001011{]}~{[}010101101:001010010:000001010{]}

{[}011000001:001010000:000001000{]}~{[}011101001:001010001:000001001{]}

{[}100000110:001011010:000000010{]}~{[}100101110:001011011:000000011{]}

{[}101000010:001011001:000000001{]}~{[}101101010:001011000:000000000{]}

{[}110000100:001011111:000000111{]}~{[}110101100:001011110:000000110{]}

{[}111000000:001011100:000000100{]}~{[}111101000:001011101:000000101{]}

sym~seq=100111010:

{[}000011110:100111110:000000100{]}~{[}000110110:100111111:000000101{]}

{[}001011010:100111101:000000111{]}~{[}001110010:100111100:000000110{]}

{[}010011100:100111011:000000001{]}~{[}010110100:100111010:000000000{]}

{[}011011000:100111000:000000010{]}~{[}011110000:100111001:000000011{]}

{[}100011111:100110010:000001000{]}~{[}100110111:100110011:000001001{]}

{[}101011011:100110001:000001011{]}~{[}101110011:100110000:000001010{]}

{[}110011101:100110111:000001101{]}~{[}110110101:100110110:000001100{]}

{[}111011001:100110100:000001110{]}~{[}111110001:100110101:000001111{]}

sym~seq=001100010:

{[}000000100:001100000:000000010{]}~{[}000101100:001100001:000000011{]}

{[}001000000:001100011:000000001{]}~{[}001101000:001100010:000000000{]}

{[}010000110:001100101:000000111{]}~{[}010101110:001100100:000000110{]}

{[}011000010:001100110:000000100{]}~{[}011101010:001100111:000000101{]}

{[}100000101:001101100:000001110{]}~{[}100101101:001101101:000001111{]}

{[}101000001:001101111:000001101{]}~{[}101101001:001101110:000001100{]}

{[}110000111:001101001:000001011{]}~{[}110101111:001101000:000001010{]}

{[}111000011:001101010:000001000{]}~{[}111101011:001101011:000001001{]}

sym~seq=011101001:

{[}000001100:011101010:000000011{]}~{[}000100100:011101011:000000010{]}

{[}001001000:011101001:000000000{]}~{[}001100000:011101000:000000001{]}

{[}010001110:011101111:000000110{]}~{[}010100110:011101110:000000111{]}

{[}011001010:011101100:000000101{]}~{[}011100010:011101101:000000100{]}

{[}100001101:011100110:000001111{]}~{[}100100101:011100111:000001110{]}

{[}101001001:011100101:000001100{]}~{[}101100001:011100100:000001101{]}

{[}110001111:011100011:000001010{]}~{[}110100111:011100010:000001011{]}

{[}111001011:011100000:000001001{]}~{[}111100011:011100001:000001000{]}

sym~seq=110100111:

{[}000010111:110101001:000001110{]}~{[}000111111:110101000:000001111{]}

{[}001010011:110101010:000001101{]}~{[}001111011:110101011:000001100{]}

{[}010010101:110101100:000001011{]}~{[}010111101:110101101:000001010{]}

{[}011010001:110101111:000001000{]}~{[}011111001:110101110:000001001{]}

{[}100010110:110100101:000000010{]}~{[}100111110:110100100:000000011{]}

{[}101010010:110100110:000000001{]}~{[}101111010:110100111:000000000{]}

{[}110010100:110100000:000000111{]}~{[}110111100:110100001:000000110{]}

{[}111010000:110100011:000000100{]}~{[}111111000:110100010:000000101{]}

sym~seq=011010011:

{[}000001111:011011100:000001111{]}~{[}000100111:011011101:000001110{]}

{[}001001011:011011111:000001100{]}~{[}001100011:011011110:000001101{]}

{[}010001101:011011001:000001010{]}~{[}010100101:011011000:000001011{]}

{[}011001001:011011010:000001001{]}~{[}011100001:011011011:000001000{]}

{[}100001110:011010000:000000011{]}~{[}100100110:011010001:000000010{]}

{[}101001010:011010011:000000000{]}~{[}101100010:011010010:000000001{]}

{[}110001100:011010101:000000110{]}~{[}110100100:011010100:000000111{]}

{[}111001000:011010110:000000101{]}~{[}111100000:011010111:000000100{]}

sym~seq=001110100:

{[}000000101:001111101:000001001{]}~{[}000101101:001111100:000001000{]}

{[}001000001:001111110:000001010{]}~{[}001101001:001111111:000001011{]}

{[}010000111:001111000:000001100{]}~{[}010101111:001111001:000001101{]}

{[}011000011:001111011:000001111{]}~{[}011101011:001111010:000001110{]}

{[}100000100:001110001:000000101{]}~{[}100101100:001110000:000000100{]}

{[}101000000:001110010:000000110{]}~{[}101101000:001110011:000000111{]}

{[}110000110:001110100:000000000{]}~{[}110101110:001110101:000000001{]}

{[}111000010:001110111:000000011{]}~{[}111101010:001110110:000000010{]}

sym~seq=111101001:

{[}000010001:111100010:000001011{]}~{[}000111001:111100011:000001010{]}

{[}001010101:111100001:000001000{]}~{[}001111101:111100000:000001001{]}

{[}010010011:111100111:000001110{]}~{[}010111011:111100110:000001111{]}

{[}011010111:111100100:000001101{]}~{[}011111111:111100101:000001100{]}

{[}100010000:111101110:000000111{]}~{[}100111000:111101111:000000110{]}

{[}101010100:111101101:000000100{]}~{[}101111100:111101100:000000101{]}

{[}110010010:111101011:000000010{]}~{[}110111010:111101010:000000011{]}

{[}111010110:111101000:000000001{]}~{[}111111110:111101001:000000000{]}

sym~seq=011111111:

{[}000001101:011110111:000001000{]}~{[}000100101:011110110:000001001{]}

{[}001001001:011110100:000001011{]}~{[}001100001:011110101:000001010{]}

{[}010001111:011110010:000001101{]}~{[}010100111:011110011:000001100{]}

{[}011001011:011110001:000001110{]}~{[}011100011:011110000:000001111{]}

{[}100001100:011111011:000000100{]}~{[}100100100:011111010:000000101{]}

{[}101001000:011111000:000000111{]}~{[}101100000:011111001:000000110{]}

{[}110001110:011111110:000000001{]}~{[}110100110:011111111:000000000{]}

{[}111001010:011111101:000000010{]}~{[}111100010:011111100:000000011{]}

sym~seq=111111111:

{[}000010000:111111111:000000000{]}~{[}000111000:111111110:000000001{]}

{[}001010100:111111100:000000011{]}~{[}001111100:111111101:000000010{]}

{[}010010010:111111010:000000101{]}~{[}010111010:111111011:000000100{]}

{[}011010110:111111001:000000110{]}~{[}011111110:111111000:000000111{]}

{[}100010001:111110011:000001100{]}~{[}100111001:111110010:000001101{]}

{[}101010101:111110000:000001111{]}~{[}101111101:111110001:000001110{]}

{[}110010011:111110110:000001001{]}~{[}110111011:111110111:000001000{]}

{[}111010111:111110101:000001010{]}~{[}111111111:111110100:000001011{]}
\end{codesmall}

\subsection{\label{apn:NumPerSeq}Number of Periods and Unique Sequences of 9-cell
CAs}

\begin{table}[t]
\begin{centering}
\resizebox*{\textwidth}{.9\textheight}{%
\begin{tabular}{|c|c|c|c|c|c|c|c|c|}
\hline 
Ruleset & Period & Sequences & Ruleset & Period & Sequences & Ruleset & Period & Sequences\tabularnewline
\hline 
\hline 
110111101 & 30 & 511 & 110001100 & 14 & 511 & 110001110 & 10 & 256\tabularnewline
\hline 
110100101 & 30 & 511 & 110010111 & 14 & 511 & {*}110000011 & 10 & 32\tabularnewline
\hline 
110110000 & 30 & 511 & 110011000 & 14 & 511 & {*}110010011 & 10 & 32\tabularnewline
\hline 
111000011 & 30 & 511 & 110100111 & 14 & 256 & 110101000 & 8 & 256\tabularnewline
\hline 
111001000 & 30 & 511 & 111001011 & 14 & 256 & 111100011 & 8 & 256\tabularnewline
\hline 
111011110 & 30 & 511 & 111111101 & 14 & 256 & 111110001 & 8 & 256\tabularnewline
\hline 
111101010 & 30 & 511 & 110101100 & 14 & 128 & 110001111 & 8 & 256\tabularnewline
\hline 
111101110 & 30 & 511 & 110010100 & 14 & 128 & 110010000 & 8 & 256\tabularnewline
\hline 
111110000 & 30 & 511 & 110011010 & 14 & 128 & 111111000 & 8 & 128\tabularnewline
\hline 
110000110 & 30 & 511 & 111110010 & 14 & 64 & 110000100 & 8 & 128\tabularnewline
\hline 
110000111 & 30 & 511 & {*}111101111 & 14 & 32 & 110001010 & 8 & 128\tabularnewline
\hline 
110001001 & 30 & 511 & 111001110 & 12 & 511 & {*}111000111 & 8 & 32\tabularnewline
\hline 
111000100 & 30 & 256 & 111001111 & 12 & 511 & {*}111111111 & 7 & 32\tabularnewline
\hline 
111010010 & 30 & 256 & 111011100 & 12 & 511 & 110010010 & 6 & 511\tabularnewline
\hline 
111010101 & 30 & 256 & 111100111 & 12 & 511 & 111000101 & 6 & 256\tabularnewline
\hline 
111011101 & 30 & 256 & 111111001 & 12 & 511 & 111010001 & 6 & 256\tabularnewline
\hline 
111100001 & 30 & 256 & 110001000 & 12 & 511 & 111101000 & 6 & 128\tabularnewline
\hline 
110010001 & 30 & 256 & 110011101 & 12 & 511 & 110100011 & 6 & 64\tabularnewline
\hline 
{*}110111011 & 30 & 32 & 110110111 & 12 & 256 & 111100010 & 6 & 64\tabularnewline
\hline 
{*}110101011 & 30 & 32 & 110111001 & 12 & 256 & 110001011 & 6 & 64\tabularnewline
\hline 
110100000 & 28 & 511 & 110111110 & 12 & 256 & 110100100 & 6 & 32\tabularnewline
\hline 
111001100 & 28 & 511 & 110100001 & 12 & 256 & 110101001 & 4 & 256\tabularnewline
\hline 
110000010 & 28 & 511 & 111000010 & 12 & 256 & 111001101 & 4 & 256\tabularnewline
\hline 
110011001 & 28 & 511 & 111011011 & 12 & 256 & 111010100 & 4 & 256\tabularnewline
\hline 
110011100 & 28 & 511 & 111101101 & 12 & 256 & 111011001 & 4 & 256\tabularnewline
\hline 
110110010 & 24 & 511 & 111110011 & 12 & 256 & 111101100 & 4 & 256\tabularnewline
\hline 
110010110 & 24 & 511 & 111111100 & 12 & 256 & 110011110 & 4 & 256\tabularnewline
\hline 
110100110 & 24 & 64 & 110000000 & 12 & 256 & 110111100 & 4 & 128\tabularnewline
\hline 
110110100 & 24 & 32 & 110011111 & 12 & 256 & 110110011 & 4 & 64\tabularnewline
\hline 
110101111 & 16 & 511 & 111010000 & 12 & 128 & 111001010 & 4 & 64\tabularnewline
\hline 
111101011 & 16 & 511 & 111110110 & 12 & 32 & 110011011 & 4 & 64\tabularnewline
\hline 
111110101 & 16 & 511 & 110000101 & 12 & 32 & 111100110 & 4 & 32\tabularnewline
\hline 
111111010 & 16 & 511 & 110100010 & 10 & 511 & 110010101 & 4 & 32\tabularnewline
\hline 
{*}111010111 & 16 & 32 & 110101010 & 10 & 511 & 110110110 & 3 & 64\tabularnewline
\hline 
110111111 & 14 & 511 & 111000110 & 10 & 511 & 110000001 & 1 & 256\tabularnewline
\hline 
111010011 & 14 & 511 & 111011000 & 10 & 511 & 111000000 & 1 & 128\tabularnewline
\hline 
111011111 & 14 & 511 & 111100000 & 10 & 511 & 110110101 & 0 & 511\tabularnewline
\hline 
111100101 & 14 & 511 & 111100100 & 10 & 511 & 110111010 & 0 & 511\tabularnewline
\hline 
111101001 & 14 & 511 & 110001101 & 10 & 511 & 110101101 & 0 & 511\tabularnewline
\hline 
111110100 & 14 & 511 & 110110001 & 10 & 256 & 111010110 & 0 & 511\tabularnewline
\hline 
111110111 & 14 & 511 & 110111000 & 10 & 256 & 110101110 & 0 & 256\tabularnewline
\hline 
111111011 & 14 & 511 & 111000001 & 10 & 256 & 111011010 & 0 & 64\tabularnewline
\hline 
111111110 & 14 & 511 & 111001001 & 10 & 256 & \multicolumn{1}{c}{} & \multicolumn{1}{c}{} & \multicolumn{1}{c}{}\tabularnewline
\cline{1-6} 
\end{tabular}}
\par\end{centering}

\textsf{\caption{\textsf{\label{tab:RawPerSeq}Raw data for number of periods and sequences
of various 9-cell CA. Rulesets have a 1 for rule 150 cells, 0 for
rule 90 cells. Rulesets marked with '{*}' are symmetric.}}
}
\end{table}

\end{document}